\documentclass[11pt,british,refpage,intoc,bibliography=totoc,index=totoc,BCOR=7.5mm,captions=tableheading]{extarticle}
\usepackage{mathptmx}
\usepackage[T1]{fontenc}
\usepackage[latin9]{inputenc}
\usepackage[a4paper]{geometry}
\usepackage{color}
\usepackage{babel}
\usepackage{algorithm}
\usepackage{algpseudocode}
\usepackage{amsmath}
\usepackage{amsthm}
\usepackage{amssymb}
\usepackage{graphicx}
\usepackage{subfig}
\usepackage[numbers]{natbib}
\usepackage[unicode=true,pdfusetitle,
 bookmarks=true,bookmarksnumbered=true,bookmarksopen=false,
 breaklinks=false,pdfborder={0 0 0},pdfborderstyle={},backref=false,colorlinks=true]
 {hyperref}
\hypersetup{
 linkcolor=black, citecolor=blue, urlcolor=black, filecolor=blue, pdfpagelayout=OneColumn, pdfnewwindow=true, pdfstartview=XYZ, plainpages=false}

\makeatletter
\numberwithin{equation}{section}
\numberwithin{figure}{section}
\theoremstyle{plain}
\newtheorem{thm}{\protect\theoremname}[section]
\theoremstyle{plain}
\newtheorem{lem}[thm]{\protect\lemmaname}

\@ifundefined{date}{}{\date{}}
\usepackage{caption}
\usepackage[nottoc]{tocbibind}
\allowdisplaybreaks[4]
\usepackage{dsfont}
\DeclareMathAlphabet{\mathcal}{OMS}{cmsy}{m}{n}

\makeatother

\providecommand{\lemmaname}{Lemma}
\providecommand{\theoremname}{Theorem}

\begin{document}
\global\long\def\Sgm{\boldsymbol{\Sigma}}%

\global\long\def\W{\boldsymbol{W}}%

\global\long\def\H{\boldsymbol{H}}%

\global\long\def\P{\mathbb{P}}%

\global\long\def\Q{\mathbb{Q}}%

\title{Monte-Carlo method for incompressible fluid\\ flows past obstacles}
\author{By V. Cherepanov\thanks{Mathematical Institute, University of Oxford, Oxford OX2 6GG. Email:
\protect\href{mailto:vladislav.cherepanov@maths.ox.ac.uk}{vladislav.cherepanov@maths.ox.ac.uk}} \ \ and \ \   Z. Qian\thanks{Mathematical Institute, University of Oxford, Oxford OX2 6GG and Oxford Suzhou Centre for Advanced Research, Suzhou, China. Email:
\protect\href{mailto:qianz@maths.ox.ac.uk}{qianz@maths.ox.ac.uk}}}
\maketitle
\begin{abstract}
We establish stochastic functional integral representations for incompressible fluid flows occupying wall-bounded domains using the conditional law duality for a class of diffusion processes. These representations are used to derive a Monte-Carlo scheme based on the corresponding exact random vortex formulation. We implement several numerical experiments based on the Monte-Carlo method 
without appealing to the boundary layer flow computations, to demonstrate the methodology.

\medskip


\medskip

\end{abstract}

\section{Introduction}

In this paper we aim to develop Monte-Carlo schemes for the numerical analysis of incompressible flows past solid obstacles, that is flows occupying certain domains with a nontrivial boundary. These fluid flows satisfy the incompressible Navier-Stokes equations, and we derive exact functional integral representations for their solutions. Phrasing these in terms of random vortex formulations, we are able to establish corresponding numerical schemes and conduct simulations for such flows.

The flows in question are characterised by their velocity $u(x,t)$ for $x \in D$, where $D$ is a domain in $\mathbb{R}^2$ or $\mathbb{R}^3$ with nontrivial boundary $\partial D$. The velocity then follows the Navier-Stokes equations:
\begin{equation}\label{NSgeneral}
\frac{\partial}{\partial t}u + (u\cdot\nabla) u-\nu\Delta u+\nabla P=F,
\end{equation}
\begin{equation}
\nabla\cdot u = 0
\end{equation}
in $D$, and the no-slip condition for $u$ is satisfied, i.e. $u(x,t)=0$ for $x\in\partial D$. The constant $\nu > 0$ is the kinematic viscosity of the fluid, $P$ is the pressure and $F$ denotes an external force. 

One can get another important formulation of the problem, writing it in terms of different dynamical variables, namely the velocity $u$ and the vorticity $\omega=\nabla\wedge u$, the curl of the velocity. For simplicity, we assume in the following derivation that the domain $D$ is two-dimensional --- we refer the reader to Section \ref{RV_for_wall-bounded_flows} for the three-dimensional case. We also assume that the velocity $u$ is extended to the whole space $\mathbb{R}^2$ such that $\nabla\cdot u = 0$ in distribution  (e.g., by letting $u=0$ outside of the domain $D$, see Sections \ref{RV_for_wall-bounded_flows} and \ref{LimitingRepresentationsSection} for details).

Indeed, the equation \eqref{NSgeneral} implies the vorticity transport equation
\begin{equation}\label{VTEgeneral}
\frac{\partial}{\partial t}\omega+(u\cdot\nabla)\omega-\nu\Delta\omega=G
\end{equation}
in $D$, where $G=\nabla\wedge F$. Notice, however, that in general the value of the vorticity $\omega$ is non-zero along the boundary $\partial D$, and we denote $\theta = \left. \omega \right|_{\partial D}$. It is well-known that the solution to equations of the form \eqref{VTEgeneral} can be expressed in terms of functional integrals --- this idea may trace back to Feynman \cite{Feynman1948} and Kac \cite{Kac1949}. However, we first have to introduce a perturbation $W_{\varepsilon}$ of the vorticity $\omega$ such that $W_{\varepsilon}=\omega$ everywhere except a thin $\varepsilon$-layer of the boundary $\partial D$, and $\left. W_{\varepsilon} \right|_{\partial D} = 0$. One way of defining such perturbation is given in \cite{QQZW2022} using a smooth cutoff function, which we also utilise in the subsequent sections.

Therefore, for $W_{\varepsilon}$ the vorticity equation \eqref{VTEgeneral} implies
\begin{equation}
\frac{\partial}{\partial t}W_{\varepsilon}+(u\cdot\nabla)W_{\varepsilon}-\nu\Delta W_{\varepsilon}=g_{\varepsilon}
\end{equation}
in $D$ where $\left. W_{\varepsilon} \right|_{\partial D} = 0$ and $g_{\varepsilon}$ is the corresponding perturbation of $G$. Introducing the Green function $p(s,x,t,y)$ for the forward parabolic operator $\nu\Delta-u\cdot\nabla-\frac{\partial}{\partial t}$ in $D$ subject to the Dirichlet boundary condition, we write
\begin{equation}\label{PerturbedVorticityRepresentation}
W_{\varepsilon}(y,t)=\int_{D}p(0,\xi,t,y)W_{\varepsilon}(\xi,0)\textrm{d}\xi+\int_{0}^{t}\int_{D}p(s,\xi,t,y)g_{\varepsilon}(\xi,s)\textrm{d}\xi\textrm{d}s
\end{equation}
which we can then use to recover the unperturbed vorticity $\omega$. Note that for incompressible flows one then recovers the velocity by
\begin{equation}\label{BiotSavartLawIntroduction}
    u(x,t)=\int_{D}K(x,y)\omega(y,t)\textrm{d}y
\end{equation}
where $K$ is the Biot-Savart kernel of the domain $D$ (assuming both the velocity and the vorticity decay to zero at infinity sufficiently fast). 

To give a formulatic expression, we assume throughout the introduction that a representation similar to \eqref{PerturbedVorticityRepresentation} holds for the unperturbed vorticity $\omega$, that is
\begin{equation}\label{HomogeneousVorticityWithDensity}
\omega(y,t)=\int_{D}p(0,\xi,t,y) \omega(\xi,0)\textrm{d}\xi +\int_{0}^{t}\int_{D}p(s,\xi,t,y)g(\xi,s)\textrm{d}\xi\textrm{d}s.
\end{equation}
Then for the velocity, one has 
\begin{equation}\label{HomogeneousVelocityWithDensity}
 u(x,t)=\int_{D}\int_{D} K(x,y) p(0,\xi,t,y) \omega(\xi,0)\textrm{d}\xi\textrm{d}y + \int_{0}^{t}\int_{D}\int_{D}K(x,y)p(s,\xi,t,y)g(\xi,s)\textrm{d}\xi\textrm{d}s\textrm{d}y.
\end{equation}
We stress again that even though the above formulae hold for the case of unbounded flows, in the case of nontrivial boundary the actual representations essentially involve the boundary vorticity $\theta$, see Sections \ref{RV_for_wall-bounded_flows} and \ref{LimitingRepresentationsSection} for exact expressions.

It turns out one can formulate the above as stochastic representations due to the following observation. Noting again that the vector field $u$ is divergence-free, we have that the function $p(s,x,t,y)$ is also the transition density for the diffusion process with infinitesimal generator $\nu\Delta+u\cdot\nabla$. This means that for the family of stochastic processes $X^{\xi, s}$ given by 
\begin{equation}\label{TaylorsDiffusionsIntroduction}
\textrm{d} X_{t}^{\xi,s}=u(X^{\xi,s}_t,t)\textrm{d}t+\sqrt{2\nu}\textrm{d}B_{t}
\end{equation}
for $t\geq s$ and $X^{\xi,s}_t=\xi$ for $t\leq s$, we have that 
\begin{equation}
p(s,\xi,t,y)\textrm{d}y=\mathbb{P}\left[X^{\xi,s}_t\in\textrm{d}y\right].
\end{equation}
The processes are called Taylor's diffusions (due to Taylor \cite{Taylor1921}), one interprets the processes as "imaginary" Brownian particles following the flow with velocity $u(x,t)$.

 Let us view the integral representations we have so far in terms of expectations with respect to the family of Taylor's diffusions introduced above. Indeed, it is easy to see that the integrals in the representation \eqref{HomogeneousVelocityWithDensity} can be viewed in terms of the expectations with respect to $X^{\xi, s}$, that is 
\begin{equation}\label{HomogeneousVelocityWithDiffusions}
 u(x,t)=\int_{D} \mathbb{E}\left[K(x,X^{\xi,0}_t)\right] \omega(\xi,0)\textrm{d}\xi + \int_{0}^{t}\int_{D}\mathbb{E}\left[K(x,X^{\xi,s}_t)\right]g(\xi,s)\textrm{d}\xi\textrm{d}s.
\end{equation}

Having discussed the above representations, we propose the following numerical scheme. We couple the processes $X^{\xi,s}$ with the velocity $u(x,t)$ and update them 
according to approximations given by 
\eqref{TaylorsDiffusionsIntroduction} and 
\eqref{HomogeneousVelocityWithDiffusions} respectively. Indeed, the former depends only on the velocity $u(x,t)$ for which one can take the approximation given by the latter. Note that, in turn, the velocity is determined by the diffusions $X^{\xi,s}$ and the presence of expectations in the representation 
\eqref{HomogeneousVelocityWithDiffusions} allows for Monte-Carlo-type schemes (we refer the reader to the Section 
\ref{MCSimulationsSection} for particular schemes we used in simulations).

One of the drawbacks for the scheme described above is that it involves Taylor's diffusions 
\eqref{TaylorsDiffusionsIntroduction} started at every time $s \in [0, t]$. This implies that the memory required for the computation and the computation complexity itself are increasing with time $t$. In Section \ref{RV_for_wall-bounded_flows}, we overcome this difficulty by deriving representations different from 
\eqref{HomogeneousVorticityWithDensity} and 
\eqref{HomogeneousVelocityWithDensity} using the duality of the conditional laws of Taylor's diffusions proved in \cite{QSZ3D}. Although the representations in this case involve more complicated integrands (see Theorems \ref{thm7.2}, \ref{thm:5.2} and their two-dimensional versions, Theorems \ref{thm5.2new}, \ref{thm5.4new}), we are able to devise the schemes with non-increasing memory and computational complexity in Section \ref{MCSimulationsSection}.

We also have to mention again that the above representations \eqref{HomogeneousVorticityWithDensity} and \eqref{HomogeneousVelocityWithDensity} in the stated form hold for domains without boundary. In the case of domains with nontrivial boundary, the representations are derived from \eqref{PerturbedVorticityRepresentation} for the perturbed vorticity $W_{\varepsilon}$. These depend on the boundary vorticity $\theta$ which poses an additional difficulty in simulations as we are unable to compute the values of $\theta$ using the expressions for $\omega$. However, as in practice we use a desingularised kernel $K_{\delta}$ to compute the velocity, we write the formal derivative of the representation for $u(x,t)$ to compute the vorticity $\omega$ used in plots and the boundary vorticity $\theta$. 

Note also that as the expression 
\eqref{PerturbedVorticityRepresentation} for the perturbed vorticity $W_{\varepsilon}$ depends on the boundary layer thickness parameter $\varepsilon$, it is interesting to consider the limit of the representations derived from 
\eqref{PerturbedVorticityRepresentation} as $\varepsilon \to 0$. We do the corresponding computations in Section \ref{LimitingRepresentationsSection} for different domains and notice that many terms in fact do not contribute to the limit. This observation turns out to be useful in Section \ref{MCSimulationsSection} as it simplifies the numerical schemes significantly. 

Let us remark that the discussed above limiting representations obtained in Section \ref{LimitingRepresentationsSection} have a certain interest on their own as they display how the boundary vorticity $\theta$ influences the velocity $u$ (see Section \ref{LimitingRepresentationsSection} for details), which seems to be in accordance with Prandtl's boundary layer theory \cite{Prandtl1904}. In some of the experiments reported in Section \ref{MCSimulationsSection} we omit the terms involving $\theta$ when we compute the velocity $u$ and notice that they essentially contribute to the chaotic behaviour of the flow close to the boundary. 

We also note that we derive these representations in Section \ref{LimitingRepresentationsSection} for two particular domains (namely, the ones for flows passing a flat plate and a wedge) and for general two-dimensional domains $D$ conformally equivalent to the half-plane with smooth boundary $\partial D$. However, we believe that similar representations should hold for more complicated domains as well (e.g., not necessary simply connected domains with piecewise smooth boundary).

For the existing literature on the subject, we first mention that there has been extensive literature studying Navier-Stokes equations, in particular, numerical analysis of their solutions. For the general numerical methods, we mention Computational Fluid Dynamics (CFD) (see \cite{Fletcher1991}, \cite{Wesseling2001}), while more particular numerical approaches include Direct Numerical Simulations (DNS) (see \cite{OrszagPatterson1972}, \cite{Spalart1988}, \cite{MoinMashesh1998} and also \cite{Wesseling2001}), Large Eddy Simulations (LES) (see \cite{Deardorff1970}, \cite{Lilly1967}, \cite{LesieurLDS}), Probability Density Function (PDF) (see \cite{Pope2000}). There is also a large volume of literature with numerical study of turbulence, e.g., see the works \cite{ChauhanPhilipetl2014}, \cite{Heisel2018}, \cite{DawsonMcKeon2019}, \cite{HeadBandyopadhyay1981}, \cite{RaiMoin1993}, \cite{SpalartWatmuff1993}, \cite{WuMoin2008, WuMoin2009, WuMoinHickey2014}, as well as those that study other various aspects of boundary flows \cite{BalakumarAdrian2007}, \cite{DeGraaffEaton2000}, \cite{ErmJoubert1991}, \cite{ErmSmitsJoubert1985}, \cite{HonkanAndreopoulos1997}, \cite{Keller1978}, \cite{Schlichting9th-2017}, \cite{WuJacobsHuntDurbin1999}.

Random vortex method is based originally on the work \cite{Chorin 1973} by Chorin, though the idea of using Brownian fluid particles dates back to Taylor \cite{Taylor1921}. It is also successfully used in studying different aspects of turbulent flows, see \cite{Pope2000}, \cite{Majda and Bertozzi 2002} and \cite{Falkovich2001}. We also mention that some convergence results are known for random vortex methods, see \cite{AndersonGreengard1985}, \cite{Goodman1987} and \cite{Long1988}, however, these concern the flows without boundary and the convergence results for the schemes described above will be studied in a future work. For other probabilistic aspects of fluid dynamics we refer the reader to \cite{Constantin2001a, Constantin2001b} and \cite{ConstantinIyer2011}, and we also mention particularly stochastic Lagrangian approach that is used to study isotropic turbulence, see \cite{Drivas2017a, Drivas2017b}, \cite{EyinkGuptaZaki2020a, EyinkGuptaZaki2020b}.

This paper is based on the previous works \citep{Qian-Stochastic2022}, \cite{QSZ3D} and \citep{QQZW2022} and generalises the results for wall-bounded domains presented in \citep{Qian-Stochastic2022} and \citep{QQZW2022}. In the recent paper \cite{LQX2023}, the ideas similar to those presented here are used to study Oberbeck-Boussinesq flows which are outside of the scope of the current work.

The paper is organised as follows. In Section \ref{RV_for_wall-bounded_flows}, we derive a random vortex representation for general two- and three-dimensional flows occupying the half-space. This representation is derived using the duality of conditional laws of the Taylor diffusions and thus depends only on diffusions started at time $0$. In Section \ref{BS_Section}, we review the Biot-Savart law and state auxiliary results concerning the Biot-Savart kernel for the domains in question. In Section \ref{LimitingRepresentationsSection}, we consider representations as those described above and compute their limits as the thickness of the thin layers converges to zero. In Section \ref{MCSimulationsSection}, we use the above results to derive numerical schemes and provide the experiment results.

\section{Random vortex for wall-bounded flows}\label{RV_for_wall-bounded_flows}

In this section, we follow the approach proposed in \citep{Qian-Stochastic2022} and \citep{QQZW2022}. In the latter paper the random
vortex dynamics has been established for viscous fluid flows moving
along with a solid wall modelled by the two dimensional Navier-Stokes
equations. Let $u=(u^{1},u^{2},u^{3})$ denote the velocity of an
incompressible viscous fluid flow constrained in the upper half space
$D$ where $x_{3}\geq0$. The velocity $u$ satisfies the no-slip 
condition, i.e. $u(x,t)=0$ for $x\in \partial D$ and $t>0$. Let us assume 
that there is an external force $F=(F^{1},F^{2},F^{3})$ supplying the energy
to the fluid dynamical system. Therefore the velocity $u$ and the pressure $P$
are evolved according to the Navier-Stokes equations:
\begin{equation}
\frac{\partial u^{i}}{\partial t}+\sum_{j=1}^{3}u^{j}\frac{\partial u^{i}}{\partial x_{j}}-\nu\Delta u^{i}+\frac{\partial P}{\partial x_{i}}-F^{i}=0\quad\textrm{ in }D\label{3D-Ns01-1}
\end{equation}
for $i=1,2,3$, and
\begin{equation}
\sum_{j=1}^{3}\frac{\partial u^{j}}{\partial x_{j}}=0\quad\textrm{ in }D.\label{3D-Ns02-1}
\end{equation}
The initial velocity is denoted by $u_{0}(x)=u(x,0)$. In the following we will often omit the summation over repeated indices. 

To proceed our discussion, we assume that the velocity $u$ is
smooth inside $D$ and is $C^3$ up to the boundary $\partial D$. 
This assumption is although technical, but in no means
it is trivial and apparent. Indeed it remains and will still be stand
as an open problem in mathematics for establishing the regularity
for solutions to the general three dimensional Navier-Stokes equations.
We do not pursue this line of research which lies outside the scope
of current project. 

We introduce the vorticity $\omega=\nabla\wedge u$,
whose components $\omega^{i}=\varepsilon^{ijk}\frac{\partial}{\partial x^{j}}u^{k}$. 
In our study the vorticity transport equations play a crucial role.
The vorticity transport equations are evolution equations for $\omega$, which are given as the
following
\begin{equation}
\frac{\partial\omega^{i}}{\partial t}+u^{j}\frac{\partial\omega^{i}}{\partial x_{j}}-\nu\Delta\omega^{i}+\omega^{j}S_{j}^{i}-G^{i}=0\quad\textrm{ in }D \label{3D-V-01}
\end{equation}
where $S_{j}^{i}=\frac{1}{2}(\frac{\partial u^{i}}{\partial x_{j}}+\frac{\partial u^{j}}{\partial x_{i}})$
is the symmetric tensor of rate-of-strain, and $G=\nabla\wedge F$. 

We adopt the simple idea in \citep{QQZW2022} of extending the
definition of the velocity $u=(u^{1},u^{2},u^{3})$ to the whole space
$\mathbb{R}^{3}$ by the reflection principle, that is 
\[
u^{i}(\bar{x},t)=u^{i}(x,t)\quad\textrm{ for }i=1,2;\quad u^{3}(\bar{x},t)=-u^{3}(x,t)
\]
for $x_{3}<0$. Since $u$ satisfies the no-slip condition, this
yields a crucial fact that the extended velocity $u$ is divergence
free in the distribution sense on $\mathbb{R}^{3}$. Therefore the adjoint
of the heat operator $\nu\Delta-u\cdot\nabla-\frac{\partial}{\partial t}$
is just the heat operator $\nu\Delta+u\cdot\nabla+\frac{\partial}{\partial t}$.
Therefore the Green function $\Gamma_{D}(\tau,x,t,y)$ to the Dirichlet
boundary problem of the (forward) parabolic equation
\[
\begin{cases}
\left(\frac{\partial}{\partial t}-\nu\Delta+u\cdot\nabla\right)f=0 & \textrm{ in }D,\\
f(x,t)=0 & \textrm{ if }x_{3}=0
\end{cases}
\]
coincides with the transition probability density of the diffusion
process with its infinitesimal generator $\nu\Delta+u\cdot\nabla$
killed on leaving the domain $D$. This leads to the following construction.
Let $X$ be the diffusion process with infinitesimal generator $\nu\Delta+u\cdot\nabla$,
which is a diffusion with state space $\mathbb{R}^{3}$, without ``killing'' at the boundary $x_3=0$. That
is, $X$ is a weak solution of the stochastic differential equation
\begin{equation}
\textrm{d}X=u(X,t)\textrm{d}t+\sqrt{2\nu}\textrm{d}B.\label{taylor-01}
\end{equation}
Let $p(\tau,x,t,y)$ (for $t>\tau\geq0$, and $x,y\in\mathbb{R}^{3}$)
be the transition probability density function, which is positive
and H\"older's continuous in all arguments, as long as $u$ is bounded
and Borel measurable. Formally 
\[
p(\tau,x,t,y)\textrm{d}y=\mathbb{P}\left[X_{t}\in\textrm{d}y|X_{\tau}=x\right]
\]
and therefore it is clear that
\[
p(\tau,x,t,y)=p(\tau,\overline{x},t,\overline{y})\quad\textrm{ for }x,y\in\mathbb{R}^{3}.
\]
Hence, by applying the reflection principle we have therefore the
following representation
\begin{equation}
\Gamma_{D}(\tau,x,t,y)=p(\tau,x,t,y)-p(\tau,x,t,\overline{y})\quad\textrm{ for }x,y\in D.\label{rep-green}
\end{equation}

We next apply this representation to the study of the vorticity $\omega$. The
vorticity $\omega$ may be considered as a solution of the linear
parabolic equation \eqref{3D-V-01} if $u$ is supposed as a given
fluid dynamic variable. However, $\omega$ has nontrivial boundary
value in general, so let $\theta$ be the trace of $\omega$ along the
boundary $\partial D$ (note that the vorticity boundary values may be identified with the stress of the fluid flow immediately injected to the wall, cf. \cite{Schlichting9th-2017} \cite{WeinanLiu1996}). Then it is easy to see that $\omega^{3}$
has trace zero, and due to the no-slip condition imposed on $u$,
$\theta$ can be identified with the normal part of the stress applied
to the boundary $\partial D$. Note that $\theta$ is a time dependent
vector field on the boundary $\partial D$. 

To handle the no vanishing boundary vorticity $\theta$, we employ the same technique used in \citep{QQZW2022}. Let $\phi$ be a smooth cut-off function defined on $[0,\infty)$ with values in $[0,1]$, such that $\phi(r)=1$ for $r\in[0,1/3]$ and $\phi(r)=0$ for $r\geq2/3$, and define $\sigma_{\varepsilon}(x,t)=\phi(x_{3}/\varepsilon)\theta(x_{1},x_{2},t)$ for $x=(x_{1},x_{2},x_{3})\in D$. Let 
\begin{equation}
W_{\varepsilon}(x,t)=\omega(x,t)-\sigma_{\varepsilon}(x,t).
\end{equation}
Then 
\begin{equation}
\frac{\partial W_{\varepsilon}^{i}}{\partial t}+u^{j}\frac{\partial W_{\varepsilon}^{i}}{\partial x_{j}}-\nu\Delta W_{\varepsilon}^{i}-W_{\varepsilon}^{j}\frac{\partial u^{i}}{\partial x_{j}}-g_{\varepsilon}^{i}=0\quad\textrm{ in }D,\label{3D-V-01-1}
\end{equation}
for $i=1,2,3$, and
\begin{equation}
\left.W_{\varepsilon}(x,t)\right|_{x\in\partial D}=0\textrm{ }\quad\textrm{ for }x\in\partial D,\label{vor-B-01}
\end{equation}
where 
\begin{equation}
g_{\varepsilon}^{i}=G^{i}-\frac{\partial}{\partial t}\sigma_{\varepsilon}^{i}-u^{j}\frac{\partial}{\partial x_{j}}\sigma_{\varepsilon}^{i}+\nu\Delta\sigma_{\varepsilon}^{i}-\sigma_{\varepsilon}^{j}S_{j}^{i}\label{e-g-01}
\end{equation}
for $i=1,2,3$.

\subsection{Random vortex for wall-bounded flows}

According to the Feynman-Kac formula for forward heat equations established in \citep{Qian-Stochastic2022} we can represent $W_{\varepsilon}$ in terms of the distribution of the Taylor diffusion. Here we present a slightly different approach. To this end, we have to introduce some notation. Let $\mathbb{P}^{\xi}$ denote the distribution of the diffusion process with infinitesimal generator $\nu\Delta+u\cdot\nabla$ (where $u$ is extended on $\mathbb{R}^{3}$ by the reflection principle) started from $\xi\in\mathbb{R}^{3}$ at time $0$, and $\mathbb{P}_{t}^{\xi\rightarrow\eta}$ the conditional law of $\mathbb{P}^{\xi}\left[\cdot|\psi(t)=\eta\right]$ where $\psi$ is the coordinate process (i.e. canonical element) on the path space $C([0,\infty),\mathbb{R}^{3})$. 

We recall that $p_{b}(s,x,t,y)$ denotes the transition probability density function of the diffusion with generator $\nu\Delta+b\cdot\nabla$.

\begin{thm}
\label{thm7.2} Let $X^{\eta}$ be the Taylor diffusion:
\begin{equation}
\textrm{d}X_{t}^{\eta}=u(X_{t}^{\eta},t)\textrm{d}t+\sqrt{2\nu}\textrm{d}B_{t},\quad X_{0}^{\eta}=\eta\label{3D-T1}
\end{equation}
for every $\eta\in\mathbb{R}^{3}$. For each pair $\eta\in\mathbb{R}^{3}$
and $t>0$, define $s\mapsto Q(\eta,t;s)$ to be the unique solution
of the differential equations:
\begin{equation}
\frac{\textrm{d}}{\textrm{d}s}Q_{j}^{i}(\eta,t;s)=-Q_{k}^{i}(\eta,t;s)1_{D}(X_{s}^{\eta})A_{j}^{k}(X_{s}^{\eta},s),\quad Q_{j}^{i}(\eta,t;t)=\delta_{ij}\label{3DQ-1}
\end{equation}
where $A_{j}^{k}=\frac{\partial}{\partial x_{j}}u^{k}$ and $i,j,k=1,2,3$. Then 
\begin{align}
\omega^{i}(\xi,t) & =\sigma_{\varepsilon}^{i}(\xi,t)+\int_{D}\mathbb{P}^{\eta\rightarrow\xi}\left[Q_{j}^{i}(\eta,t;0)1_{\{t<\zeta(X^{\eta}\circ\tau_{t})\}}\right]W_{\varepsilon}^{j}(\eta,0)p_{u}(0,\eta,t,\xi)\textrm{d}\eta\nonumber \\
 & +\int_{0}^{t}\int_{D}\mathbb{P}^{\eta\rightarrow\xi}\left[Q_{j}^{i}(\eta,t;s)1_{\{t-s<\zeta(X^{\eta}\circ\tau_{t})\}}g_{\varepsilon}^{j}(X_{s}^{\eta},s)\right]p_{u}(0,\eta,t,\xi)\textrm{d}\eta\textrm{d}s\label{3D-om}
\end{align}
for every $\xi\in D$ and $t>0$, $i=1,2,3$, where $\zeta(\psi)=\inf\left\{ s:\psi(s)\notin D\right\}$ and $\tau_{t}$ denotes the time reversion operator on the path space $C([0,t];\mathbb{R}^{3})$.
\end{thm}

\begin{proof}
For simplicity, let $q_{j}^{i}(x,t)=1_{D}(x)A_{j}^{i}(x,t)$, where $i,j=1,2,3$.
Let $T>0$ be fixed. Recall that $u(x,t)$ is extended for all $x\in\mathbb{R}^{3}$
such that $\nabla\cdot u=0$ in distribution on $\mathbb{R}^{3}$.
Let $\tilde{X}^{\xi}$ be the solution to the stochastic differential
equation
\begin{equation}
\textrm{d}\tilde{X}_{t}^{\xi}=-u(\tilde{X}_{t}^{\xi},T-t)\textrm{d}t+\sqrt{2\nu}\textrm{d}B_{t},\quad\tilde{X}_{0}^{\xi}=\xi\label{ba-s01}
\end{equation}
which is understood as the stochastic integral equation 
\begin{equation}
\tilde{X}_{t}^{\xi}=\xi-\int_{0}^{t\wedge T}u(\tilde{X}_{s}^{\xi},T-s)\textrm{d}s+\sqrt{2\nu}\int_{0}^{t\wedge T}\textrm{d}B_{s}\quad\textrm{ for all }t\geq0.\label{sint-01}
\end{equation}
 Let 
\begin{equation*}
T_{\xi}=\inf\left\{ t\geq0:\tilde{X}_{t}^{\xi}\notin D\right\} 
\end{equation*}
be the first time the process $\tilde{X}^{\xi}$ leaves the region
$D$. Then $t\rightarrow\tilde{X}_{t\wedge T_{\xi}}^{\xi}$ is a diffusion
process too, and 
\begin{align}
\tilde{X}_{t\wedge T_{\xi}}^{\xi} & =\xi-\int_{0}^{t\wedge T\wedge T_{\xi}}u(\tilde{X}_{s}^{\xi},T-s)\textrm{d}s+\sqrt{2\nu}\int_{0}^{t\wedge T\wedge T_{\xi}}\textrm{d}B_{s}\nonumber \\
 & =\xi-\int_{0}^{t}1_{\left\{ s<T\wedge T_{\xi}\right\} }u(\tilde{X}_{s}^{\xi},T-s)\textrm{d}s+\sqrt{2\nu}\int_{0}^{t}1_{\left\{ s<T\wedge T_{\xi}\right\} }\textrm{d}B_{s}\label{kidiff-01}
\end{align}
for all $t\geq0$. According to Feynman-Kac, we define
\begin{equation}
\textrm{d}\tilde{Q}_{j}^{i}(s)=\tilde{Q}_{k}^{i}(s)q_{j}^{k}(\tilde{X}_{s}^{\xi},T-s)\textrm{d}s,\quad\tilde{Q}_{j}^{i}(0)=\delta_{ij}\label{Qij-01}
\end{equation}
where $i,j=1,2,3$. Let $Y_{t}^{j}=W_{\varepsilon}^{j}(\tilde{X}_{t\wedge T_{\xi}}^{\xi},T-t)$
which is well defined, as according to our assumption, $W_{\varepsilon}$
is $C^{2}$ on $\overline{D}$. Since $W_{\varepsilon}$ vanishes
along the boundary $\partial D$, so that
\begin{align*}
Y_{t}^{j} & =W_{\varepsilon}^{j}(\tilde{X}_{t\wedge T_{\xi}}^{\xi},T-t)=1_{\{t<T_{\xi}\}}W_{\varepsilon}^{j}(\tilde{X}_{t}^{\xi},T-t)+1_{\{t\geq T_{\xi}\}}W_{\varepsilon}^{j}(\tilde{X}_{T_{\xi}}^{\xi},T-t)\\
 & =1_{\{t<T_{\xi}\}}W_{\varepsilon}^{j}(\tilde{X}_{t}^{\xi},T-t)
\end{align*}
for all $t\leq T$. Moreover, according to It\^o's formula, 
\[
\textrm{d}Y_{t}^{j}=\nabla W_{\varepsilon}^{j}(\tilde{X}_{t\wedge T_{\xi}}^{\xi},T-t)\cdot\textrm{d}\tilde{X}_{t\wedge T_{\xi}}^{\xi}-\frac{\partial W_{\varepsilon}^{j}}{\partial t}(\tilde{X}_{t\wedge T_{\xi}}^{\xi},T-t)\textrm{d}t+\nu\Delta W_{\varepsilon}^{j}(\tilde{X}_{t\wedge T_{\xi}}^{\xi},T-t)\textrm{d}t.
\]
Since $W_{\varepsilon}$ vanishes identically on $\partial D$,
\begin{align*}
\frac{\partial W_{\varepsilon}^{j}}{\partial t}(\tilde{X}_{t\wedge T_{\xi}}^{\xi},T-t) & =1_{\left\{ t<T_{\xi}\right\} }\frac{\partial W_{\varepsilon}^{j}}{\partial t}(\tilde{X}_{t}^{\xi},T-t)+1_{\left\{ t\geq T_{\xi}\right\} }\frac{\partial W_{\varepsilon}^{j}}{\partial t}(\tilde{X}_{T_{\xi}}^{\xi},T-t)\\
 & =1_{\left\{ t<T_{\xi}\right\} }\frac{\partial W_{\varepsilon}^{j}}{\partial t}(\tilde{X}_{t}^{\xi},T-t)
\end{align*}
and similarly
\begin{align*}
 \nabla W_{\varepsilon}^{j}(\tilde{X}_{t\wedge T_{\xi}}^{\xi},T-t) &= 1_{\left\{ t<T_{\xi}\right\}} \nabla W_{\varepsilon}^{j}(\tilde{X}_{t}^{\xi},T-t),\\
 \Delta W_{\varepsilon}^{j}(\tilde{X}_{t\wedge T_{\xi}}^{\xi},T-t) &= 1_{\left\{ t<T_{\xi}\right\} } \Delta W_{\varepsilon}^{j}(\tilde{X}_{t}^{\xi},T-t)
\end{align*}
for all $t\leq T$. Therefore 
\begin{align}
\textrm{d}Y_{t}^{j} & =\sqrt{2\nu}1_{\left\{ t<T_{\xi}\right\} }\nabla W_{\varepsilon}^{j}(\tilde{X}_{t}^{\xi},T-t)\cdot\textrm{d}B_{t}\nonumber \\
 & +1_{\left\{ t<T_{\xi}\right\} }\left(\nu\Delta W_{\varepsilon}^{j}-(u\cdot\nabla W_{\varepsilon}^{j})-\frac{\partial W_{\varepsilon}^{j}}{\partial t}\right)(\tilde{X}_{t}^{\xi},T-t)\textrm{d}t\label{Yj-06}
\end{align}
for all $t\leq T$. By \eqref{3D-V-01-1} which may be written as
the following
\begin{equation}
\left(\nu\Delta-u\cdot\nabla-\frac{\partial}{\partial t}\right)W_{\varepsilon}^{j}=-A_{k}^{j}W_{\varepsilon}^{k}-g_{\varepsilon}^{j}\quad\textrm{ in }D\label{3D-V-04}
\end{equation}
hence, by substituting this into \eqref{Yj-06}, it follows that
\begin{equation}
\textrm{d}Y_{t}^{j}=\sqrt{2\nu}1_{\{t<T_{\xi}\}}\nabla W_{\varepsilon}^{j}(\tilde{X}_{t}^{\xi},T-t)\cdot\textrm{d}B_{t}-1_{\{t<T_{\xi}\}}\left(A_{k}^{j}W_{\varepsilon}^{k}+g_{\varepsilon}^{j}\right)(\tilde{X}_{t}^{\xi},T-t)\textrm{d}t.\label{Yj-001}
\end{equation}
Let $M^{i}=\tilde{Q}_{j}^{i}Y^{j}$. Using stochastic integration
by parts one obtains 
\begin{align*}
M_{t}^{i} & =\tilde{Q}_{j}^{i}(0)Y_{0}^{j}+\int_{0}^{t}\tilde{Q}_{j}^{i}(s)\textrm{d}Y_{s}^{j}+\int_{0}^{t}Y_{s}^{j}\textrm{d}\tilde{Q}_{j}^{i}(s)\\
 & =\tilde{Q}_{j}^{i}(0)Y_{0}^{j}+\sqrt{2\nu}\int_{0}^{t}\tilde{Q}_{j}^{i}(s)1_{\{s<T_{\xi}\}}\nabla W_{\varepsilon}^{j}(\tilde{X}_{s}^{\xi},T-s)\cdot\textrm{d}B_{s}\\
 & -\int_{0}^{t}\tilde{Q}_{j}^{i}(s)1_{\{s<T_{\xi}\}}\left(A_{k}^{j}W_{\varepsilon}^{k}+g_{\varepsilon}^{j}\right)(\tilde{X}_{s}^{\xi},T-s)\textrm{d}s+\int_{0}^{t}W_{\varepsilon}^{j}(\tilde{X}_{s\wedge T_{\xi}}^{\xi},T-s)\textrm{d}\tilde{Q}_{j}^{i}(s)\\
 & =\tilde{Q}_{j}^{i}(0)Y_{0}^{j}+\sqrt{2\nu}\int_{0}^{t}\tilde{Q}_{j}^{i}(s)1_{\{s<T_{\xi}\}}\nabla W_{\varepsilon}^{j}(\tilde{X}_{s}^{\xi},T-s)\cdot\textrm{d}B_{s}\\
 & -\int_{0}^{t}1_{\{s<T_{\xi}\}}\tilde{Q}_{j}^{i}(s)g_{\varepsilon}^{j}(\tilde{X}_{s}^{\xi},T-s)\textrm{d}s\\
 & +\int_{0}^{t}1_{\{s<T_{\xi}\}}W_{\varepsilon}^{j}(\tilde{X}_{s}^{\xi},T-s)\left(-\tilde{Q}_{k}^{i}(s)q_{j}^{k}(\tilde{X}_{s}^{\xi},T-s)\textrm{d}s+\textrm{d}\tilde{Q}_{j}^{i}(s)\right)
\end{align*}
here the last equality follows from the fact that on $\left\{ s<T_{\xi}\right\}$, the process
$\tilde{X}_{s}^{\xi}$ takes values in $D$, so that $A_{j}^{i}$
coincides with $q_{j}^{i}$. This computation holds for any differentiable,
adapted processes $\tilde{Q}_{j}^{i}$. In particular, if $\tilde{Q}$
solves \eqref{Qij-01}, then 
\begin{align*}
M_{t}^{i} & =Y_{0}^{i}+\sqrt{2\nu}\int_{0}^{t}1_{\{s<T_{\xi}\}}\tilde{Q}_{j}^{i}(s)\nabla W_{\varepsilon}^{j}(\tilde{X}_{s}^{\xi},T-s)\cdot\textrm{d}B_{s}\\
 & -\int_{0}^{t}1_{\{s<T_{\xi}\}}\tilde{Q}_{j}^{i}(s)g_{\varepsilon}^{j}(\tilde{X}_{s}^{\xi},T-s)\textrm{d}s
\end{align*}
which yields that
\begin{align}
M_{T\wedge T_{\xi}}^{i} & =W_{\varepsilon}^{i}(\xi,T)-\int_{0}^{T\wedge T_{\xi}}1_{\{t<T_{\xi}\}}\tilde{Q}_{j}^{i}(t)g_{\varepsilon}^{j}(\tilde{X}_{t}^{\xi},T-t)\textrm{d}t\nonumber \\
 & +\sqrt{2\nu}\int_{0}^{T\wedge T_{\xi}}1_{\{t<T_{\xi}\}}\tilde{Q}_{j}^{i}(t)\nabla W_{\varepsilon}^{j}(\tilde{X}_{t}^{\xi},T-t)\cdot\textrm{d}B.\label{ba-0051}
\end{align}
Since
\[
M_{T\wedge T_{\xi}}^{i}=1_{\{T<T_{\xi}\}} \tilde{Q}_{j}^{i}(T) W_{\varepsilon}^{j}(\tilde{X}_{T}^{\xi},0),\quad Y_{0}^{i}=W_{\varepsilon}^{i}(\xi,T)
\]
so that, after taking expectation of the both sides of \eqref{ba-0051} to
obtain that
\begin{equation}
W_{\varepsilon}^{i}(\xi,T)=\mathbb{E}\left[1_{\{T<T_{\xi}\}} \tilde{Q}_{j}^{i}(T) W_{\varepsilon}^{j}(\tilde{X}_{T}^{\xi},0)\right]+\mathbb{E}\left[\int_{0}^{T}1_{\{t<T_{\xi}\}}\tilde{Q}_{j}^{i}(t)g_{\varepsilon}^{j}(\tilde{X}_{t}^{\xi},T-t)\textrm{d}t\right].\label{Wi-004}
\end{equation}
By taking conditional expectation on the event that $\tilde{X}_{T}^{\xi}=\eta$,
we obtain that
\begin{align*}
\mathbb{E}\left[1_{\{T<T_{\xi}\}} \tilde{Q}_{j}^{i}(T) W_{\varepsilon}^{j}(\tilde{X}_{T}^{\xi},0)\right] & =\int_{D}\mathbb{E}\left[\left.\tilde{Q}_{j}^{i}(T)1_{\{T<T_{\xi}\}}\right|\tilde{X}_{T}^{\xi}=\eta\right]W_{\varepsilon}^{j}(\eta,0)p_{-u_{T}}(0,\xi,T,\eta)\textrm{d}\eta\\
 & =\int_{D}\mathbb{E}\left[\left.\tilde{Q}_{j}^{i}(T)1_{\{T<T_{\xi}\}}\right|\tilde{X}_{T}^{\xi}=\eta\right]W_{\varepsilon}^{j}(\eta,0)p_{u}(0,\eta,T,\xi)\textrm{d}\eta
\end{align*}
and
\begin{align*}
\mathbb{E}\left[\int_{0}^{T\wedge T_{\xi}}\tilde{Q}_{j}^{i}(t)g_{\varepsilon}^{j}(\tilde{X}_{t}^{\xi},T-t)\textrm{d}t\right] & =\int_{0}^{T}\int_{D}\mathbb{E}\left[\left.\tilde{Q}_{j}^{i}(t)1_{\{t<T_{\xi}\}}g_{\varepsilon}^{j}(\tilde{X}_{t}^{\xi},T-t)\right|\tilde{X}_{T}^{\xi}=\eta\right]p_{-u_{T}}(0,\xi,T,\eta)\textrm{d}\eta\textrm{d}t\\
 & =\int_{0}^{T}\int_{D}\mathbb{E}\left[\left.\tilde{Q}_{j}^{i}(t)1_{\{t<T_{\xi}\}}g_{\varepsilon}^{j}(\tilde{X}_{t}^{\xi},T-t)\right|\tilde{X}_{T}^{\xi}=\eta\right]p_{u}(0,\eta,T,\xi)\textrm{d}\eta\textrm{d}t
\end{align*}
where $p_{-u_{T}}$ is the transition probability density function
of the diffusion $\tilde{X}^{\xi}$, which is a diffusion process
with infinitesimal generator $\nu\Delta-u_{T}\cdot\nabla$, where
$u_{T}(x,t)=u(x,T-t)$. Since $\nabla\cdot u=0$, we have that $p_{-u_{T}}(0,\xi,T,\eta)$
coincides with $p_{u}(0,\eta,T,\xi)$. Therefore
\begin{align}
W_{\varepsilon}^{i}(\xi,T) & =\int_{D}\mathbb{E}\left[\left.\tilde{Q}_{j}^{i}(T)1_{\{T<T_{\xi}\}}\right|\tilde{X}_{T}^{\xi}=\eta\right]W_{\varepsilon}^{j}(\eta,0)p_{u}(0,\eta,T,\xi)\textrm{d}\eta\nonumber \\
 & +\int_{0}^{T}\int_{D}\mathbb{E}\left[\left.\tilde{Q}_{j}^{i}(t)1_{\{t<T_{\xi}\}}g_{\varepsilon}^{j}(\tilde{X}_{t}^{\xi},T-t)\right|\tilde{X}_{T}^{\xi}=\eta\right]p_{u}(0,\eta,T,\xi)\textrm{d}\eta\textrm{d}t.\label{WW-002}
\end{align}

Let us rewrite the representation \eqref{WW-002} in terms of the
distributions of the Taylor diffusion. To this end, let us use $\tilde{\mathbb{P}}^{\xi}$
to denote the law of $\tilde{X}^{\xi}$ and $\tilde{\mathbb{P}}^{\xi\rightarrow\eta}$
to denote the conditional law, or the diffusion bridge measure of $\tilde{X}^{\xi}$
given the terminal value that $\tilde{X}_{T}^{\xi}=\eta$. The conditional
law $\mathbb{\tilde{P}}^{\xi\rightarrow\eta}$ can be considered as
a probability measure on the path space $C([0,T];\mathbb{R}^{3})$
canonically. For any continuous path $\psi:[0,T]\rightarrow\mathbb{R}^{3}$,
$\tilde{Q}(\psi;t)$ denotes the solution to the following linear
ordinary differential equations
\begin{equation}
\frac{\textrm{d}}{\textrm{d}t}\tilde{Q}_{j}^{i}(\psi;t)=\tilde{Q}_{k}^{i}(\psi;t)q_{j}^{k}(\psi(t),T-t),\quad\tilde{Q}_{j}^{i}(\psi;0)=\delta_{j}^{i}\label{t-Q-01}
\end{equation}
for $i,j=1,2,3$. Then $\tilde{Q}_{j}^{i}(\tilde{X}^{\xi};t)$ gives
rise a version of the gauge functional defined above $\tilde{Q}_{j}^{i}(t)$.
Therefore, under the notations we just set up, the representation
\eqref{WW-002} may be written as
\begin{align}
W_{\varepsilon}^{i}(\xi,T) & =\int_{D}\mathbb{\tilde{P}}^{\xi\rightarrow\eta}\left[\tilde{Q}_{j}^{i}(\psi;T)1_{\{T<\zeta(\psi)\}}\right]W_{\varepsilon}^{j}(\eta,0)p_{u}(0,\eta,T,\xi)\textrm{d}\eta\nonumber \\
 & +\int_{0}^{T}\int_{D}\mathbb{\tilde{P}}^{\xi\rightarrow\eta}\left[\tilde{Q}_{j}^{i}(\psi;t)1_{\{t<\zeta(\psi)\}}g_{\varepsilon}^{j}(\psi(t),T-t)\right]p_{u}(0,\eta,T,\xi)\textrm{d}\eta\textrm{d}t\label{Wi-06}
\end{align}
where
\[
\zeta(\psi)=\inf\left\{ t\geq0:\psi(t)\notin D\right\} .
\]
Next we apply the fundamental duality for conditional laws: since
$\nabla\cdot u=0$ on $\mathbb{R}^{3}$ in the distribution sense,
so that
\[
\mathbb{\tilde{P}}^{\xi\rightarrow\eta}=\mathbb{P}^{\eta\rightarrow\xi}\circ\tau_{T},
\]
cf. \citep{QSZ3D}. Let $Q(\psi,T;t)=\tilde{Q}(\psi\circ\tau_{T};T-t)$.
Since
\begin{equation}
\tilde{Q}_{j}^{i}(\psi;s)=\delta_{ij}+\int_{0}^{s}\tilde{Q}_{k}^{i}(\psi;r)q_{j}^{k}(\psi(r),T-r)\textrm{d}r.\label{t-Q-01-1}
\end{equation}
Apply this to $\psi\circ\tau_{T}$ one obtains that
\begin{align*}
\tilde{Q}_{j}^{i}(\psi\circ\tau_{T};s) & =\delta_{ij}+\int_{0}^{s}\tilde{Q}_{k}^{i}(\psi\circ\tau_{T};r)q_{j}^{k}(\psi(T-r),T-r)\textrm{d}r\\
 & =\delta_{ij}+\int_{T-s}^{T}\tilde{Q}_{k}^{i}(\psi\circ\tau_{T};T-r)q_{j}^{k}(\psi(r),r)\textrm{d}r
\end{align*}
so that
\[
\tilde{Q}_{j}^{i}(\psi\circ\tau_{T};T-s)=\delta_{ij}+\int_{s}^{T}\tilde{Q}_{k}^{i}(\psi\circ\tau_{T};T-r)q_{j}^{k}(\psi(r),r)\textrm{d}r.
\]
This means that $t\mapsto Q(\psi,T;t)$ is the unique solution to
the ordinary differential equation:
\[
dQ_{j}^{i}(\psi,T;t)=-Q_{k}^{i}(\psi,T;t)q_{j}^{k}(\psi(t),t)dt,\quad Q_{j}^{i}(\psi,T;T)=\delta_{ij}.
\]
By definition, $\tilde{Q}(\psi;t)=Q(\psi\circ\tau_{T},T;T-t)$ for
every $t\in[0,T].$ Hence, by the duality of the conditional laws,
we may rewrite \eqref{Wi-06} in terms of the law of the Taylor diffusion:
\begin{align*}
W_{\varepsilon}^{i}(\xi,T) & =\int_{D}\mathbb{P}^{\eta\rightarrow\xi}\left[Q_{j}^{i}(\psi,T;0)1_{\{T<\zeta(\psi\circ\tau_{T})\}}\right]W_{\varepsilon}^{j}(\eta,0)p_{u}(0,\eta,T,\xi)\textrm{d}\eta\\
 & +\int_{0}^{T}\int_{D}\mathbb{P}^{\eta\rightarrow\xi}\left[Q_{j}^{i}(\psi,T;T-t)1_{\{t<\zeta(\psi\circ\tau_{T})\}}g_{\varepsilon}^{j}(\psi(T-t),T-t)\right]p_{u}(0,\eta,T,\xi)\textrm{d}\eta\textrm{d}t\\
 & =\int_{D}\mathbb{P}^{\eta\rightarrow\xi}\left[Q_{j}^{i}(\psi,T;0)1_{\{T<\zeta(\psi\circ\tau_{T})\}}\right]W_{\varepsilon}^{j}(\eta,0)p_{u}(0,\eta,T,\xi)\textrm{d}\eta\\
 & +\int_{0}^{T}\int_{D}\mathbb{P}^{\eta\rightarrow\xi}\left[Q_{j}^{i}(\psi,T;t)1_{\{T-t<\zeta(\psi\circ\tau_{T})\}}g_{\varepsilon}^{j}(\psi(t),t)\right]p_{u}(0,\eta,T,\xi)\textrm{d}\eta\textrm{d}t
\end{align*}
which yields the claim.
\end{proof}

We are now in a position to derive an important stochastic representation
formula. 
\begin{thm}
\label{thm:5.2} Let $X^{\eta}$ be the Taylor diffusion:
\begin{equation}
\textrm{d}X_{t}^{\eta}=u(X_{t}^{\eta},t)\textrm{d}t+\sqrt{2\nu}\textrm{d}B_{t},\quad X_{0}^{\eta}=\eta\label{ta-003}
\end{equation}
for each $\eta\in\mathbb{R}^{3}$. The following stochastic representation
holds:
\begin{align*}
u^{i}(x,t) & =\varepsilon^{ijk}\int_{D}K^{j}(x,\eta)\sigma_{\varepsilon}^{k}(\eta,t)\textrm{d}\eta\\
 & +\varepsilon^{ijk}\int_{D}\mathbb{E}\left[Q_{l}^{k}(\eta,t;0)1_{\{t<\zeta(X^{\eta}\circ\tau_{t})\}}K^{j}(x,X_{t}^{\eta})\right]W_{\varepsilon}^{l}(\eta,0)\textrm{d}\eta\\
 & +\varepsilon^{ijk}\int_{0}^{t}\int_{D}\mathbb{E} \left[Q_{l}^{k}(\eta,t;s)1_{\{t-s<\zeta(X^{\eta}\circ\tau_{t})\}}K^{j}(x,X_{t}^{\eta})g_{\varepsilon}^{l}(X_{s}^{\eta},s)\right]\textrm{d}\eta\textrm{d}s
\end{align*}
for any $t>0$ and $x\in D$, where $i=1,2,3$.
\end{thm}

\begin{proof}
Recall the Biot-Savart law
\[
u(x,t)=\int_{D}K(x,\xi)\wedge\omega(\xi,t)\textrm{d}\xi
\]
where $K$ is the singular kernel of the half-space $x_{3}>0$.
By Theorem \ref{thm7.2} 
\begin{align}
\omega^{k}(\xi,t) & =\sigma_{\varepsilon}^{k}(\xi,t)+\int_{D}\mathbb{P}^{\eta\rightarrow\xi}\left[Q_{l}^{k}(\eta,t;0)1_{\{t<\zeta(X^{\eta}\circ\tau_{t})\}}\right]W_{\varepsilon}^{l}(\eta,0)p_{u}(0,\eta,t,\xi)\textrm{d}\eta\nonumber \\
 & +\int_{0}^{t}\int_{D}\mathbb{P}^{\eta\rightarrow\xi}\left[Q_{l}^{k}(\eta,t;s)1_{\{t-s<\zeta(X^{\eta}\circ\tau_{t})\}}g_{\varepsilon}^{l}(X_{s}^{\eta},s)\right]p_{u}(0,\eta,t,\xi)\textrm{d}\eta\textrm{d}s.\label{3D-om-1}
\end{align}
Integrating with $\varepsilon^{ijk}K^{j}(x,\cdot)$ and using Fubini
theorem, we obtain
\begin{align*}
u^{i}(x,t) & =\varepsilon^{ijk}\int_{D}K^{j}(x,\eta)\sigma_{\varepsilon}^{k}(\eta,t)\textrm{d}\eta\\
 & +\varepsilon^{ijk}\int_{D}\mathbb{E}\left[K^{j}(x,X_{t}^{\eta})Q_{l}^{k}(\eta,t;0)1_{\{t<\zeta(X^{\eta}\circ\tau_{t})\}}\right]W_{\varepsilon}^{l}(\eta,0)\textrm{d}\eta\\
 & +\varepsilon^{ijk}\int_{0}^{t}\int_{D}\mathbb{E}\left[Q_{l}^{k}(\eta,t;s)1_{\{t-s<\zeta(X^{\eta}\circ\tau_{t})\}}K^{j}(x,X_{t}^{\eta})g_{\varepsilon}^{l}(X_{s}^{\eta},s)\right]\textrm{d}\eta\textrm{d}s
\end{align*}
which completes the proof. 
\end{proof}
\subsection{Two dimensional wall-bounded flows}

There is a simplified version for two dimensional case.
For two dimensional flows, the non-linear stretching term $A^i_j\omega^j$ vanishes identically, so that we can take $Q^i_j=\delta_{ij}$. Therefore we have the following two dimensional representation formula.
\begin{thm}
\label{thm5.2new}Suppose $d=2$, so that $\omega=\frac{\partial}{\partial x_{1}}u^{2}-\frac{\partial}{\partial x_{2}}u^{1}$
is a scalar function. Then 
\begin{align}
\omega(\xi,t) & =\sigma_{\varepsilon}(\xi,t)+\int_{D}\mathbb{P}^{\eta\rightarrow\xi}\left[t<\zeta(X^{\eta}\circ\tau_{t})\right]W_{\varepsilon}(\eta,0)p_{u}(0,\eta,t,\xi)\textrm{d}\eta\nonumber \\
 & +\int_{0}^{t}\int_{D}\mathbb{E}^{\eta\rightarrow\xi}\left[1_{\{t-s<\zeta(X^{\eta}\circ\tau_{t})\}}g_{\varepsilon}(X_{s}^{\eta},s)\right]p_{u}(0,\eta,t,\xi)\textrm{d}\eta\textrm{d}s\label{2D-om1}
\end{align}
for every $\xi\in D$ and $t>0$, where $\zeta(\psi)=\inf\left\{ s:\psi(s)\notin D\right\} $
and $X^{\eta}$ is the Taylor diffusion:
\begin{equation}
\textrm{d}X_{t}^{\eta}=u(X_{t}^{\eta},t)\textrm{d}t+\sqrt{2\nu}\textrm{d}B_{t},\quad X_{0}^{\eta}=\eta\label{2D-X1}
\end{equation}
for every $\eta\in\mathbb{R}^{2}$. 
\end{thm}
As a consequence, we have the following representation for the velocity $u$ in two dimensional case.
\begin{thm}
\label{thm5.4new}Suppose $d=2$, and let $K=(K^{1},K^{2})$ be the Biot-Savart
kernel for the half-plane $D:x_{2}>0$. Then 
\begin{align}
u^{i}(x,t) & =\int_{D} K^i(x,\eta) \sigma_{\varepsilon}(\eta,t)\textrm{d}\eta+\int_{D}\mathbb{E}\left[K^{i}(x,X_{t}^{\eta})1_{\left\{ t<\zeta(X^{\eta}\circ\tau_{t})\right\} }\right]W_{\varepsilon}(\eta,0)\textrm{d}\eta\nonumber \\
 & +\int_{0}^{t}\int_{D}\mathbb{E}\left[1_{\{t-s<\zeta(X^{\eta}\circ\tau_{t})\}}K^{i}(x,X_{t}^{\eta})g_{\varepsilon}(X_{s}^{\eta},s)\right]\textrm{d}\eta\textrm{d}s\label{2D-om1-1}
\end{align}
for every $x\in D$ and $t>0$, where $\zeta(\psi)=\inf\left\{ s:\psi(s)\notin D\right\}$
and $X^{\eta}$ is the Taylor diffusion:
\begin{equation}
\textrm{d}X_{t}^{\eta}=u(X_{t}^{\eta},t)\textrm{d}t+\sqrt{2\nu}\textrm{d}B_{t},\quad X_{0}^{\eta}=\eta\label{2D-X1-1}
\end{equation}
for every $\eta\in\mathbb{R}^{2}$. 
\end{thm}

\section{2D wall-bounded flows}\label{BS_Section}

In the subsequent parts we aim to present a stochastic formulation of a viscous
fluid flow, with kinematic viscosity constant $\nu>0$, which passes
through obstacles, such as a thin plane plate and a wedge. These flows have
been studied as important and classical examples which demonstrate
boundary layer phenomena. Turbulence may be build up near the solid
obstacle when the Reynolds number becomes large.

For simplicity we take a two dimensional model, although we admit 
that the three dimensional model is more sophisticated and
will be studied in a future work. Therefore the fluid flow in
question is described by its velocity, a time dependent vector field
$u=(u^{1},u^{2})$ in the domain $D$.

\subsection{A general case}

Let $D\subset\mathbb{R}^{2}$ be a general simply connected domain
which is not the whole plane. For such domain, according to Riemann mapping
theorem, there is a one-to-one and onto conformal mapping $T:D\rightarrow\mathbb{H}$
where $\mathbb{H}$ is the upper half-plane (the mapping is unique up to a rotation 
at a point whose image is assigned to $i$ for example). Then the Green function for 
such domain $D$ is given by
\begin{equation}\label{eq-w-002}
G_{D}(x,y)=\Gamma(T(x),T(y))-\Gamma(T(x),\overline{T(y)})    
\end{equation}
for $x,y\in D$, where $\Gamma(x,y)=\frac{1}{2\pi} \log |x - y|$. While the boundary $\partial D$ may be complicated,
so some care is needed, although formally $G$ satisfies the Dirichlet
boundary condition. 

Notice also that $\omega=\nabla\wedge u$ and $\nabla\cdot u=0$ imply 
\begin{equation}\label{eq-uw-12}
\Delta u^{1}=-\frac{\partial}{\partial x_{2}}\omega,\quad\Delta u^{2}=\frac{\partial}{\partial x_{1}}\omega\quad\text{ in }D,
\end{equation}
and, as $u^{i}$ are subject to the no-slip condition, the Green formula implies the Biot-Savart law
\begin{equation}\label{BS-law}
    u^{i}(x,t)=\int_{D}K^{i}(x,y)\omega(y,t)\textrm{d}y,
\end{equation}
where the kernel $K$ is given by
\begin{equation}\label{BSkernel}
    K^1(x,y) = \frac{\partial}{\partial y_2} G_D(x,y),\quad K^2(x,y) = -\frac{\partial}{\partial y_1} G_D(x,y).
\end{equation}
Hence, given the Green function as in \eqref{eq-w-002}, we might state the following lemma.

\begin{lem}\label{lem-e-003}
The kernel $K$ can be computed as follows
\begin{equation}
\begin{split}
    K^1(x,y) &=  k^{-}(x,y) \cdot \frac{\partial T}{\partial y_2}(y) - k^{+}(x,y) \cdot \frac{\partial T}{\partial y_2}(y),\\
    K^2(x,y) &=  -k^{-}(x,y) \cdot \frac{\partial T}{\partial y_1}(y) + k^{+}(x,y) \cdot \frac{\partial T}{\partial y_1}(y),
\end{split}
\end{equation}
where the vector functions $k^{-}(x,y), k^{+}(x,y)$ are given by
\begin{equation}
\begin{split}
    k^{-}(x,y)=\frac{1}{2\pi} \frac{(T^1(y)-T^1(x), T^2(y)-T^2(x))}{(T^1(y)-T^1(x))^2+(T^2(y)-T^2(x))^2},\\
    k^{+}(x,y)=\frac{1}{2\pi} \frac{(T^1(y)-T^1(x), T^2(y)+T^2(x))}{(T^1(y)-T^1(x))^2+(T^2(y)+T^2(x))^2},
\end{split}
\end{equation}
and $\frac{\partial T}{\partial y_1}(y)$ and $\frac{\partial T}{\partial y_2}(y)$ denote the vectors $(\frac{\partial T^1}{\partial y_1}(y), \frac{\partial T^2}{\partial y_1}(y))$ and $(\frac{\partial T^1}{\partial y_2}(y), \frac{\partial T^2}{\partial y_2}(y))$ respectively. 
\end{lem}

As for the numerical simulations we also need to compute the boundary stress $\theta$, we do so by differentiating the representation \eqref{BS-law} with respect to $x$ (and taking limits as $x$ approaches the boundary). Note that in practice we have to desingularise the kernels (see Section \ref{MCSimulationsSection} for details), which insures we can differentiate the singular integral representation. Therefore, we can write the representation, according to the lemma above, in terms of derivatives of $k^{-}, k^{+}$ given explicitly in the following.

\begin{lem}\label{DerivativeK_Lemma}
The derivatives $\dfrac{\partial k^{-}}{\partial x}, \dfrac{\partial k^{+}}{\partial x}$ are given by
\begin{equation}
\begin{split}
    \dfrac{\partial (k^{-})^1}{\partial x_i}(x,y)=\frac{1}{2\pi} \Bigg( \frac{-\frac{\partial T^1}{\partial x_i}(x)}{(T^1(y)-T^1(x))^2+(T^2(y)-T^2(x))^2} + 2(T^1(y)-T^1(x)) \\ 
    \cdot\frac{(T^1(y)-T^1(x))\frac{\partial T^1}{\partial x_i}(x)+(T^2(y)-T^2(x))\frac{\partial T^2}{\partial x_i}(x)}{\big((T^1(y)-T^1(x))^2+(T^2(y)-T^2(x))^2\big)^2} \Bigg),\\
    \dfrac{\partial (k^{-})^2}{\partial x_i}(x,y)=\frac{1}{2\pi} \Bigg( \frac{-\frac{\partial T^2}{\partial x_i}(x)}{(T^1(y)-T^1(x))^2+(T^2(y)-T^2(x))^2} + 2(T^2(y)-T^2(x)) \\ 
    \cdot\frac{(T^1(y)-T^1(x))\frac{\partial T^1}{\partial x_i}(x)+(T^2(y)-T^2(x))\frac{\partial T^2}{\partial x_i}(x)}{\big((T^1(y)-T^1(x))^2+(T^2(y)-T^2(x))^2\big)^2} \Bigg),
\end{split}
\end{equation}
and
\begin{equation}
\begin{split}
    \dfrac{\partial (k^{+})^1}{\partial x_i}(x,y)=\frac{1}{2\pi} \Bigg( \frac{-\frac{\partial T^1}{\partial x_i}(x)}{(T^1(y)-T^1(x))^2+(T^2(y)+T^2(x))^2} + 2(T^1(y)-T^1(x)) \\ 
    \cdot\frac{(T^1(y)-T^1(x))\frac{\partial T^1}{\partial x_i}(x)-(T^2(y)+T^2(x))\frac{\partial T^2}{\partial x_i}(x)}{\big((T^1(y)-T^1(x))^2+(T^2(y)+T^2(x))^2\big)^2} \Bigg),\\
    \dfrac{\partial (k^{+})^2}{\partial x_i}(x,y)=\frac{1}{2\pi} \Bigg( \frac{\frac{\partial T^2}{\partial x_i}(x)}{(T^1(y)-T^1(x))^2+(T^2(y)+T^2(x))^2} + 2(T^2(y)+T^2(x)) \\ 
    \cdot\frac{(T^1(y)-T^1(x))\frac{\partial T^1}{\partial x_i}(x)-(T^2(y)+T^2(x))\frac{\partial T^2}{\partial x_i}(x)}{\big((T^1(y)-T^1(x))^2+(T^2(y)+T^2(x))^2\big)^2} \Bigg),
\end{split}
\end{equation}
where $(k^{\pm})^i$ denote the components of $k^{\pm}$, i.e. $k^{\pm}=((k^{\pm})^1, (k^{\pm})^2)$.
\end{lem}

\subsection{Flows past a thin plate}

We consider a flow in the domain $D=\mathbb{R}^{2}\setminus\left\{(x_{1},0):x_{1}\geq0\right\} $,
where the boundary $\partial D$: $x_{2}=0$ and $x_{1}\geq0$ is
an obstacle in the flow. The fluid flow has to satisfy the no-slip
condition
\begin{equation}
\lim_{x_{2}\uparrow0}u((x_{1},x_{2}),t)=\lim_{x_{2}\downarrow0}u((x_{1},x_{2}),t)=0,\quad\text{when }x_{1}\geq0\label{pl-01}
\end{equation}
for $t>0$. Therefore automatically $u(x,t)$ is extended to the whole
plane by setting $u(x,t)=0$ for $x\in\partial D$. This extension
has the following nice property: if $\nabla\cdot u=0$ in $D$, then
$\nabla\cdot u=0$ in $\mathbb{R}^{2}$ in the sense of distributions. 

To find the Green function we use the conformal principle. The domain $D$ is conformally
equivalent to $\mathbb{H}$, it is realised by the conformal transform
$T^{-1}:x\rightarrow x^{2}$ (consider $x=(x_{1},x_{2})=x_{1}+i x_{2}$
as a complex coordinate) from $\mathbb{H}$ one-to-one and onto $D$.
In terms of real coordinates
\begin{equation*}
    T^{-1}:(x_{1},x_{2})\rightarrow(x_{1}^{2}-x_{2}^{2},2x_{1}x_{2})
\end{equation*}
and its inverse
\begin{equation}\label{Plate_T}
    T:(x_{1},x_{2})\rightarrow\left(\textrm{sgn}(x_{2})\sqrt{\frac{1}{2}\left(|x|+x_{1}\right)},\sqrt{\frac{1}{2}\left(|x|-x_{1}\right)}\right)
\end{equation}
where $|x|=\sqrt{x_{1}^{2}+x_{2}^{2}}$. 

Recall that the Green function for $D$ is given by
\[
G(x,y)=\Gamma(Tx,Ty)-\Gamma(Tx,\overline{Ty}),
\]
which vanishes for $x$ such that
\[
Tx=\overline{Tx},
\]
which is equivalent to that $|x|-x_{1}=0$, i.e. $x_{1}\geq0$ and $x_{2}=0$. To compute the kernel $K$ in this case, we explicitly find the partial derivatives of the map $T$ and use Lemma \ref{lem-e-003}.

\begin{lem}\label{PlateDerivativesOfT}
    Define $T=T^1+i T^2$ by \eqref{Plate_T}. Its partial derivatives are given by
\begin{equation}
\begin{split}    
    \frac{\partial T^1}{\partial x_1} (x) &= \textrm{sgn}(x_{2})\frac{\sqrt{\frac{1}{2}\left(|x|+x_{1}\right)}}{2|x|},\\
    \frac{\partial T^2}{\partial x_1} (x) &= - \frac{\sqrt{\frac{1}{2}\left(|x|-x_{1}\right)}}{2|x|},
\end{split}
\end{equation}
and
\begin{equation}
\begin{split}    
    \frac{\partial T^1}{\partial x_2} (x) &= \textrm{sgn}(x_{2}) \frac{x_2}{4|x| \sqrt{\frac{1}{2}\left(|x|+x_{1}\right)}},\\
    \frac{\partial T^2}{\partial x_2} (x) &= \frac{x_2}{4|x| \sqrt{\frac{1}{2}\left(|x|-x_{1}\right)}}.
\end{split}
\end{equation}
\end{lem}

\subsection{Flows past a wedge obstacle}

Suppose the wedge with angle $2\alpha$ at its tip (where $\alpha\in(0,\frac{\pi}{2})$)
is modelled by 
\[
\Lambda=\left\{ x=r\textrm{e}^{\textrm{i}\theta}:r\geq0\textrm{ and }-\alpha\leq\theta\leq\alpha\right\} 
\]
so that the fluid occupies the region $D=\mathbb{R}^{2}\setminus\Lambda$. It is again conformally equivalent to the upper half space $\mathbb{H}$,
the conformal mapping $T:D\rightarrow\mathbb{H}$ is given explicitly as follows:
\[
T(x)=\left(\textrm{e}^{-\textrm{i}\alpha}x\right)^{\beta}=\exp\left[\beta\ln|x|+\textrm{i}\beta\arg(\textrm{e}^{-\textrm{i}\alpha}x)\right]\quad\textrm{ for }x\in D
\]
where $0<\arg x<2\pi$,
\[
\beta=\frac{\pi}{2(\pi-\alpha)}\quad\textrm{ and }\quad\alpha\beta=\frac{\pi\alpha}{2(\pi-\alpha)}
\]
and
\[
x^{\beta}=|x|^{\beta}\exp\left(\textrm{i}\beta\arg x\right)\quad\textrm{ with }0<\arg x<2\pi.
\]
Therefore
\[
T(x)=\exp\left(\textrm{i}\beta\arg\left(e^{-\textrm{i}\alpha}x\right)\right)|x|^{\beta}\quad\textrm{ for }x\in D.
\]
In terms of real coordinates we have
\[
T(x)=\left(|x|^{\beta}\cos\left(\beta\arg\left(e^{-\textrm{i}\alpha}x\right)\right),|x|^{\beta}\sin\left(\beta\arg\left(e^{-\textrm{i}\alpha}x\right)\right)\right)
\]
for $x=(x_{1},x_{2})\in D$, where $\arg x\in(0,2\pi)$, so that 
\[
\arg(x_{1},x_{2})=\begin{cases}
\arctan\frac{x_{2}}{x_{1}} & \textrm{if }x_{1}>0,x_{2}\geq0,\\
\frac{\pi}{2} & \textrm{if }x_{1}=0,x_{2}>0,\\
\pi+\arctan\frac{x_{2}}{x_{1}} & \textrm{if }x_{1}<0,\\
\frac{3\pi}{2} & \textrm{if }x_{1}=0,x_{2}<0,\\
2\pi+\arctan\frac{x_{2}}{x_{1}} & \textrm{if }x_{1}>0,x_{2}\leq0.
\end{cases}
\]
Since $\alpha\in(0,\frac{\pi}{2})$ and $\arg x \in (\alpha, 2\pi - \alpha)$ for $x \in D$, it follows that
\[
\arg\left(e^{-\textrm{i}\alpha}x\right)=\arg x-\alpha,
\]
hence the real and imaginary parts $T=T^1 +i T^2$ can be written as
\begin{equation}\label{e-sq-004}
    T^1(x) = |x|^{\beta} \cos{(\beta (\arg x -\alpha) )},\quad T^2(x) = |x|^{\beta} \sin{(\beta (\arg x -\alpha))}.
\end{equation}
Therefore the Green function for this domain is given by
\[
G(x,y)=\Gamma(Tx,Ty)-\Gamma(Tx,\overline{Ty})
\]
which vanishes for $x$ such that
\[
Tx=\overline{Tx}
\]
that is
\[
|x|^{\beta}\sin\left(\beta\arg\left(e^{-\textrm{i}\alpha}x\right)\right)=-|x|^{\beta}\sin\left(\beta\arg\left(e^{-\textrm{i}\alpha}x\right)\right)
\]
so that either $x=(0,0)$ or 
\[
\sin\left(\beta\arg(e^{-\textrm{i}\alpha}x)\right)=0.
\]
That is 
\[
\beta\arg(e^{-\textrm{i}\alpha}x)=n\pi
\]
so that
\[
\arg(e^{-\textrm{i}\alpha}x)=2n(\pi-\alpha).
\]
That is
\[
\arg x-\alpha=2n(\pi-\alpha)\quad\textrm{ and }0<\arg x-\alpha<2\pi.
\]
The requirement
\[
0<2n(\pi-\alpha)<2\pi
\]
gives rise to the constraint that $n=0,1$ only that is, $\arg x=\alpha$
and $\arg x=2\pi-\alpha$ which are the edges of the wedge. 

For computing the Biot-Savart kernel $K$ using Lemma \ref{lem-e-003}, we write the partial derivatives of the map $T$.

\begin{lem}\label{WedgeDerivativesOfT}
    Let $T=T^1+i T^2$ be as in \eqref{e-sq-004}, then its partial derivatives are given by
\begin{equation}
\begin{split}    
    \frac{\partial T^1}{\partial x_1} (x) = \beta |x|^{\beta-2} (x_1 \cos{(\beta (\arg x -\alpha) )} +x_2 \sin{(\beta (\arg x -\alpha) )}),\\
    \frac{\partial T^2}{\partial x_1} (x) = \beta |x|^{\beta-2} (x_1 \sin{(\beta (\arg x -\alpha) )} - x_2 \cos{(\beta (\arg x -\alpha) )} ),
\end{split}
\end{equation}
and
\begin{equation}
\begin{split}    
    \frac{\partial T^1}{\partial x_2} (x) = \beta |x|^{\beta-2} (x_2 \cos{(\beta (\arg x -\alpha) )} - x_1 \sin{(\beta (\arg x -\alpha) )} ),\\
    \frac{\partial T^2}{\partial x_2} (x) = \beta |x|^{\beta-2} (x_2 \sin{(\beta (\arg x -\alpha) )} + x_1 \cos{(\beta (\arg x -\alpha) )} ).
\end{split}
\end{equation}
\end{lem}

\section{Limiting stochastic representations}\label{LimitingRepresentationsSection}

The goal of this section is to derive representations similar to those in Section \ref{RV_for_wall-bounded_flows} for different domains described above. Note that so far we have a stochastic representation for the velocity $u$, e.g., the one given in Theorem \ref{thm5.4new}, for the half-plane. This representation contains the time-integral with the term $g_{\varepsilon}$ given by 
\begin{equation}
g_{\varepsilon}=G-\frac{\partial \sigma_{\varepsilon}}{\partial t}-(u\cdot\nabla)\sigma_{\varepsilon}+\nu\Delta\sigma_{\varepsilon},
\end{equation}
where $\varepsilon$ defines the width of the cutoff modification in $\sigma_{\varepsilon}$. It is therefore interesting theoretically to consider the limits of the representations as the cutoff width parameter $\varepsilon$ converges to zero. It turns out that many terms in $g_{\varepsilon}$ do not contribute to the limit, we are able to simplify the expressions used for simulations in Section \ref{MCSimulationsSection}. 

\subsection{Flat-plate case}\label{PlateSection}

Let us deal with the flows passing a flat-plate first, that is $D=\{x=(x_{1},x_{2}):x_{1}<0\textrm{ or }x_{2}\neq0\}$.
The flat-plate is modelled by the boundary $\partial D = \{x_{2}=0,x_{1}\geq0\}$, and the flow is split into upper flow and lower
flow after hitting the plate. The velocity $u$ has to satisfy the
no-slip condition, that is, 
\begin{equation}
\lim_{x_{2}\rightarrow0\pm}u(x_{1},x_{2},t)=0\quad\textrm{ for }x_{1}\geq0\textrm{ and }t>0.\label{noslip-p01}
\end{equation}
Therefore there is non-trivial stress on both sides of the plate.
The stress at the upper side of the plate is denoted by $\theta_{+}$
and the one at the lower side by $\theta_{-}$, that is
\begin{equation}
\theta_{\pm}(x_{1},t)=\begin{cases}
-\left.\frac{\partial}{\partial x_{2}}u^{1}(x_{1},x_{2},t)\right|_{x_{2}=0\pm} & \textrm{ for }x_{1}\geq0,\\
0 & \textrm{ for }x_{1}<0.
\end{cases}\label{ThetaDefPlate}
\end{equation}
We make a technical assumption that $\theta_{+}$ and $\theta_{-}$ have continuous derivatives up to second order except at $x_1=0$. We also assume that the limit of $-\frac{\partial}{\partial x_{2}}u^{1}(x_{1},x_{2},t)$ exists as $(x_1, x_2) \to (0,0)$ in $D$, therefore $\theta_{+}(0,t)=\theta_{-}(0,t)$, for all $t \geq 0$, denoted $\theta (t)$ in this case.
Let $\phi:\mathbb{R}\rightarrow[0,1]$ be a cut-off function which
is smooth on $(0,\infty)$ such that $\phi(r)=0$ for $r<0$, $\phi(r)=1$
for $0 \leq r < 1/3$ and $\phi(r)=0$ for $r\geq2/3$. Define
\begin{equation}
\sigma_{\varepsilon}(x_{1},x_{2},t)=\theta_{+}(x_{1},t)\phi(x_{2}/\varepsilon)+\theta_{-}(x_{1},t)\phi(-x_{2}/\varepsilon)+\theta(t)\phi(|x|/\varepsilon)1_{\{x_1<0\}}\label{ext-01}
\end{equation}
for $x=(x_{1},x_{2})\in\mathbb{R}^{2}$. 

Let $W_{\varepsilon}=\omega-\sigma_{\varepsilon}$. Then 
\begin{equation}
\left(\frac{\partial}{\partial t}+u\cdot\nabla-\nu\Delta\right)W_{\varepsilon}=G+\rho_{\varepsilon}\quad\textrm{ in }D,\label{eq:qq4}
\end{equation}
satisfying the boundary condition $\left.W_{\varepsilon}\right|_{\partial D}=0$
in the sense that
\begin{equation}
\lim_{x_{2}\rightarrow0\pm}W_{\varepsilon}(x_{1},x_{2},t)=0\quad\textrm{ for }x_{1}\geq0,\label{bd-W}
\end{equation}
where
\begin{equation*}
\rho_{\varepsilon}=\nu\Delta \sigma_{\varepsilon}-\frac{\partial \sigma_{\varepsilon}}{\partial t}-(u\cdot\nabla) \sigma_{\varepsilon}.
\end{equation*}
Taking into account that $\sigma_{\varepsilon}$ is given as in \eqref{ext-01}, we compute
\begin{align}
\rho_{\varepsilon}(x,t) &= \frac{\nu}{\varepsilon^{2}}\phi''(x_{2}/\varepsilon)\theta_{+}(x_{1},t)-\frac{1}{\varepsilon}\phi'(x_{2}/\varepsilon)u^{2}(x,t)\theta_{+}(x_{1},t)\nonumber \\
 & +\phi(x_{2}/\varepsilon)\left(\nu\frac{\partial^{2}\theta_{+}}{\partial x_{1}^{2}}(x_{1},t)-\frac{\partial\theta_{+}}{\partial t}(x_{1},t)\right)-\phi(x_{2}/\varepsilon)u^{1}(x,t)\frac{\partial\theta_{+}}{\partial x_{1}}(x_{1},t) \nonumber \\
& +\frac{\nu}{\varepsilon^{2}}\phi''(-x_{2}/\varepsilon)\theta_{-}(x_{1},t)+\frac{1}{\varepsilon}\phi'(-x_{2}/\varepsilon)u^{2}(x,t)\theta_{-}(x_{1},t)\nonumber \\
 & +\phi(-x_{2}/\varepsilon)\left(\nu\frac{\partial^{2}\theta_{-}}{\partial x_{1}^{2}}(x_{1},t)-\frac{\partial\theta_{-}}{\partial t}(x_{1},t)\right)-\phi(-x_{2}/\varepsilon)u^{1}(x,t)\frac{\partial\theta_{-}}{\partial x_{1}}(x_{1},t)\nonumber \\
 & +\Big( -\frac{1}{\varepsilon} \theta(t)\phi'(|x|/\varepsilon) \frac{u^1(x_1,x_2) x_1 + u^2(x_1,x_2) x_2}{|x|} \nonumber \\
 & -\theta'(t)\phi(|x|/\varepsilon) + \frac{\nu}{\varepsilon^2} \theta(t) \phi''(|x|/\varepsilon) + \frac{\nu}{\varepsilon} \theta(t) \phi'(|x|/\varepsilon) \frac{1}{|x|}\Big) 1_{\{x_1<0\}} \label{g-e-001}
\end{align}
for any $x\in D$. The initial data for $W_{\varepsilon}$ is given as follows
\begin{equation}
W_{\varepsilon}(x,0)=\omega_{0}(x_{1},x_{2})-\theta_{+}(x_{1},0)\phi(x_{2}/\varepsilon)-\theta_{-}(x_{1},0)\phi(-x_{2}/\varepsilon)-\theta(0)\phi(|x|/\varepsilon)1_{\{x_1<0\}},\label{W-e-001}
\end{equation}
for $x\in D$.

The velocity $u(x,t)$ is extended to the whole space $\mathbb{R}^{2}$
trivially by defining $u(x,t)=0$ if $x\notin D$. Then $\nabla\cdot u=0$
on $\mathbb{R}^{2}$ in distribution sense. Let $p(s,\xi,t,\eta)$
denote the transition probability density function of the diffusion
process with infinitesimal generator $\nu\Delta+u\cdot\nabla$, and
$p^{D}(s,\xi,t,\eta)$ be the transition probability density of the
same diffusion killed on hitting the boundary $\partial D$. Since
$u$ is divergence-free, $p^{D}(s,\xi,t,\eta)$ coincides with the
Green function of the Dirichlet problem associated with the heat operator
$\nu\Delta-u\cdot\nabla-\frac{\partial}{\partial t}$ on $D$, therefore
(cf. \citep[Chapter 1,  Theorem 12]{Friedman1964})
\begin{equation}
\omega(y,t)=\sigma_{\varepsilon}(y,t)+\int_{D}p^{D}(0,\xi,t,y)W_{\varepsilon}(\xi,0)\textrm{d}\xi+\int_{0}^{t}\int_{D}p^{D}(s,\xi,t,y)(G(\xi,s)+\rho_{\varepsilon}(\xi,s))\textrm{d}\xi\textrm{d}s\label{W-e-fud-01}
\end{equation}
for $y\in D$ and $t>0$. 
\begin{lem}\label{PlateOmegaLemma}
    The following integral representation holds
    \begin{align}
    \omega(y,t) & = (\omega_0(y_1,0+,t) + \omega_0(y_1,0-,t)) 1_{\{y \in \partial D\}} + \int_{D}p^{D}(0,\xi,t,y)\omega_{0}(\xi)\textrm{d}\xi \nonumber \\
     & +\int_{0}^{t}\int_{D}p^{D}(s,\xi,t,y)G(\xi,s)\textrm{d}\xi\textrm{d}s\nonumber \\
     & +\nu\int_{0}^{t}\int_{0}^{\infty}\frac{\partial}{\partial\xi_{2}}p^{D}(s,(\xi_{1},0+),t,y)\theta_{+}(\xi_{1},s)\textrm{d}\xi_{1}\textrm{d}s\nonumber \\
     & -\nu\int_{0}^{t}\int_{0}^{\infty}\frac{\partial}{\partial\xi_{2}}p^{D}(s,(\xi_{1},0-),t,y)\theta_{-}(\xi_{1},s)\textrm{d}\xi_{1}\textrm{d}s\label{ome-002}
    \end{align}
    for $y\in D$ and $t>0$. 
\end{lem}
\begin{proof}
    Following \cite{QQZW2022}, we obtain the result by letting $\varepsilon\downarrow 0$ in \eqref{W-e-fud-01} --- one can also take the same cutoff function $\phi$ given by
    \begin{equation}\label{CutoffPhi}
    \phi(r) = \begin{cases}
    1 &\text{ if } 0 \leq r \leq \frac{1}{3},\\
    \frac{1}{2}+54\Big(r-\frac{1}{2}\Big)^3 - \frac{9}{2} \Big(r-\frac{1}{2}\Big) &\text{ if } \frac{1}{3} < r \leq \frac{2}{3},\\
    0 &\text{ if } r > \frac{2}{3}.
    \end{cases}
    \end{equation}
    Taking \eqref{ext-01} to the limit, we have 
    \begin{equation}
        \sigma_{\varepsilon}(y, t) \to \theta_{+}(y_1, t) 1_{\{y_2=0\}} + \theta_{-}(y_1, t) 1_{\{y_2=0\}} = (\omega_0(y_1,0+,t) + \omega_0(y_1,0-,t)) 1_{\{y_1 \geq 0, y_2=0\}},
    \end{equation}
    as $\varepsilon \downarrow 0$ pointwise. We thus have that the integral 
    \begin{equation}
        \int_{D} p^{D}(0, \xi, t, y) \sigma_{\varepsilon}(\xi, 0) \textrm{d} \xi \to 0 \quad \textrm{ as } \varepsilon \downarrow 0,
    \end{equation}
    by dominated convergence as the limit of $\sigma_{\varepsilon}$ is zero almost everywhere. 

    Let us consider the integral 
    \begin{equation*}
        \int_{0}^{t}\int_{D}p^{D}(s,\xi,t,y)\rho_{\varepsilon}(\xi,s)\textrm{d}\xi\textrm{d}s.
    \end{equation*}
    The following terms in \eqref{g-e-001} 
    \begin{equation}
        \phi(\pm x_2/\varepsilon) \left( \nu \frac{\partial^2 \theta_{\pm}}{\partial x_1^2} (x_1, t) - \frac{\partial \theta_{\pm}}{\partial t}(x_1, t) \right) - \phi(\pm x_2/\varepsilon) u^{1}(x,t) \frac{\partial \theta_{\pm}}{\partial x_1}(x_1, t),
    \end{equation}
    converge pointwise to 
    \begin{equation}\label{d-e002}
        \Big(\nu \frac{\partial^2 \theta_{\pm}}{\partial x_1^2} (x_1, t) - \frac{\partial \theta_{\pm}}{\partial t}(x_1, t)\Big) 1_{\{x_2=0\}},
    \end{equation}
    as $\varepsilon \downarrow 0$, due to the no-slip condition. These terms again do not contribute to the limit as the function \eqref{d-e002} is zero almost everywhere. Similarly, the term 
    \begin{equation*}
        \theta'(t)\phi(|x|/\varepsilon)1_{\{x_1<0\}}
    \end{equation*}
    converges pointwise to zero and does not contribute to the limit.

    Consider now the following terms 
    \begin{equation}\label{s-g-030}
        \frac{1}{\varepsilon}\phi'(\pm x_{2}/\varepsilon)u^{2}(x,t)\theta_{\pm}(x_{1},t),
    \end{equation}
    and compute the limit of these terms in weak sense. Take a smooth function $\beta$ with compact support and write
    \begin{align*}
        \frac{1}{\varepsilon} \int_{\mathbb{R}^2} \beta(x) \phi'(x_2/\varepsilon) \textrm{d} x 
        &= \int_{-\infty}^{+\infty} \int_{\frac{1}{3}}^{\frac{2}{3}} \beta(x_1, \varepsilon x_2) \phi'(x_2) \textrm{d} x_2 \textrm{d} x_1 \\ 
        &= - \varepsilon \int_{-\infty}^{+\infty} \int_{\frac{1}{3}}^{\frac{2}{3}} \frac{\partial \beta}{\partial x_2} (x_1, \varepsilon x_2) \phi(x_2) \textrm{d} x_2 \textrm{d} x_1 \\ 
        &- \int_{-\infty}^{+\infty} \beta\Big(x_1, \frac{\varepsilon}{3}\Big) \textrm{d} x_1.
    \end{align*}
    Therefore,
    \begin{equation}
    \begin{split}
        \lim_{\varepsilon \to 0} \frac{1}{\varepsilon} \int_{\mathbb{R}^2} \beta(x) \phi'(x_2/\varepsilon) \textrm{d} x = - \int_{-\infty}^{+\infty} \beta(x_1, 0) \textrm{d} x_1,
    \end{split} 
    \end{equation}
    which is zero for terms in \eqref{s-g-030} due to the no-slip condition. Similarly, the integral with $\theta_{-}$ vanishes in the limit. For the terms
    \begin{equation}\label{s-f-020}
        \Big(-\frac{1}{\varepsilon} \theta(t)\phi'(|x|/\varepsilon) \frac{u^1(x_1,x_2) x_1 + u^2(x_1,x_2) x_2}{|x|} + \frac{\nu}{\varepsilon} \theta(t) \phi'(|x|/\varepsilon) \frac{1}{|x|}\Big) 1_{\{x_1<0\}},
    \end{equation}
    we write in polar coordinates
    \begin{align*}
        \frac{1}{\varepsilon} \int_{\mathbb{R}^2} \beta(x) \frac{\phi'(|x|/\varepsilon)}{|x|} \textrm{d} x 
        &= \int_{0}^{2\pi} \int_{\frac{1}{3}}^{\frac{2}{3}} \beta(\varepsilon r, \psi) \phi'(r) \textrm{d} r \textrm{d} \psi \\ 
        &= - \varepsilon \int_{0}^{2\pi} \int_{\frac{1}{3}}^{\frac{2}{3}} \frac{\partial \beta}{\partial r} (\varepsilon r, \psi) \phi(r) \textrm{d} r \textrm{d} \psi \\ 
        &- \int_{0}^{2\pi} \beta\Big(\frac{\varepsilon}{3}, \psi\Big) \textrm{d} \psi.
    \end{align*}
    Hence
    \begin{equation}
    \begin{split}
        \lim_{\varepsilon \to 0} \frac{1}{\varepsilon} \int_{\mathbb{R}^2} \beta(x) \frac{\phi'(|x|/\varepsilon)}{|x|} \textrm{d} x  = - \lim_{r \to 0} \int_{0}^{2\pi} \beta(r \cos\psi, r \sin\psi) \textrm{d} \psi,
    \end{split} 
    \end{equation}
    which is zero for \eqref{s-f-020}.
    
    Finally, consider the terms of the form
    \begin{equation}
        \frac{\nu}{\varepsilon^2} \theta_{\pm}(x_1, t)  \phi''(\pm x_2/\varepsilon).
    \end{equation}
    Again, take a function $\beta$ and write
    \begin{align*}
        \frac{1}{\varepsilon^2} \int_{\mathbb{R}^2} \beta(x) \phi''(x_2/\varepsilon) \textrm{d} x 
        &= \frac{1}{\varepsilon} \int_{-\infty}^{+\infty} \int_{\frac{1}{3}}^{\frac{2}{3}} \beta(x_1, \varepsilon x_2) \phi''(x_2) \textrm{d} x_2 \textrm{d} x_1 \\
        &= - \int_{-\infty}^{+\infty} \int_{\frac{1}{3}}^{\frac{2}{3}} \frac{\partial \beta}{\partial x_2} (x_1, \varepsilon x_2) \phi'(x_2) \textrm{d} x_2 \textrm{d} x_1 \\
        &= \varepsilon \int_{-\infty}^{+\infty} \int_{\frac{1}{3}}^{\frac{2}{3}} \frac{\partial^2 \beta}{\partial x_2^2} (x_1, \varepsilon x_2) \phi(x_2) \textrm{d} x_2 \textrm{d} x_1 \\
        &+ \int_{-\infty}^{+\infty} \frac{\partial \beta}{\partial x_2} \Big(x_1, \frac{\varepsilon}{3}\Big) \textrm{d} x_1.
    \end{align*}
    Therefore, 
    \begin{equation}
        \lim_{\varepsilon \to 0} \frac{1}{\varepsilon^2} \int_{\mathbb{R}^2} \beta(x) \phi''(x_2/\varepsilon) \textrm{d} x = \int_{-\infty}^{+\infty} \frac{\partial \beta}{\partial x_2} (x_1, 0) \textrm{d} x_1,
    \end{equation}
    and the corresponding term in $\rho_{\varepsilon}$ contributes to the limit as 
    \begin{equation}
        \lim_{\varepsilon \to 0} \frac{\nu}{\varepsilon^2} \int_{D} p^{D}(s, \xi, t, y) \theta_{+}(\xi_1, s)  \phi'' (\xi_2/\varepsilon) \textrm{d} \xi = \nu \int_0^{+\infty} \frac{\partial}{\partial \xi_2} p^{D}(s, (\xi_1, 0+), t, y)\theta_{+}(\xi_{1}, s) \textrm{d} \xi_{1},
    \end{equation}
    and we also have a similar contribution from the integral with $\theta_{-}$. Similarly, for the term
    \begin{equation}\label{d-s-200}
        \frac{\nu}{\varepsilon^2} \theta(t) \phi''(|x|/\varepsilon) 1_{\{x_1<0\}},
    \end{equation}
    we write
    \begin{align*}
        \frac{1}{\varepsilon^2} \int_{\mathbb{R}^2} \beta(x) \phi''(|x|/\varepsilon) \textrm{d} x 
        &= \int_{0}^{2\pi} \int_{\frac{1}{3}}^{\frac{2}{3}} \beta(\varepsilon r, \psi) \phi''(r) r \textrm{d} r \textrm{d} \psi \\
        &= - \int_{0}^{2\pi} \int_{\frac{1}{3}}^{\frac{2}{3}} \Big( \beta(\varepsilon r, \psi) + \varepsilon r \frac{\partial \beta}{\partial r} (\varepsilon r, \psi) \Big) \phi'(r) \textrm{d} r \textrm{d} \psi \\
        &= \int_{0}^{2\pi} \beta\Big(\frac{\varepsilon}{3}, \psi\Big) \textrm{d} \psi + \varepsilon \int_{0}^{2\pi} \int_{\frac{1}{3}}^{\frac{2}{3}} \frac{\partial \beta}{\partial r} (\varepsilon r, \psi) \phi(r) \textrm{d} r \textrm{d} \psi \\
        &-\varepsilon \int_{0}^{2\pi} \int_{\frac{1}{3}}^{\frac{2}{3}} r \frac{\partial \beta}{\partial r} (\varepsilon r, \psi) \phi'(r) \textrm{d} r \textrm{d} \psi.
    \end{align*}
    Therefore,
    \begin{equation}
    \begin{split}
        \lim_{\varepsilon \to 0} \frac{1}{\varepsilon^2} \int_{\mathbb{R}^2} \beta(x) \phi''(|x|/\varepsilon) \textrm{d} x = \lim_{r \to 0} \int_{0}^{2\pi}  \beta(r \cos\psi, r \sin\psi) \textrm{d} \psi,
    \end{split} 
    \end{equation}
    which is zero for the term \eqref{d-s-200}.
\end{proof}

Therefore, the Biot-Savart law implies the following representation for the velocity $u$.
\begin{lem}
\label{lem-3.1}The following integral
formula holds:
\begin{align}
u^{i}(x,t) & =\int_{D}\left(\int_{D}K^{i}(x,y)p^{D}(0,\xi,t,y)\textrm{d}y\right)\omega_{0}(\xi)\textrm{d}\xi\nonumber \\
 & +\int_{0}^{t}\int_{D}\left(\int_{D}K^{i}(x,y)p^{D}(s,\xi,t,y)\textrm{d}y\right)G(\xi,s)\textrm{d}\xi\textrm{d}s\nonumber \\
 & +\nu\int_{0}^{t}\int_{0}^{\infty}\left(\int_{D}K^{i}(x,y)\frac{\partial}{\partial\xi_{2}}p^{D}(s,(\xi_{1},0+),t,y)\textrm{d}y\right)\theta_{+}(\xi_{1},s)\textrm{d}\xi_{1}\textrm{d}s\nonumber \\
 & -\nu\int_{0}^{t}\int_{0}^{\infty}\left(\int_{D}K^{i}(x,y)\frac{\partial}{\partial\xi_{2}}p^{D}(s,(\xi_{1},0-),t,y)\textrm{d}y\right)\theta_{-}(\xi_{1},s)\textrm{d}\xi_{1}\textrm{d}s\label{lem-s01}
\end{align}
for every $x\in D$ and $t>0$.
\end{lem}
\begin{proof}
    As noticed above, the formula is implied by the Biot-Savart law \eqref{BS-law}. Note that the following term $(\omega_0(y_1,0+,t) + \omega_0(y_1,0-,t)) 1_{\{y \in \partial D\}}$ does not contribute to the representation as it is zero outside the boundary.
\end{proof}

Thanks to the integral representation \eqref{lem-s01}, we may establish
the random vortex dynamics as follows.

Firstly extend $u(x,t)$ to be zero if $x\notin D$, and define the
Taylor diffusion process family $X^{\xi,s}$ as the weak solution
to the stochastic differential equation
\begin{equation}
\begin{cases}
\textrm{d} X_{t}^{\xi,s}=u(X_{t}^{\xi,s},t)\textrm{d}t+\sqrt{2\nu}\textrm{d}B_{t}, & \textrm{ for }t\geq s,\\
X_{t}^{\xi,s}=\xi, & \textrm{ for }t\leq s.
\end{cases}\label{Ta-001}
\end{equation}
Let 
\[
\tau_{\xi,s}=\inf\left\{ t\geq s:X_{t}^{\xi,s}\in\partial D\right\} .
\]
Then
\[
\int_{D}K^{i}(x,y)p^{D}(s,\xi,t,y)\textrm{d}y=\mathbb{E}\left[K^{i}(x,X_{t}^{\xi,s}):t<\tau_{\xi,s}\right],
\]
\[
\int_{D}K^{i}(x,y)\frac{\partial}{\partial\xi_{2}}p^{D}(s,(\xi_{1},0+),t,y)\textrm{d}y=\left.\frac{\partial}{\partial\xi_{2}}\right|_{\xi_{2}=0+}\mathbb{E}\left[K^{i}(x,X_{t}^{\xi,s}):t<\tau_{\xi,s}\right]
\]
and
\[
\int_{D}K^{i}(x,y)\frac{\partial}{\partial\xi_{2}}p^{D}(s,(\xi_{1},0-),t,y)\textrm{d}y=\left.\frac{\partial}{\partial\xi_{2}}\right|_{\xi_{2}=0-}\mathbb{E}\left[K^{i}(x,X_{t}^{\xi,s}):t<\tau_{\xi,s}\right].
\]
Therefore the velocity $u(x,t)$ can be written as the following
\begin{align}\label{UReprPlate}
u^{i}(x,t) & =\int_{D}\mathbb{E}\left[K^{i}(x,X_{t}^{\xi,0}):t<\tau_{\xi,0}\right]\omega_{0}(\xi)\textrm{d}\xi\nonumber \\
 & +\int_{0}^{t}\int_{D}\mathbb{E}\left[K^{i}(x,X_{t}^{\xi,s}):t<\tau_{\xi,s}\right]G(\xi,s)\textrm{d}\xi\textrm{d}s\nonumber \\
 & +\nu\int_{0}^{t}\int_{0}^{\infty}\left.\frac{\partial}{\partial\xi_{2}}\right|_{\xi_{2}=0+}\mathbb{E}\left[K^{i}(x,X_{t}^{\xi,s}):t<\tau_{\xi,s}\right]\theta_{+}(\xi_{1},s)\textrm{d}\xi_{1}\textrm{d}s\nonumber \\
 & -\nu\int_{0}^{t}\int_{0}^{\infty}\left.\frac{\partial}{\partial\xi_{2}}\right|_{\xi_{2}=0-}\mathbb{E}\left[K^{i}(x,X_{t}^{\xi,s}):t<\tau_{\xi,s}\right]\theta_{-}(\xi_{1},s)\textrm{d}\xi_{1}\textrm{d}s.
\end{align}

\subsection{Flows past a wedge type obstacle}\label{WedgeSection}

In this subsection we derive a representation similar to the above for the wedge case. Here we will use the local coordinates for boundary components, and this idea will be used to derive a representation for more general domains subsequently. 

Recall that the wedge obstacle is given by $\Lambda=\left\{ x=r\textrm{e}^{\textrm{i}\theta}:r\geq0\textrm{ and }-\alpha\leq\theta\leq\alpha\right\}$ for a fixed $\alpha \in (0, \frac{\pi}{2})$, so the domain $D=\mathbb{R}^{2}\setminus\Lambda$. Note that the boundary $\partial D =  \{x=re^{i\theta}: r \geq 0, \theta = \pm \alpha\}$ has two components denoted $\partial D^{\pm}$ correspondingly. We also introduce normal coordinates $n^{\pm}, \tau^{\pm}$ in neighbourhood of $\partial D^{\pm}$ which are given by the following transformation
\begin{equation}\label{NormCoordDef}
    \left\{\begin{aligned}
    &\tau^{\pm}=x_1 \cos \alpha \pm x_2 \sin \alpha,\\
    &n^{\pm}= -x_1 \sin \alpha \pm x_2 \cos \alpha.
    \end{aligned}\right.
\end{equation}
In these coordinates, we denote the components of the vector field $u$ by $u^{\tau^{\pm}}, u^{n^{\pm}}$. Note that the boundary components are given by $\partial D^{\pm}=\{\tau^{\pm} \geq 0, n^{\pm}=0\}$. In this case, the no-slip condition reads
\begin{equation}
    u(x,t) = \lim_{n^{\pm} \to 0} u(\tau^{\pm}, n^{\pm}, t) = 0,
\end{equation}
for any $x \in \partial D^{\pm}$. As in \eqref{ThetaDefPlate}, the stress at the boundary is given by 
\begin{equation}\label{ThetaDefWedge}
    \theta_{\pm}(\tau^{\pm}, t) = \begin{cases}
    - \left. \frac{\partial u^{\tau^{\pm}}}{\partial n^{\pm}} (\tau^{\pm}, n^{\pm}, t) \right|_{n^{\pm}=0+} & \textrm{ for }\tau^{\pm} \geq 0,\\
    0 & \textrm{ for }\tau^{\pm}<0,
\end{cases}
\end{equation}
as the tangential derivatives $\frac{\partial u}{\partial \tau^{\pm}}$ vanish due to the no-slip condition. Note that the derivative operators are given by 
\begin{equation}
    \left\{\begin{aligned}
    &\frac{\partial}{\partial \tau^{\pm}}=\cos \alpha \frac{\partial}{\partial x_1} \pm \sin \alpha \frac{\partial}{\partial x_2},\\
    &\frac{\partial}{\partial n^{\pm}}=-\sin \alpha \frac{\partial}{\partial x_1} \pm \cos \alpha \frac{\partial}{\partial x_2};
    \end{aligned}\right.
\end{equation}
due to the definition \eqref{NormCoordDef}. 

We introduce the extension $\sigma_{\varepsilon}$ as follows 
\begin{equation}\label{SigmaEpsDef}
    \sigma_{\varepsilon}(x_1, x_2, t)=\begin{cases}
    \theta_{+}(\tau^{+}, t) \phi(n^{+}/\varepsilon), &\text{ if } \arg x \in [\alpha, \alpha+\frac{\pi}{2}],\\
    \theta(t)\phi(|x|/\varepsilon), &\text{ if } \arg x \in (\alpha+\frac{\pi}{2}, \frac{3\pi}{2}-\alpha),\\
    \theta_{-}(\tau^{-}, t) \phi(n^{-}/\varepsilon), &\text{ if } \arg x \in [\frac{3\pi}{2}-\alpha, 2\pi-\alpha];
    \end{cases}
\end{equation}
where the coordinates $n^{\pm}, \tau^{\pm}$ are given as functions of $x_1, x_2$ in \eqref{NormCoordDef}, and $\theta(t)$ denotes the value of $\theta_{+}(0,t)=\theta_{-}(0,t)$. Then $W^{\varepsilon}=\omega-\sigma_{\varepsilon}$ as before satisfies the following equation
\begin{equation}\label{WepsEq}
\left(\frac{\partial}{\partial t}+u\cdot\nabla-\nu\Delta\right)W^{\varepsilon}=G+\rho_{\varepsilon}\quad\textrm{ in }D,
\end{equation}
with homogeneous boundary condition
\begin{equation}
\lim_{n^{\pm} \to 0+} W^{\varepsilon}(\tau^{\pm}, n^{\pm}, t)=0\quad\textrm{ for } \tau^{\pm}\geq0,
\end{equation}
and initial data
\begin{align}
W_{0}^{\varepsilon}(x)&=\omega_{0}(x_{1},x_{2})-\theta_{+}(\tau^{+}, 0) \phi(n^{+}/\varepsilon) 1_{\{\arg x \in [\alpha, \alpha+\frac{\pi}{2}]\}} \nonumber \\ 
&-\theta_{-}(\tau^{-}, 0) \phi(n^{-}/\varepsilon) 1_{\{\arg x \in [\frac{3\pi}{2}-\alpha, 2\pi-\alpha]\}} \nonumber \\
&-\theta(0)\phi(|x|/\varepsilon) 1_{\{\arg x \in (\alpha+\frac{\pi}{2}, \frac{3\pi}{2}-\alpha)\}},
\end{align}
for $x\in D$.

The function $\rho_{\varepsilon}$ is given by
\begin{equation}
\rho_{\varepsilon} = \nu \Delta \sigma_{\varepsilon} - \frac{\partial \sigma_{\varepsilon}}{\partial t}-u\cdot \nabla \sigma_{\varepsilon},
\end{equation}
and, using the definition of $\sigma_{\varepsilon}$, can be written as follows
\begin{align}\label{Rhoeps}
    \rho_{\varepsilon}(x,t) &= \Big(\nu \frac{\partial^2 \theta_{+}}{\partial (\tau^{+})^2} (\tau^{+}, t) \phi(n^{+}/\varepsilon) + \frac{\nu}{\varepsilon^2} \theta_{+}(\tau^{+}, t)  \phi'' (n^{+}/\varepsilon) - \frac{\partial \theta_{+}}{\partial t}(\tau^{+}, t) \phi(n^{+}/\varepsilon) \nonumber \\
    &- u^{\tau^{+}}(\tau^{+}, n^{+},t) \frac{\partial \theta_{+}}{\partial \tau^{+}}(\tau^{+}, t) \phi(n^{+}/\varepsilon) - \frac{1}{\varepsilon} u^{n^{+}}(\tau^{+}, n^{+},t) \theta_{+}(\tau^{+}, t) \phi'(n^{+}/\varepsilon) \Big) 1_{\{\arg x \in [\alpha, \alpha+\frac{\pi}{2}]\}} \nonumber \\
    &+ \Big(\nu \frac{\partial^2 \theta_{-}}{\partial (\tau^{-})^2} (\tau^{-}, t) \phi(n^{-}/\varepsilon) + \frac{\nu}{\varepsilon^2} \theta_{-}(\tau^{-}, t)  \phi'' (n^{-}/\varepsilon) - \frac{\partial \theta_{-}}{\partial t}(\tau^{-}, t) \phi(n^{-}/\varepsilon) \nonumber \\
    &- u^{\tau^{-}}(\tau^{-}, n^{-},t) \frac{\partial \theta_{-}}{\partial \tau^{-}}(\tau^{-}, t) \phi(n^{-}/\varepsilon) - \frac{1}{\varepsilon} u^{n^{-}}(\tau^{-}, n^{-},t) \theta_{-}(\tau^{-}, t) \phi'(n^{-}/\varepsilon) \Big) 1_{\{\arg x \in [\frac{3\pi}{2}-\alpha, 2\pi-\alpha]\}} \nonumber \\
    & +\Big( -\frac{1}{\varepsilon} \theta(t)\phi'(|x|/\varepsilon) \frac{u^1(x_1,x_2) x_1 + u^2(x_1,x_2) x_2}{|x|} -\theta'(t)\phi(|x|/\varepsilon) + \frac{\nu}{\varepsilon^2} \theta(t) \phi''(|x|/\varepsilon) \nonumber \\
    & + \frac{\nu}{\varepsilon} \theta(t) \phi'(|x|/\varepsilon) \frac{1}{|x|}\Big) 1_{\{\arg x \in (\alpha+\frac{\pi}{2}, \frac{3\pi}{2}-\alpha)\}}.
\end{align}

As before, we write the representation for the solution $W^{\varepsilon}$ to \eqref{WepsEq} in terms of the transition probability density $p^D(s, \xi, t, \eta)$ which implies 
\begin{equation}\label{WepsRepr}
    \omega(y,t) = \sigma_{\varepsilon}(y,t) + \int_{D} p^{D}(0, \xi, t, y) W^{\varepsilon}(\xi, 0) \textrm{d} \xi + \int_0^t \int_{D} p^{D}(s, \xi, t, y) (G(\xi, s) + \rho_{\varepsilon}(\xi, s)) \textrm{d} \xi \textrm{d} s,
\end{equation}
for $y \in D$ and $t>0$.

\begin{lem}
    We have the following representation:
\begin{align}
     \omega(y,t) &= \omega(y,t) 1_{\{y \in \partial D\}} + \int_{D} p^{D}(0, \xi, t, y) \omega_0(\xi) \textrm{d} \xi + \int_0^t \int_{D} p^{D}(s, \xi, t, y) G(\xi, s) \textrm{d} \xi \textrm{d} s \nonumber \\ 
     &+ \nu \int_0^t \int_0^{+\infty} \frac{\partial}{\partial n^{+}} p^{D}(s, (\tau^{+}, 0+), t, y)\theta_{+}(\tau^{+}, s) \textrm{d} \tau^{+} \textrm{d} s \nonumber \\ 
     &+ \nu \int_0^t \int_0^{+\infty} \frac{\partial}{\partial n^{-}} p^{D}(s, (\tau^{-}, 0+), t, y)\theta_{-}(\tau^{-}, s) \textrm{d} \tau^{-} \textrm{d} s,
\end{align}
for $y \in D$ and $t>0$.
\end{lem}
\begin{proof}
    To obtain the result, we again let $\varepsilon \downarrow 0$ in the representation \eqref{WepsRepr}. First, it is easy to see that due to \eqref{SigmaEpsDef}, 
    \begin{equation}
    \sigma_{\varepsilon}(y, t) \to \theta_{+}(\tau^{+}, t) 1_{\{\tau^{+} \geq 0, n^{+}=0\}} + \theta_{-}(\tau^{-}, t) 1_{\{\tau^{-} \geq 0, n^{-}=0\}} = \omega(y,t) 1_{\{y \in \partial D\}},
    \end{equation}
    pointwise as $\varepsilon \to 0$. This also implies that the integral 
    \begin{equation}
    \int_{D} p^{D}(0, \xi, t, y) \sigma_{\varepsilon}(\xi, 0) \textrm{d} \xi \to 0, \quad \textrm{ as } \varepsilon \downarrow 0,
    \end{equation}
    by dominated convergence. 

    Now let us take the limit in the last integral in 
    \eqref{WepsRepr}. It is easy to see that in \eqref{Rhoeps}, the following terms 
    \begin{equation}
        \phi(n^{\pm}/\varepsilon) \left( \nu \frac{\partial^2 \theta_{\pm}}{\partial (\tau^{\pm})^2} (\tau^{\pm}, t) - \frac{\partial \theta_{\pm}}{\partial t}(\tau^{\pm}, t) \right) - u^{\tau^{\pm}}(\tau^{\pm}, n^{\pm},t) \frac{\partial \theta_{\pm}}{\partial \tau^{\pm}}(\tau^{\pm}, t) \phi(n^{\pm}/\varepsilon),
    \end{equation}
    as $\varepsilon \to 0$, converge pointwise to 
    \begin{equation}
        \Big(\nu \frac{\partial^2 \theta_{\pm}}{\partial (\tau^{\pm})^2} (\tau^{\pm}, t) - \frac{\partial \theta_{\pm}}{\partial t}(\tau^{\pm}, t)\Big) 1_{\{\tau^{\pm} \geq 0, n^{\pm}=0\}},
    \end{equation}
    since $u|_{\partial D^{\pm}} = 0$ and therefore, these terms do not contribute to the limit. Similarly, 
    \begin{equation*}
        -\theta'(t)\phi(|x|/\varepsilon)1_{\{\arg x \in (\alpha+\frac{\pi}{2}, \frac{3\pi}{2}-\alpha)\}}
    \end{equation*}
    converges to zero pointwise and does not contribute to the limiting representation. 

    Let us consider now the term
    \begin{equation}\label{Terms1}
        \frac{1}{\varepsilon} u^{n^{+}}(\tau^{+}, n^{+},t) \theta_{+}(\tau^{+}, t) \phi'(n^{+}/\varepsilon) 1_{\{\arg x \in [\alpha, \alpha+\frac{\pi}{2}]\}}.
    \end{equation}
    Take a smooth function $\beta$ with compact support and consider first the integral
    \begin{align*}
        \frac{1}{\varepsilon} \int_{\mathbb{R}^2} \beta(x) \phi'(n^{+}/\varepsilon) 1_{\{\arg x \in [\alpha, \alpha+\frac{\pi}{2}]\}} \textrm{d} x &= \frac{1}{\varepsilon} \int_0^{+\infty} \int_0^{+\infty} \beta(\tau^{+}, n^{+}) \phi'(n^{+}/\varepsilon) \textrm{d} n^{+} \textrm{d} \tau^{+} \\ 
        &= \int_0^{+\infty} \int_{\frac{1}{3}}^{\frac{2}{3}} \beta(\tau^{+}, \varepsilon n^{+}) \phi'(n^{+}) \textrm{d} n^{+} \textrm{d} \tau^{+} \\ 
        &= - \int_0^{+\infty} \beta\Big(\tau^{+}, \frac{\varepsilon}{3}\Big) \textrm{d} \tau^{+}\\ 
        &- \varepsilon \int_0^{+\infty} \int_{\frac{1}{3}}^{\frac{2}{3}} \frac{\partial \beta}{\partial n^{+}} (\tau^{+}, \varepsilon n^{+}) \phi(n^{+}) \textrm{d} n^{+} \textrm{d} \tau^{+}.
    \end{align*}
    Therefore, we have 
    \begin{equation}
    \begin{split}
        \lim_{\varepsilon \to 0} \frac{1}{\varepsilon} \int_{\mathbb{R}^2} \beta(x) \phi'(n^{+}/\varepsilon) 1_{\{\arg x \in [\alpha, \alpha+\frac{\pi}{2}]\}} \textrm{d} x = - \int_0^{+\infty} \beta(\tau^{+}, 0) \textrm{d} \tau^{+},
    \end{split} 
    \end{equation}
    which is zero for terms in \eqref{Terms1} due to the no-slip condition. Similarly, the integral with the term 
    \begin{equation*}
        \frac{1}{\varepsilon} u^{n^{-}}(\tau^{-}, n^{-},t) \theta_{-}(\tau^{-}, t) \phi'(n^{-}/\varepsilon) 1_{\{\arg x \in [\frac{3\pi}{2}-\alpha, 2\pi-\alpha]\}}.
    \end{equation*}
    vanishes in the limit. Consider now the following terms
    \begin{equation}\label{s-x-020}
    \Big(-\frac{1}{\varepsilon} \theta(t)\phi'(|x|/\varepsilon) \frac{u^1(x_1,x_2) x_1 + u^2(x_1,x_2) x_2}{|x|} + \frac{\nu}{\varepsilon} \theta(t) \phi'(|x|/\varepsilon) \frac{1}{|x|}\Big) 1_{\{\arg x \in (\alpha+\frac{\pi}{2}, \frac{3\pi}{2}-\alpha)\}}.
    \end{equation}
    Recall that in the proof of Lemma \ref{PlateOmegaLemma} we had
    \begin{equation}
    \begin{split}
        \lim_{\varepsilon \to 0} \frac{1}{\varepsilon} \int_{\mathbb{R}^2} \beta(x) \frac{\phi'(|x|/\varepsilon)}{|x|} \textrm{d} x  = - \lim_{r \to 0} \int_{0}^{2\pi} \beta(r\cos\psi,r\sin\psi) \textrm{d} \psi,
    \end{split} 
    \end{equation}
    which is zero for the terms \eqref{s-x-020}.

    Finally, consider the terms of the form
    \begin{equation}
        \frac{\nu}{\varepsilon^2} \theta_{\pm}(\tau^{\pm}, t)  \phi'' (n^{\pm}/\varepsilon).
    \end{equation}
    Again, take a function $\beta$ and write
    \begin{align*}
        \frac{1}{\varepsilon^2} \int_{\mathbb{R}^2} \beta(x) \phi''(n^{+}/\varepsilon) 1_{\{\arg x \in [\alpha, \alpha+\frac{\pi}{2}]\}} \textrm{d} x &= \frac{1}{\varepsilon^2} \int_0^{+\infty} \int_0^{+\infty} \beta(\tau^{+}, n^{+}) \phi''(n^{+}/\varepsilon) \textrm{d} n^{+} \textrm{d} \tau^{+} \\ 
        &= \frac{1}{\varepsilon} \int_0^{+\infty} \int_{\frac{1}{3}}^{\frac{2}{3}} \beta(\tau^{+}, \varepsilon n^{+}) \phi''(n^{+}) \textrm{d} n^{+} \textrm{d} \tau^{+} \\
        &= - \int_0^{+\infty} \int_{\frac{1}{3}}^{\frac{2}{3}} \frac{\partial \beta}{\partial n^{+}} (\tau^{+}, \varepsilon n^{+}) \phi'(n^{+}) \textrm{d} n^{+} \textrm{d} \tau^{+} \\
        &= \int_0^{+\infty} \frac{\partial \beta}{\partial n^{+}} \Big(\tau^{+}, \frac{\varepsilon}{3}\Big) \textrm{d} \tau^{+} \\
        &+ \varepsilon \int_0^{+\infty} \int_{\frac{1}{3}}^{\frac{2}{3}} \frac{\partial^2 \beta}{\partial (n^{+})^2} (\tau^{+}, \varepsilon n^{+}) \phi(n^{+}) \textrm{d} n^{+} \textrm{d} \tau^{+}.
    \end{align*}
    Therefore, 
    \begin{equation}
        \lim_{\varepsilon \to 0} \frac{1}{\varepsilon^2} \int_{\mathbb{R}^2} \beta(x) \phi''(n^{+}/\varepsilon) 1_{\{\arg x \in [\alpha, \alpha+\frac{\pi}{2}]\}} \textrm{d} x = \int_0^{+\infty} \frac{\partial \beta}{\partial n^{+}} (\tau^{+}, 0) \textrm{d} \tau^{+},
    \end{equation}
    and the corresponding term in $\rho_{\varepsilon}$ contributes to the limit as 
    \begin{equation}
    \begin{split}
        \lim_{\varepsilon \to 0} \frac{\nu}{\varepsilon^2} \int_{D} p^{D}(s, \xi, t, y) \theta_{+}(\tau^{+}, s)  \phi'' (n^{+}/\varepsilon) 1_{\{y_2 \geq 0 \}} \textrm{d} \xi \\
        = \nu \int_0^{+\infty} \frac{\partial}{\partial n^{+}} p^{D}(s, (\tau^{+}, 0+), t, y)\theta_{+}(\tau^{+}, s) \textrm{d} \tau^{+},
    \end{split}
    \end{equation}
    and we also have a similar contribution from the integral with $\theta_{-}$. Lastly, consider the term
    \begin{equation}\label{d-x-200}
        \frac{\nu}{\varepsilon^2} \theta(t) \phi''(|x|/\varepsilon) 1_{\{\arg x \in (\alpha+\frac{\pi}{2}, \frac{3\pi}{2}-\alpha)\}}.
    \end{equation}
    Recall that in Lemma \ref{PlateOmegaLemma}, it was shown that
    \begin{equation}
    \begin{split}
        \lim_{\varepsilon \to 0} \frac{1}{\varepsilon^2} \int_{\mathbb{R}^2} \beta(x) \phi''(|x|/\varepsilon) \textrm{d} x = \lim_{r \to 0} \int_{0}^{2\pi} \beta(r\cos\psi,r\sin\psi) \textrm{d} \psi,
    \end{split} 
    \end{equation}
    which is zero for the term \eqref{d-x-200}.
\end{proof}

We again have the following representation for the velocity due to the Biot-Savart law.

\begin{lem}
    The following representation holds:
    \begin{align}
    u^i(x,t) & =\int_{D}\left(\int_{D}K^i(x,y)p^{D}(0,\xi,t,y)\textrm{d}y\right)\omega_{0}(\xi)\textrm{d}\xi\nonumber \\
    & +\int_{0}^{t}\int_{D}\left(\int_{D}K^i(x,y)p^{D}(s,\xi,t,y)\textrm{d}y\right)G(\xi,s)\textrm{d}\xi\textrm{d}s\nonumber \\
    & +\nu\int_{0}^{t}\int_{0}^{\infty}\left(\int_{D}K^i(x,y)\frac{\partial}{\partial n^{+}} p^{D}(s, (\tau^{+}, 0+), t, y)\textrm{d}y\right)\theta_{+}(\tau^{+}, s) \textrm{d} \tau^{+}\textrm{d}s\nonumber \\
    & +\nu\int_{0}^{t}\int_{0}^{\infty}\left(\int_{D}K^i(x,y)\frac{\partial}{\partial n^{-}} p^{D}(s, (\tau^{-}, 0+), t, y)\textrm{d}y\right)\theta_{-}(\tau^{-}, s) \textrm{d} \tau^{-}\textrm{d}s,
\end{align}
for any $x \in D$ and $t > 0$.
\end{lem}

We can use the Taylor diffusion processes $X^{\xi, s}$ defined in \eqref{Ta-001} to represent $u$ as follows. Notice that
\[
\int_{D}K^i(x,y)p^{D}(s,\xi,t,y)\textrm{d}y=\mathbb{E}\left[K^i(x,X_{t}^{\xi,s}):t<\tau_{\xi,s}\right],
\]
and
\[
\int_{D}K^i(x,y)\left.\frac{\partial}{\partial n^{\pm}} \right|_{n^{\pm} = 0+} p^{D}(s, (\tau^{\pm}, n^{\pm}), t, y)\textrm{d}y=\left.\frac{\partial}{\partial n^{\pm}} \right|_{n^{\pm} = 0+}\mathbb{E}\left[K^i(x,X_{t}^{\xi,s}):t<\tau_{\xi,s}\right],
\]
where $\tau_{\xi,s}$ is again the boundary hitting time for $X^{\xi, s}$. This implies that
\begin{align}\label{UReprWedge}
     u^i(x,t) = \int_{D} \mathbb{E}\left[K^i(x,X_{t}^{\xi,0}):t<\tau_{\xi,0}\right] \omega_0(\xi) \textrm{d} \xi 
     + \int_0^t \int_{D} \mathbb{E}\left[K^i(x,X_{t}^{\xi,s}):t<\tau_{\xi,s}\right] g(\xi, s) \textrm{d} \xi \textrm{d} s \nonumber \\ 
     + \nu \int_0^t \int_0^{+\infty} \frac{\partial}{\partial n^{+}} \Big|_{n^{+} = 0+} \mathbb{E}\left[K^i(x,X_{t}^{\xi,s}):t<\tau_{\xi,s}\right] \theta_{+}(\tau^{+}, s) \textrm{d} \tau^{+} \textrm{d} s \nonumber \\ 
     + \nu \int_0^t \int_0^{+\infty} \frac{\partial}{\partial n^{-}} \Big|_{n^{-} = 0+}\mathbb{E}\left[K^i(x,X_{t}^{\xi,s}):t<\tau_{\xi,s}\right] \theta_{-}(\tau^{-}, s) \textrm{d} \tau^{-} \textrm{d} s.
\end{align}

\subsection{General case}

We consider a flow in $D\subset\mathbb{R}^{2}$ assuming it is a proper simply connected domain. Let us fix a conformal mapping $T:D\rightarrow\mathbb{H}$. Notice that in this case we can parameterise the domain $D$ using the coordinates $(z_1, z_2) \in \mathbb{H}$, i.e. the transform is given by $x=T^{-1}(z)$ for $x \in \mathbb{H}$. The boundary in this coordinates is given by $\partial D=\{z_2=0\}$, and the no-slip condition is written as 
\begin{equation}
    \lim_{z_2 \to 0+} u(z_1, z_2, t) = 0.
\end{equation}
Notice that as the velocity $u$ is a vector field, its components $u^{z_1}, u^{z_2}$ in new coordinates are transformed correspondingly. We again introduce the boundary stress as 
\begin{equation}
    \theta(z_1, t) = - \left. \frac{\partial u^{z_1}}{\partial z_2} (z_1, z_2, t) \right|_{z_2=0+}.
\end{equation}
Note that in the subsequent argument we have to assume that the boundary $\partial D$ is sufficiently regular, in particular, $T$ gives a smooth parametrisation of the boundary (except, maybe, one point in $\partial D$ which is sent to infinity by $T$). However, as for the plate and wedge cases we were able to take care of singularities at the boundary $\partial D$, we expect a similar formula to hold more generally, say, for domains with piecewise smooth boundary.

Let us define the extension $\sigma_{\varepsilon}$ in new coordinates by
\begin{equation}
    \sigma_{\varepsilon} (z_1, z_2, t) = \theta(z_1, t) \phi(z_2/\varepsilon),
\end{equation}
and introduce $W^{\varepsilon}=\omega-\sigma_{\varepsilon}$. It has to satisfy 
\begin{equation}\label{s-p-112}
\left(\frac{\partial}{\partial t}+u\cdot\nabla-\nu\Delta\right)W^{\varepsilon}=G+\rho_{\varepsilon}\quad\textrm{ in }D,
\end{equation}
with homogeneous boundary condition $\left. W^{\varepsilon} \right|_{\partial D}=0$ written in new coordinates as 
\begin{equation*}
    \lim_{z_2 \to 0+} W^{\varepsilon}(z_1, z_2, t)=0
\end{equation*}
for all $z_1$. The initial data for $W^{\varepsilon}$ is 
\begin{equation}
    W_0^{\varepsilon}(z)=\omega_0(z_1,z_2)-\theta(z_1,0)\phi(z_2/\varepsilon),
\end{equation}
and $\rho_{\varepsilon}$ is given by
\begin{equation}
\rho_{\varepsilon} = \nu \Delta \sigma_{\varepsilon} - \frac{\partial \sigma_{\varepsilon}}{\partial t}-u\cdot \nabla \sigma_{\varepsilon}.
\end{equation}
Therefore, in new coordinates
\begin{align}
    \rho_{\varepsilon}(z,t) = - \frac{\partial \theta}{\partial t}(z_1, t)\phi(z_2/\varepsilon) &- u^{z_1}(z_1, z_2,t) \frac{\partial \theta}{\partial z_1}(z_1, t) \phi(z_2/\varepsilon) - \frac{1}{\varepsilon} u^{z_2}(z_1, z_2,t) \theta(z_1, t) \phi'(z_2/\varepsilon) \nonumber \\
    &+ \nu |T'|^2 \Big(\frac{\partial^2 \theta}{\partial z_1^2}(z_1, t) \phi(z_2/\varepsilon) + \frac{1}{\varepsilon^2} \theta(z_1, t) \phi''(z_2/\varepsilon)\Big),
\end{align}
where $T'$ denotes the derivative of $T$ as a complex function.

Writing the solution to \eqref{s-p-112} in terms of the transition density $p^D(s, \xi, t, \eta)$ we have
\begin{equation}\label{s-q-114}
    \omega(y,t) = \sigma_{\varepsilon}(y,t) + \int_{D} p^{D}(0, \xi, t, y) W^{\varepsilon}_0(\xi) \textrm{d} \xi + \int_0^t \int_{D} p^{D}(s, \xi, t, y) (G(\xi, s) + \rho_{\varepsilon}(\xi, s)) \textrm{d} \xi \textrm{d} s,
\end{equation}
for $y \in D$ and $t>0$.

\begin{lem}
    We have the following representation:
\begin{align}
     \omega(y,t) = \omega(y,t) 1_{\{y \in \partial D\}} &+ \int_{D} p^{D}(0, \xi, t, y) \omega_0(\xi) \textrm{d} \xi + \int_0^t \int_{D} p^{D}(s, \xi, t, y) G(\xi, s) \textrm{d} \xi \textrm{d} s \nonumber \\ 
     &+ \nu \int_0^t \int_{-\infty}^{+\infty} \frac{\partial}{\partial z_2} \Big|_{z_2 = 0} p^{D}(s, z, t, y)\theta(z_1, s) \textrm{d} z_1 \textrm{d} s,
\end{align}
for $y \in D$ and $t>0$.
\end{lem}
\begin{proof}
    We take as usual the limit in \eqref{s-q-114} as $\varepsilon \to 0$. Notice that as $\varepsilon \to 0$,
    \begin{equation}
    \sigma_{\varepsilon}(z, t) \to \theta(z_1, t) 1_{\{z_2=0\}},
    \end{equation}
    pointwise, which can be written as $\omega(y,t) 1_{\{y \in \partial D\}}$. This also implies that the integral 
    \begin{equation}
    \int_{D} p^{D}(0, \xi, t, y) \sigma_{\varepsilon}(\xi, 0) \textrm{d} \xi \to 0, \quad \textrm{ as } \varepsilon \downarrow 0,
    \end{equation}
    by dominated convergence. 

    Now let us find the limit of
    \begin{equation}\label{e-m-034}
    \int_0^t \int_{D} p^{D}(s, \xi, t, y) \rho_{\varepsilon}(\xi, s) \textrm{d}\xi \textrm{d}s,
    \end{equation}
    as $\varepsilon \to 0$. Notice that the terms   
    \begin{equation}
    - \frac{\partial \theta}{\partial t}(z_1, t)\phi(z_2/\varepsilon) - u^{z_1}(z_1, z_2,t) \frac{\partial \theta}{\partial z_1}(z_1, t) \phi(z_2/\varepsilon) + \nu |T'|^2 \frac{\partial^2 \theta}{\partial z_1^2}(z_1, t) \phi(z_2/\varepsilon),
    \end{equation}
    converge pointwise to 
    \begin{equation}\label{i-l-032}
    \Big(-\frac{\partial \theta}{\partial t}(z_1, t) + \nu |T'|^2 \frac{\partial^2 \theta}{\partial z_1^2}(z_1, t)\Big)1_{\{z_2=0\}},
    \end{equation}
    as $\left. u \right|_{\partial D} = 0$. These terms do not contribute to the limit as \eqref{i-l-032} is zero almost everywhere. 

    Consider now the term
    \begin{equation}\label{b-o-012}
         - \frac{1}{\varepsilon} u^{z_2}(z_1, z_2,t) \theta(z_1, t) \phi'(z_2/\varepsilon).
    \end{equation}
    Take a smooth function $\beta$ with compact support and write
    \begin{equation*}
        \frac{1}{\varepsilon} \int_{\mathbb{R}^2} \beta(x) \phi'(z_2(x)/\varepsilon) \textrm{d} x = \frac{1}{\varepsilon} \int_{-\infty}^{+\infty} \int_0^{+\infty} \beta(z_1, z_2) \phi'(z_2/\varepsilon) J(T^{-1}) \textrm{d} z_2 \textrm{d} z_1,
    \end{equation*}
    where $J(T^{-1})$ is the Jacobian of the transform $x=T^{-1}(z)$. Write the latter integral as
    \begin{align*}
        \int_{-\infty}^{+\infty} \int_{\frac{1}{3}}^{\frac{2}{3}} \beta(z_1, \varepsilon z_2) \phi'(z_2) J(T^{-1})(z_1, \varepsilon z_2) \textrm{d} z_2 \textrm{d} z_1 &= - \int_{-\infty}^{+\infty} \beta\left(z_1, \frac{\varepsilon}{3}\right) J(T^{-1})\left(z_1,\frac{\varepsilon}{3}\right) \textrm{d} z_1\\ 
        &-\varepsilon \int_{-\infty}^{+\infty} \int_{\frac{1}{3}}^{\frac{2}{3}} \phi(z_2)\frac{\partial \beta}{\partial z_2}(z_1, \varepsilon z_2) J(T^{-1})(z_1, \varepsilon z_2)\textrm{d} z_2 \textrm{d} z_1\\ 
        &-\varepsilon \int_{-\infty}^{+\infty} \int_{\frac{1}{3}}^{\frac{2}{3}} \phi(z_2) \beta(z_1, \varepsilon z_2) \frac{\partial}{\partial z_2} J(T^{-1})(z_1, \varepsilon z_2) \textrm{d} z_2 \textrm{d} z_1.
    \end{align*}
    Therefore, the limit of the term \eqref{b-o-012} is equal to zero due to the no-slip condition as
    \begin{equation}
    \begin{split}
        \lim_{\varepsilon \to 0} \frac{1}{\varepsilon} \int_{\mathbb{R}^2} \beta(x) \phi'(z_2(x)/\varepsilon) \textrm{d} x = - \int_{-\infty}^{+\infty} \beta(z_1, 0) J(T^{-1})(z_1,0) \textrm{d} z_1.
    \end{split} 
    \end{equation}

    Let us consider the last term 
    \begin{equation}
         \frac{\nu}{\varepsilon^2} |T'|^2 \theta(z_1, t) \phi''(z_2/\varepsilon).
    \end{equation}
    Again, taking a smooth compactly supported function $\beta$, write
    \begin{equation*}
        \frac{1}{\varepsilon^2} \int_{\mathbb{R}^2} \beta(x) |T'|^2 \phi''(z_2(x)/\varepsilon) \textrm{d} x = \frac{1}{\varepsilon^2} \int_{-\infty}^{+\infty} \int_0^{+\infty} \beta(z_1, z_2) \phi''(z_2/\varepsilon) \textrm{d} z_2 \textrm{d} z_1,
    \end{equation*}
    since the Jacobian $J(T) = |T'|^2$. Therefore, 
    \begin{align*}
        \frac{1}{\varepsilon} \int_{-\infty}^{+\infty} \int_{\frac{1}{3}}^{\frac{2}{3}} \beta(z_1, \varepsilon z_2) \phi''(z_2) \textrm{d} z_2 \textrm{d} z_1 
        &= - \int_{-\infty}^{+\infty} \int_{\frac{1}{3}}^{\frac{2}{3}} \frac{\partial \beta}{\partial z_2} (z_1, \varepsilon z_2) \phi'(z_2) \textrm{d} z_2 \textrm{d} z_1 \\
        &= \int_{-\infty}^{+\infty} \frac{\partial \beta}{\partial z_2} \Big(z_1, \frac{\varepsilon}{3}\Big) \textrm{d} z_1 \\
        &+ \varepsilon \int_{-\infty}^{+\infty} \int_{\frac{1}{3}}^{\frac{2}{3}} \frac{\partial^2 \beta}{\partial z_2^2} (z_1, \varepsilon z_2) \phi(z_2) \textrm{d} z_2 \textrm{d} z_1,
    \end{align*}
    and the limit 
    \begin{equation}
        \lim_{\varepsilon \to 0} \frac{1}{\varepsilon^2} \int_{\mathbb{R}^2} \beta(x) |T'|^2 \phi''(z_2(x)/\varepsilon)  \textrm{d} x = \int_{-\infty}^{+\infty} \frac{\partial \beta}{\partial z_2} (z_1, 0) \textrm{d} z_1.
    \end{equation}
    Then it is the only term contributing to the limit, and we have that 
    \begin{equation*}
        \int_0^t \int_{D} p^{D}(s, \xi, t, y) \rho_{\varepsilon}(\xi, s) \textrm{d}\xi \textrm{d}s \to \nu \int_0^t \int_{-\infty}^{+\infty} \frac{\partial}{\partial z_2} \Big|_{z_2 = 0+} p^{D}(s, z, t, y)\theta(z_1, s) \textrm{d} z_1 \textrm{d} s,
    \end{equation*}
    as $\varepsilon \to 0$.
\end{proof}

\begin{lem}
    We have the following:
    \begin{align}
    u^i(x,t) & =\int_{D}\left(\int_{D}K^i(x,y)p^{D}(0,\xi,t,y)\textrm{d}y\right)\omega_{0}(\xi)\textrm{d}\xi\nonumber \\
    & +\int_{0}^{t}\int_{D}\left(\int_{D}K^i(x,y)p^{D}(s,\xi,t,y)\textrm{d}y\right)G(\xi,s)\textrm{d}\xi\textrm{d}s\nonumber \\
    & +\nu\int_{0}^{t}\int_{-\infty}^{+\infty}\left(\int_{D}K^i(x,y)\frac{\partial}{\partial z_2} \Big|_{z_2 = 0+} p^{D}(s, (z_1, z_2), t, y)\textrm{d}y\right)\theta(z_1, s) \textrm{d} z_1 \textrm{d}s,
\end{align}
for any $x \in D$ and $t > 0$ --- note that the last integral is written in coordinates $z_1, z_2$.
\end{lem}

The lemma above follows as usual from the Biot-Savart law. Moreover, we can again write the representation for the velocity $u$ in terms of the Taylor diffusions $X^{\xi, s}$ given in \eqref{Ta-001} as follows
\begin{align}
     u^i(x,t) = \int_{D} \mathbb{E}\left[K^i(x,X_{t}^{\xi,0}):t<\tau_{\xi,0}\right] \omega_0(\xi) \textrm{d} \xi 
     + \int_0^t \int_{D} \mathbb{E}\left[K^i(x,X_{t}^{\xi,s}):t<\tau_{\xi,s}\right] G(\xi, s) \textrm{d} \xi \textrm{d} s \nonumber \\ 
     + \nu \int_0^t \int_{-\infty}^{+\infty} \frac{\partial}{\partial z_2} \Big|_{z_2 = 0+} \mathbb{E}\left[K^i(x,X_{t}^{(z_1,z_2),s}):t<\tau_{z,s}\right] \theta(z_1, s) \textrm{d} z_1 \textrm{d} s.
\end{align}

\section{Monte-Carlo simulations}\label{MCSimulationsSection}

\subsection{Numerical scheme}

The integral representations for the velocity $u$ we have obtained so far, together with the Taylor diffusions, form a closed system. Therefore, we can discretise the corresponding representations to establish the following numerical schemes. For our simulations we use the representation for $u$ as in Theorem \ref{thm5.4new} since in this case we work only with the diffusion processes $X_t^{\eta}$ initialised at time $t=0$ (compared to \eqref{UReprPlate}, \eqref{UReprWedge}). Thus we use the representation 
\begin{align}
u^{i}(x,t) & =\int_{D} K^i(x,\eta) \sigma_{\varepsilon}(\eta,t)\textrm{d}\eta+\int_{D}\mathbb{E}\left[K^{i}(x,X_{t}^{\eta})1_{\left\{ t<\zeta(X^{\eta}\circ\tau_{t})\right\} }\right]W_{\varepsilon}(\eta,0)\textrm{d}\eta\nonumber \\
 & +\int_{0}^{t}\int_{D}\mathbb{E}\left[1_{\{t-s<\zeta(X^{\eta}\circ\tau_{t})\}}K^{i}(x,X_{t}^{\eta})g_{\varepsilon}(X_{s}^{\eta},s)\right]\textrm{d}\eta\textrm{d}s,
\end{align}
for the domains $D$ in Subsections \ref{PlateSection} and \ref{WedgeSection}. In fact the derivations in these sections imply that we can pass to the limit as $\varepsilon \to 0$ in every term except for the terms with $\theta_{\pm}(\cdot \, , t) \phi'' (\cdot \, /\varepsilon)$ for which we might not have the same representation in the limit. We therefore approximate the above representation by the following
\begin{align}
u^{i}(x,t) & =\int_{D}\mathbb{E}\left[1_{\left\{ t<\zeta(X^{\eta}\circ\tau_{t})\right\} }K^{i}(x,X_{t}^{\xi})\right]\omega_0(\xi)\textrm{d}\xi\nonumber +\int_{0}^{t}\int_{D}\mathbb{E}\left[1_{\{t-s<\zeta(X^{\xi}\circ\tau_{t})\}}K^{i}(x,X_{t}^{\xi})G(X_{s}^{\xi},s)\right]\textrm{d}\xi\textrm{d}s\\
 &+\frac{\nu}{\varepsilon^2}\int_{0}^{t}\int_{D}\mathbb{E}\left[1_{\{t-s<\zeta(X^{\xi}\circ\tau_{t})\}}K^{i}(x,X_{t}^{\xi})\theta_{+}(X_{s}^{\xi}, s) \phi'' (X_{s}^{\xi}/\varepsilon)\right]\textrm{d}\xi\textrm{d}s\nonumber\\
 &+\frac{\nu}{\varepsilon^2}\int_{0}^{t}\int_{D}\mathbb{E}\left[1_{\{t-s<\zeta(X^{\xi}\circ\tau_{t})\}}K^{i}(x,X_{t}^{\xi})\theta_{-}(X_{s}^{\xi}, s) \phi''(X_{s}^{\xi}/\varepsilon)\right]\textrm{d}\xi\textrm{d}s,  
\end{align}
with some small $\varepsilon$. Notice that for our choice of the cutoff function $\phi$ as in \eqref{CutoffPhi}, the second derivative $\phi''$ is supported on $[1/3, 2/3]$ and therefore the last two integrals are taken over a layer close to the boundary.

Since the kernel $K(x,y)$ is singular, we have to mollify it, e.g. using the cutoff function $f_{\delta}(y)=(1-\exp(-|y|/\delta))$. Notice that in the general expression for $K$ given in Lemma \ref{lem-e-003} we also have the derivatives $\dfrac{\partial T^i}{\partial y_j}$ which are singular in our case according to Lemmas \ref{PlateDerivativesOfT} and \ref{WedgeDerivativesOfT}. We therefore denote $K_{\delta}(x,y)$ the kernel obtained after desingularisation of $k^{\pm}(x,y)$ and $\dfrac{\partial T^i}{\partial y_j}$.

We first have to discretise the domain $D$ choosing lattices to represent points in $D$. 

\textit{Flat-plate.} For the domain $D = \mathbb{R}^2 \setminus \{(x_1,0): x_1 \geq 0\}$, we introduce the following lattices.
\begin{enumerate}
    \item For each boundary component $\partial D^{\pm}$, we define the corresponding thin boundary layer lattice $D_b^{\pm}$. Choose mesh sizes $h_1, h_2$ and numbers of points $N_1, N_2$, and define the following lattice points
    \begin{equation}\label{PlateBLayerLattice}
        x_{b \pm}^{i_1 i_2} = (i_1 h_1, \pm i_2 h_2),
    \end{equation}
    for $0 \leq i_1 \leq N_1$ and $0 \leq i_2 \leq N_2$. The number of points used in these layers is $2(N_1+1)(N_2+1)$.
    
    \item Outside the thin boundary layers, we introduce the outer layer lattice $D_o$. Choose a mesh size $h_0$ and a number $N_0$, and define
    \begin{equation}\label{PlateOLayerLattice}
        x_{o}^{i_1 i_2} = (i_1 h_0 , i_2 h_0),
    \end{equation}
    where $-N_0 \leq i_1 < 0, |i_2| \leq N_0$ and $0 \leq i_1 \leq N_0, 0 < |i_2| \leq N_0$. The number of points used in this layer is $\leq (2N_0+1)^2$.
\end{enumerate}

\textit{Wedge.} Let us also introduce the lattice points for the domain $D =  \{x: \arg x \in (\alpha, 2\pi - \alpha)\}$. 
\begin{enumerate}
    \item To introduce two thin boundary layer lattices $D_b^{\pm}$ for the corresponding boundary components $\partial D^{\pm}$, we choose mesh sizes $h_1, h_2$ and numbers of points $N_1, N_2$. Define the following lattice points
    \begin{equation}
        x_{b \pm}^{i_1 i_2} = (i_1 h_1 \cos \alpha - i_2 h_2 \sin \alpha, \pm i_1 h_1 \sin \alpha \pm i_2 h_2 \cos \alpha),
    \end{equation}
    for $0 \leq i_1 \leq N_1$ and $0 \leq i_2 \leq N_2$. The number of points used in these layers is $2(N_1+1)(N_2+1)$.
    
    \item We also introduce the outer layer lattice $D_o$ choosing a mesh size $h_0$ and a number $N_0$. We define the points
    \begin{equation}
        x_{o}^{i_1 i_2} = (i_1 h_0 , i_2 h_0),
    \end{equation}
    for $|i_1|, |i_2| \leq N_0$. This is a square lattice in which the domain $D$ is given by the condition: $x_{o}^{i_1 i_2} \in D$ if and only if $i_1 \leq 0$ or $\arg(i_1, i_2) \in (\alpha, 2\pi - \alpha)$. The number of points used in this layer to represent points in $D$ is $\leq (2N_0+1)^2$.
\end{enumerate}

We choose a time mesh size $h$ and denote $t_k = kh$ for $k \geq 0$. We initialise the processes $X^{i_1, i_2}_{b\pm; t_0} = x_{b\pm}^{i_1 i_2}$ and $X^{i_1, i_2}_{o; t_0} = x_{o}^{i_1 i_2}$ and update them for $k \geq 0$ according to
    \begin{equation}\label{UpdX}
       X^{i_1, i_2}_{t_{k+1}} = X^{i_1, i_2}_{t_k} + h \Hat{u}(X^{i_1, i_2}_{t_k}, t_k) + \sqrt{2\nu} (B_{t_{k+1}}-B_{t_k}),
    \end{equation}
    where we suppressed the subscripts $o$ and $b\pm$ for uniform notation. The drift $\Hat{u}$ is given by
    \begin{align}\label{HatURepr}
    \Hat{u}(x,t_{k+1}) &= \sum_{(i_1,i_2)\in D} A_{i_1,i_2} \omega_{i_1,i_2} \mathbb{E} \left[1_{\left\{ t_k<\zeta(X^{i_1, i_2}\circ\tau_{t_k})\right\}} K_{\delta}(x, X^{i_1, i_2}_{t_k}) \right] \nonumber \\
     &+ \sum_{(i_1,i_2)\in D} A_{i_1,i_2} \sum_{l=0}^k h G_{i_1, i_2;t_l} \mathbb{E} \left[1_{\{t_k-t_l<\zeta(X^{i_1, i_2}\circ\tau_{t_k})\}} K_{\delta}(x, X^{i_1, i_2}_{t_l}) \right] \nonumber \\ 
     &+ \frac{\nu}{\varepsilon^2} \sum_{(i_1,i_2)\in D_b^{-}} h_1 h_2 \sum_{l=0}^k h \mathbb{E} \left[ 1_{\{t_k-t_l<\zeta(X^{i_1, i_2}\circ\tau_{t_k})\}} K_{\delta}(x, X^{i_1, i_2}_{t_k})\theta_{-}(X^{i_1, i_2}_{t_l}, t_l) \phi''(X^{i_1, i_2}_{t_l}/\varepsilon) \right] \nonumber \\ 
     &+ \frac{\nu}{\varepsilon^2} \sum_{(i_1,i_2)\in D_b^{+}} h_1 h_2 \sum_{l=0}^k h \mathbb{E} \left[ 1_{\{t_k-t_l<\zeta(X^{i_1, i_2}\circ\tau_{t_k})\}} K_{\delta}(x, X^{i_1, i_2}_{t_k})\theta_{+}(X^{i_1, i_2}_{t_l}, t_l) \phi''(X^{i_1, i_2}_{t_l}/\varepsilon) \right], 
    \end{align}
    for $x \in D$, and $\Hat{u} = 0$ otherwise. Notice that in the formula above, sums over $(i_1,i_2)\in D$ denote both sums over the boundary layer lattices $D_b^{\pm}$ and the outer layer lattice $D_o$. However, the last two summations are over boundary lattices $D_b^{\pm}$ due to the remark we made above regarding the support of the cutoff function. Thus, 
    \begin{align}\label{a-d032}
    A_{i_1,i_2} &= h_1 h_2 \text{ or } h_0^2, \nonumber \\
    \omega_{i_1,i_2} &= \omega(x_{b \pm}^{i_1 i_2}, 0) \text{ or } \omega(x_{o}^{i_1 i_2}, 0), \nonumber \\
    G_{i_1,i_2; t_l} &= G(x_{b \pm}^{i_1 i_2}, t_l) \text{ or } G(x_{o}^{i_1 i_2}, t_l),
    \end{align}
    for boundary and outer layers respectively. 
    
    Since at every step we have to compute the boundary stress $\theta_{\pm}$, we do this computation, according to \eqref{ThetaDefPlate} and 
    \eqref{ThetaDefWedge}, by applying the corresponding derivatives to the formula \eqref{HatURepr}. Indeed, this gives valid approximations for the derivatives of the velocity $u$ as we work with the mollified kernel $K_{\delta}(x,y)$. The expression for the derivatives is given in terms of a similar to \eqref{HatURepr} formula where we replace the kernel $K_{\delta}$ by the corresponding derivatives. We use the formulae for the derivatives $\dfrac{\partial k^{-}}{\partial x}$ and $\dfrac{\partial k^{+}}{\partial x}$ given in Lemma \ref{DerivativeK_Lemma} which we similarly regularise using the cutoff function $f_{\delta}$.

    To handle the expectations in the above representation, we propose the following numerical schemes.

    \textit{Numerical scheme 1.}
    In this numerical scheme, we omit the expectations in 
    \eqref{HatURepr} and run independent Brownian motions in \eqref{UpdX}. Therefore, we update the diffusions $X_{t_k}^{i_1, i_2}$, starting at $x^{i_1 i_2}$ when $k = 0$, according to
    \begin{equation}\label{NS1X}
       X^{i_1, i_2}_{t_{k+1}} = X^{i_1, i_2}_{t_k} + h \Hat{u}(X^{i_1, i_2}_{t_k}, t_k) + \sqrt{2\nu} (B_{t_{k+1}}^{i_1, i_2}-B_{t_k}^{i_1, i_2}),
    \end{equation}
    for $k \geq 0$, where $B^{i_1, i_2}$ are independent Brownian motions. Here 
    \begin{align}\label{NS1U}
    \Hat{u}(x,t_{k+1}) &= \sum_{(i_1,i_2)\in D} A_{i_1,i_2} \omega_{i_1,i_2} K_{\delta}(x, X^{i_1, i_2}_{t_k}) 1_{\{ t_k<\zeta(X^{i_1, i_2}\circ\tau_{t_k})\}} \nonumber \\
     &+ \sum_{(i_1,i_2)\in D} A_{i_1,i_2} h K_{\delta}(x, X^{i_1, i_2}_{t_k}) \sum_{l=0}^k G_{i_1, i_2;t_l} 1_{\{t_k-t_l<\zeta(X^{i_1, i_2}\circ\tau_{t_k})\}} \nonumber \\ 
     &+ \frac{\nu}{\varepsilon^2} \sum_{(i_1,i_2)\in D_b^{-}} h_1 h_2 h K_{\delta}(x, X^{i_1, i_2}_{t_k}) \sum_{l=0}^k \theta_{-}(X^{i_1, i_2}_{t_l}, t_l) \phi''(X^{i_1, i_2}_{t_l}/\varepsilon) 1_{\{t_k-t_l<\zeta(X^{i_1, i_2}\circ\tau_{t_k})\}} \nonumber \\ 
     &+ \frac{\nu}{\varepsilon^2} \sum_{(i_1,i_2)\in D_b^{+}} h_1 h_2 h K_{\delta}(x, X^{i_1, i_2}_{t_k}) \sum_{l=0}^k \theta_{+}(X^{i_1, i_2}_{t_l}, t_l) \phi''(X^{i_1, i_2}_{t_l}/\varepsilon) 1_{\{t_k-t_l<\zeta(X^{i_1, i_2}\circ\tau_{t_k})\}}, 
    \end{align}
    for $x \in D$, and $\Hat{u}=0$ otherwise with $A_{i_1,i_2}$, $ \omega_{i_1,i_2}$, $G_{i_1, i_2; t_l}$ given in \eqref{a-d032}. 

    \textit{Numerical scheme 2}
    In this scheme, we replace expectations in \eqref{HatURepr} by averages with independent Brownian motions. Thus, we start the processes $X_{t_k}^{m; i_1, i_2}$ at $x^{i_1 i_2}$ when $k = 0$ and update them according to
    \begin{equation}\label{NS2X}
       X^{m; i_1, i_2}_{t_{k+1}} = X^{m; i_1, i_2}_{t_k} + h \Hat{u}(X^{m; i_1, i_2}_{t_k}, t_k) + \sqrt{2\nu} (B_{t_{k+1}}^{m}-B_{t_k}^{m}),
    \end{equation}
    for $k \geq 0$, where $B^{m}$ are independent Brownian motions for $m = 1, \ldots, N$, and
    \begin{align}\label{NS2U}
    \Hat{u}(x,t_{k+1}) &= \sum_{(i_1,i_2)\in D} A_{i_1,i_2} \omega_{i_1,i_2} \frac{1}{N} \sum_{m=1}^{N} K_{\delta}(x, X^{m; i_1, i_2}_{t_k}) 1_{\{ t_k<\zeta(X^{m; i_1, i_2}\circ\tau_{t_k})\}} \nonumber \\
     &+ \sum_{(i_1,i_2)\in D} A_{i_1,i_2} h \frac{1}{N} \sum_{m=1}^{N} K_{\delta}(x, X^{m; i_1, i_2}_{t_k}) \sum_{l=0}^k G_{i_1, i_2;t_l} 1_{\{t_k-t_l<\zeta(X^{m;i_1, i_2}\circ\tau_{t_k})\}} \nonumber \\ 
     &+ \frac{\nu}{\varepsilon^2} \sum_{(i_1,i_2)\in D_b^{-}} h_1 h_2 h \frac{1}{N} \sum_{m=1}^{N} K_{\delta}(x, X^{m; i_1, i_2}_{t_k}) \sum_{l=0}^k \theta_{-}(X^{m; i_1, i_2}_{t_l}, t_l) \phi''(X^{m; i_1, i_2}_{t_l}/\varepsilon) 1_{\{t_k-t_l<\zeta(X^{m; i_1, i_2}\circ\tau_{t_k})\}} \nonumber \\ 
     &+ \frac{\nu}{\varepsilon^2} \sum_{(i_1,i_2)\in D_b^{+}} h_1 h_2 h \frac{1}{N} \sum_{m=1}^{N} K_{\delta}(x, X^{m; i_1, i_2}_{t_k}) \sum_{l=0}^k \theta_{+}(X^{m; i_1, i_2}_{t_l}, t_l) \phi''(X^{m; i_1, i_2}_{t_l}/\varepsilon) 1_{\{t_k-t_l<\zeta(X^{m; i_1, i_2}\circ\tau_{t_k})\}}, 
    \end{align}
    for $x \in D$, and $\Hat{u}=0$ otherwise. 

    Note that in described numerical schemes, the stopping times $\zeta(X^{i_1, i_2})$ denote the first boundary hitting times for processes $X^{i_1, i_2}$ (or $X^{m;i_1, i_2}$ for the second scheme), while $\zeta(X^{i_1, i_2}\circ\tau_{t_k})$ are seen as the last hitting times. We therefore check if $X^{i_1, i_2}_{t_{k+1}}$ crosses the boundary $\partial D$ at every step when updating it according to \eqref{NS1X} and \eqref{NS2X}. Thus, after the first boundary crossing by $X^{i_1, i_2}$ we do not count the corresponding term in the first sum in \eqref{NS1U} and \eqref{NS2U}. 

    Notice also that in \eqref{NS1U} and \eqref{NS2U} we separate the sums that track the history of processes $X^{i_1, i_2}$, i.e.
    \begin{equation*}
        \sum_{l=0}^k G_{i_1, i_2;t_l} 1_{\{t_k-t_l<\zeta(X^{i_1, i_2}\circ\tau_{t_k})\}},
    \end{equation*}
    and
    \begin{equation*}
        \sum_{l=0}^k \theta_{\pm}(X^{i_1, i_2}_{t_l}, t_l) \phi''(X^{i_1, i_2}_{t_l}/\varepsilon) 1_{\{t_k-t_l<\zeta(X^{i_1, i_2}\circ\tau_{t_k})\}},
    \end{equation*}
    as well as the corresponding sums with $X^{m;i_1, i_2}$ for the second scheme. Thus, we do not have to compute these sums at every step, instead we store the whole sums and update them at each step adding new summands. When we have a boundary crossing by $X^{i_1, i_2}$ or $X^{m;i_1, i_2}$, we set the stored sum to zero and start updating it anew. 

\subsection{Numerical experiments}

    In this subsection, we use the numerical schemes described above to conduct numerical experiments for certain flows. We exploit the Numerical Scheme 1 which is summarised in Algorithm \ref{Alg1} below.

\begin{algorithm}
\caption{Random vortex algorithm.}\label{Alg1}
\begin{algorithmic}
\State \textbf{Define:} $\omega_{i_1, i_2}$ and $G_{i_1, i_2; t_l}$
\State \textbf{Initialise:} $X^{i_1, i_2} \gets x^{i_1 i_2}$ 
\State \textbf{Initialise:} $\text{ind}^{i_1, i_2} \gets 1$ \Comment{indicators of $X^{i_1, i_2}$ hitting the boundary}
\State \textbf{Compute:} $S^{i_1, i_2} \gets \theta_{\pm}(X^{i_1, i_2}_{b\pm}, t_0) \phi''(X^{i_1, i_2}_{b\pm}/\varepsilon)$
\For{$k \gets 0$ to $N$}
\State \textbf{Compute:} $\Delta^{i_1, i_2} \gets h \Hat{u}(X^{i_1, i_2}, t_k) + \sqrt{2\nu} (B_{t_{k+1}}^{i_1, i_2}-B_{t_k}^{i_1, i_2})$
\State \textbf{Update:} $X^{i_1, i_2} \gets X^{i_1, i_2} + \Delta^{i_1, i_2} \cdot \text{ind}^{i_1, i_2}$
    \If{$X^{i_1, i_2}$ crosses $\partial D$}
        \State \textbf{Set:} $X^{i_1, i_2} \gets$ the crossing point
        \State \textbf{Set:} $\text{ind}^{i_1, i_2} \gets 0$ 
        \State \textbf{Set:} $S^{i_1, i_2} \gets 0$
    \EndIf
\State \textbf{Compute:} $\theta_{\pm}(X^{i_1, i_2}_{b\pm}, t_k)$
\State \textbf{Update:} $S^{i_1, i_2} \gets S^{i_1, i_2} + \theta_{\pm}(X^{i_1, i_2}_{b\pm}, t_k) \phi''(X^{i_1, i_2}_{b\pm}/\varepsilon)$
\EndFor
\end{algorithmic}
\end{algorithm}

    \textit{Experiment 1.} We consider the flow past the flat-plate, i.e. $D=\mathbb{R}^2 \setminus \{(x_1,0):x_1\geq0\}$. The lattice points are given by 
    \eqref{PlateBLayerLattice} and 
    \eqref{PlateOLayerLattice} with $N_0=15, N_1=30,N_2=45$. Therefore, we use 961 points in the outer layer and 1,426 points for each of the two thin boundary layers. With $H = 6$ and $H_0=0.1$ determining the size of the domain and the thickness of the thin layers respectively, we find the mesh constants $h_0=\frac{H}{N_0} = 0.4, h_1=\frac{H}{N_1} = 0.2, h_2=\frac{H_0}{N_2} \approx 0.0022$. We set the parameter $\varepsilon = 0.05$ for the boundary sums in this simulation.
    
    We initialise the vector field $u_0(x_1, x_2) = (U_0 |x_2| (6-|x_2|) / 9, 0)$ for $-H \leq x_1, x_2 \leq H$, where $U_0 = \frac{\nu Re}{L} = 125$ for $\nu = 0.1, L = 6$ and $Re = 7500$. Thus, the initial vorticity $\omega_0(x_1, x_2) = - U_0 (-2|x_2|+6)/9 \cdot \textrm{sgn}(x_{2})$, where for $x_2 = 0$ and $x_1 < 0$ we assume $\omega_0(x_1,x_2) = 0$, and the initial boundary stress $\theta_{+}(x_1,0) = - 2U_0 / 3$ and $\theta_{-}(x_1,0) = 2U_0 / 3$. The external force $G$ is taken to be identically zero. 
    
    We conduct our simulation with time steps $h=0.01$, the results for several times $t$ are shown in Figures \ref{Exp1FigOFlow}, \ref{Exp1FigUpperBFlow} and \ref{Exp1FigLowerBFlow}. In these figures, the streamlines are coloured by the velocity magnitude while the background is coloured by the vorticity value. We also plot the boundary stress $\theta_{+}$ and $\theta_{-}$ as functions of position at boundary in Figures \ref{Exp1FigLowerBStress} and \ref{Exp1FigUpperBStress}.

\textit{Experiment 2.} In this experiment, the initial velocity is again taken to be $u_0(x_1, x_2) = (U_0 |x_2| (6-|x_2|) / 9, 0)$ and $G\equiv 0$. The lattice points and other parameters are taken as before, however, we alter the representation for $\Hat{u}$. We ignore the boundary sums in this case, that is we compute
\begin{equation*}
    \Hat{u}(x,t_{k+1}) = \sum_{(i_1,i_2)\in D} A_{i_1,i_2} \omega_{i_1,i_2} K_{\delta}(x, X^{i_1, i_2}_{t_k}) 1_{\{ t_k<\zeta(X^{i_1, i_2}\circ\tau_{t_k})\}}.
\end{equation*}
Therefore, we have to use a variant of Algorithm \ref{Alg1} with the above representation for $\Hat{u}$. The results are shown in Figures \ref{Exp2FigOFlow}, \ref{Exp2FigUpperBFlow} and \ref{Exp2FigLowerBFlow} for the outer flow and boundary flows respectively, and Figures \ref{Exp2FigLowerBStress} and Figures \ref{Exp2FigUpperBStress} for boundary stress values. 

We notice that this version of the numerical scheme still allows to capture some boundary phenomena as the boundary flows display certain chaotic behaviour. However, at large times $t$ one has large regions close to the boundary where the velocity does not change direction. We conclude therefore that for simulating the boundary flows it is essential to use the boundary sums in the velocity representation as they contribute to the chaotic turbulent motion.

\textit{Experiment 3.} We conduct this experiment for the wedge domain $D=\mathbb{R}^{2}\setminus\left\{x:-\alpha\leq\arg x\leq\alpha\right\}$ with $\alpha = \frac{\pi}{4}$. The size of the domain and numbers of lattice points are taken as before, i.e. $H=6, H_0=0.1$ and $N_0=15, N_1=30,N_2=45$. Notice that in this case $h_1=\frac{\sqrt{2} H}{N_1} \approx 0.2828$ as the length of the boundary components contained in the domain is $\sqrt{2} H$.

We take the initial velocity $u_0(x_1, x_2) = 0$, and the external force term $G(x_1,x_2)=2U_0 \frac{x_2}{H^2}$ for $U_0 = 500$ independent of time. For this simulation, we choose smaller time steps $h = 0.001$. The results of the simulation are shown in Figures \ref{Exp3FigOFlow}, \ref{Exp3FigUpperBFlow}, \ref{Exp3FigLowerBFlow} --- notice that in the figures, the upper and lower boundary flows are rotated by $-\frac{\pi}{4}$ and $\frac{\pi}{4}$ respectively. Also, see Figures \ref{Exp3FigLowerBStress}, \ref{Exp3FigUpperBStress} for boundary stress values.

\begin{figure}
    \centering
    \subfloat[\centering $t=0.01$]{{\includegraphics[width=.5\linewidth]{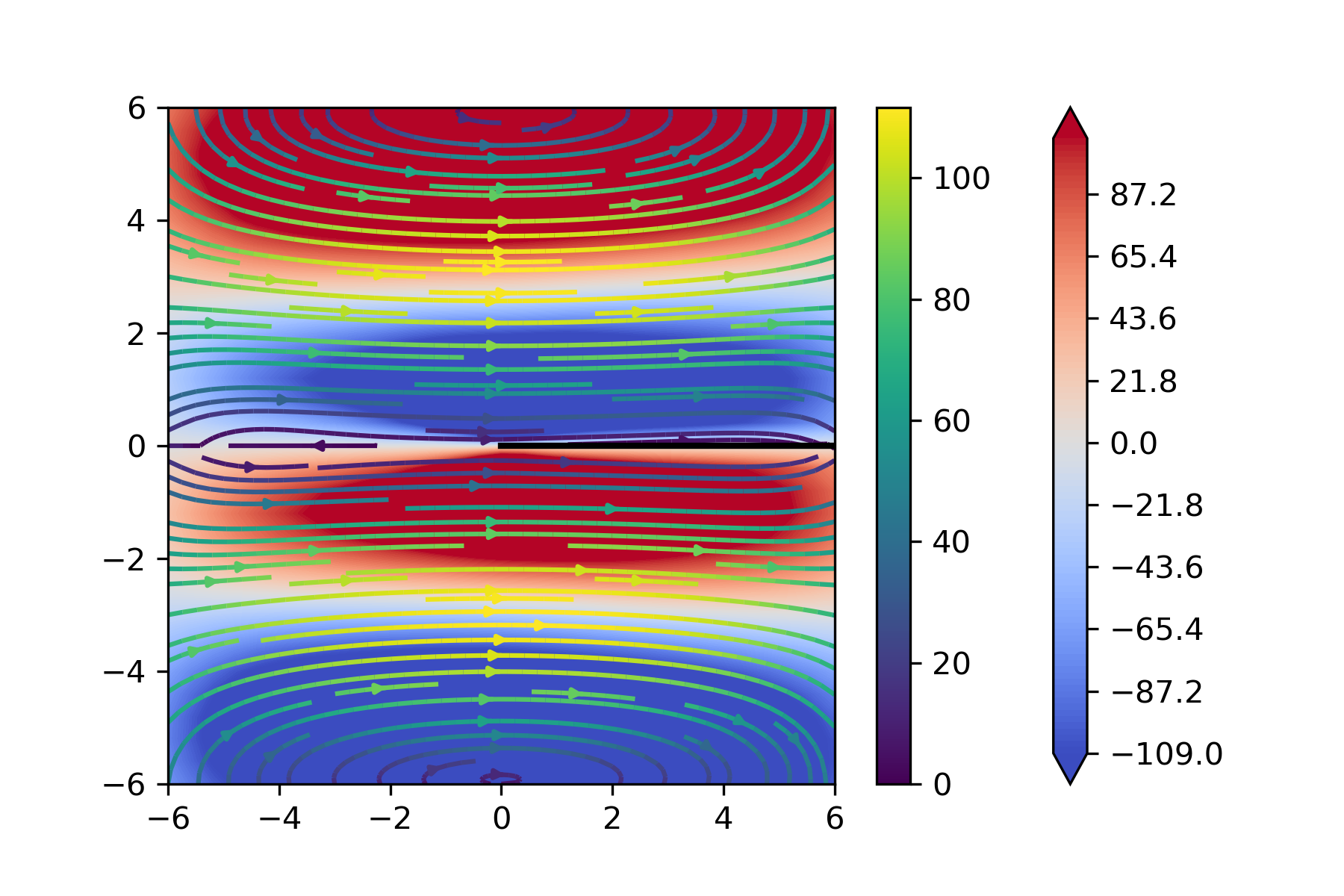} }}
    \subfloat[\centering $t=0.25$]{{\includegraphics[width=.5\linewidth]{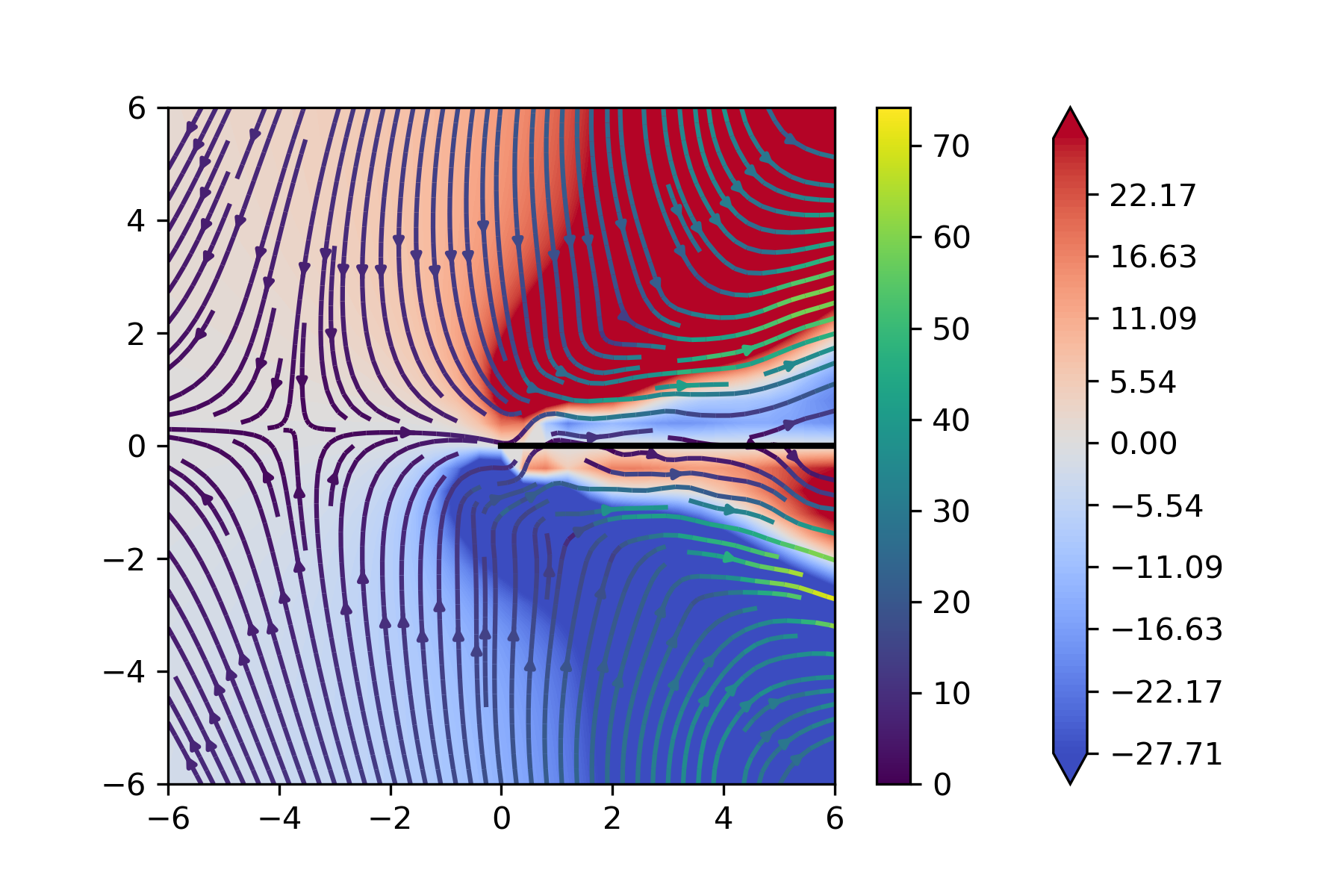}}}
    \qquad
    \subfloat[\centering $t=0.5$]{{\includegraphics[width=.5\linewidth]{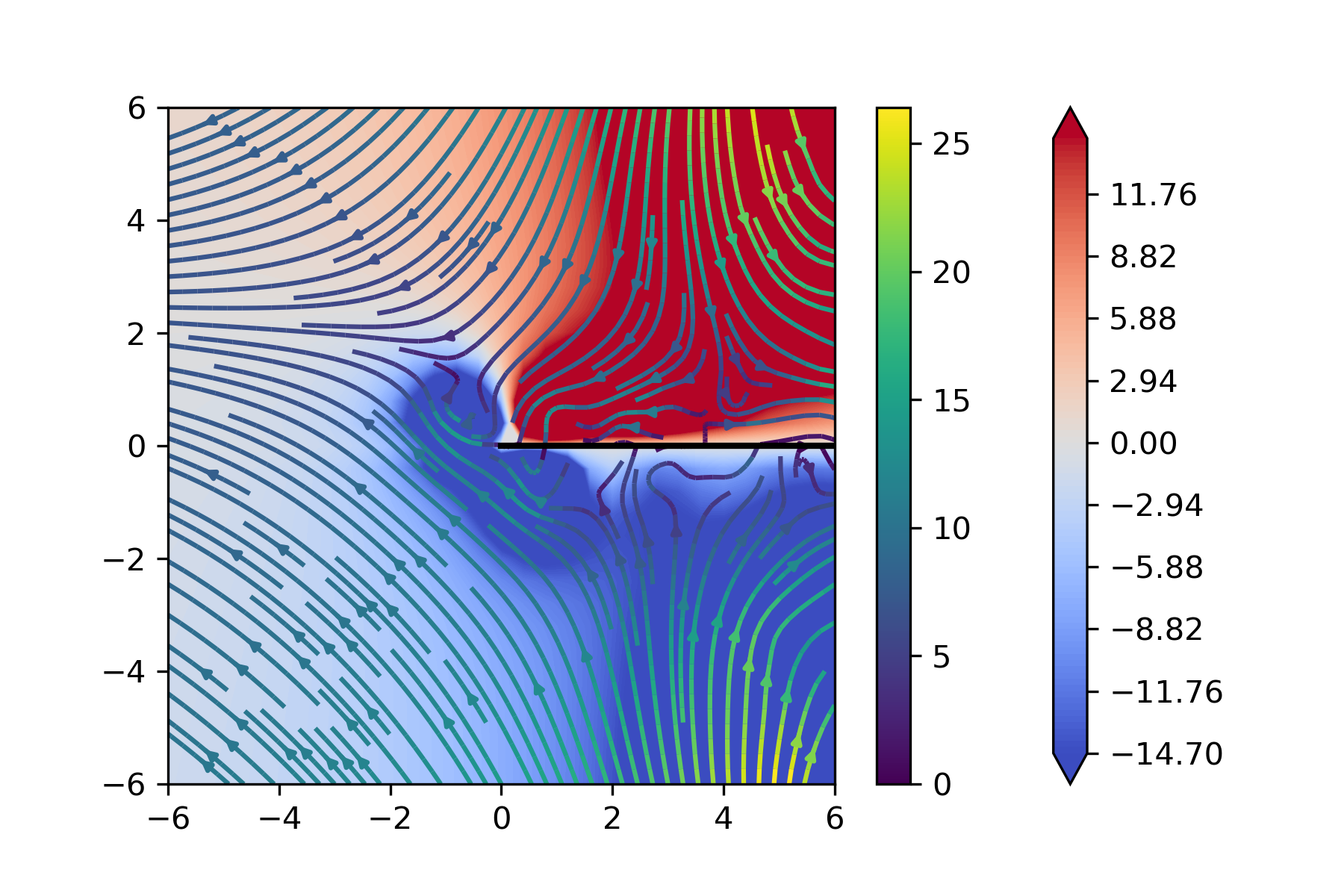} }}
    \subfloat[\centering $t=1.0$]{{\includegraphics[width=.5\linewidth]{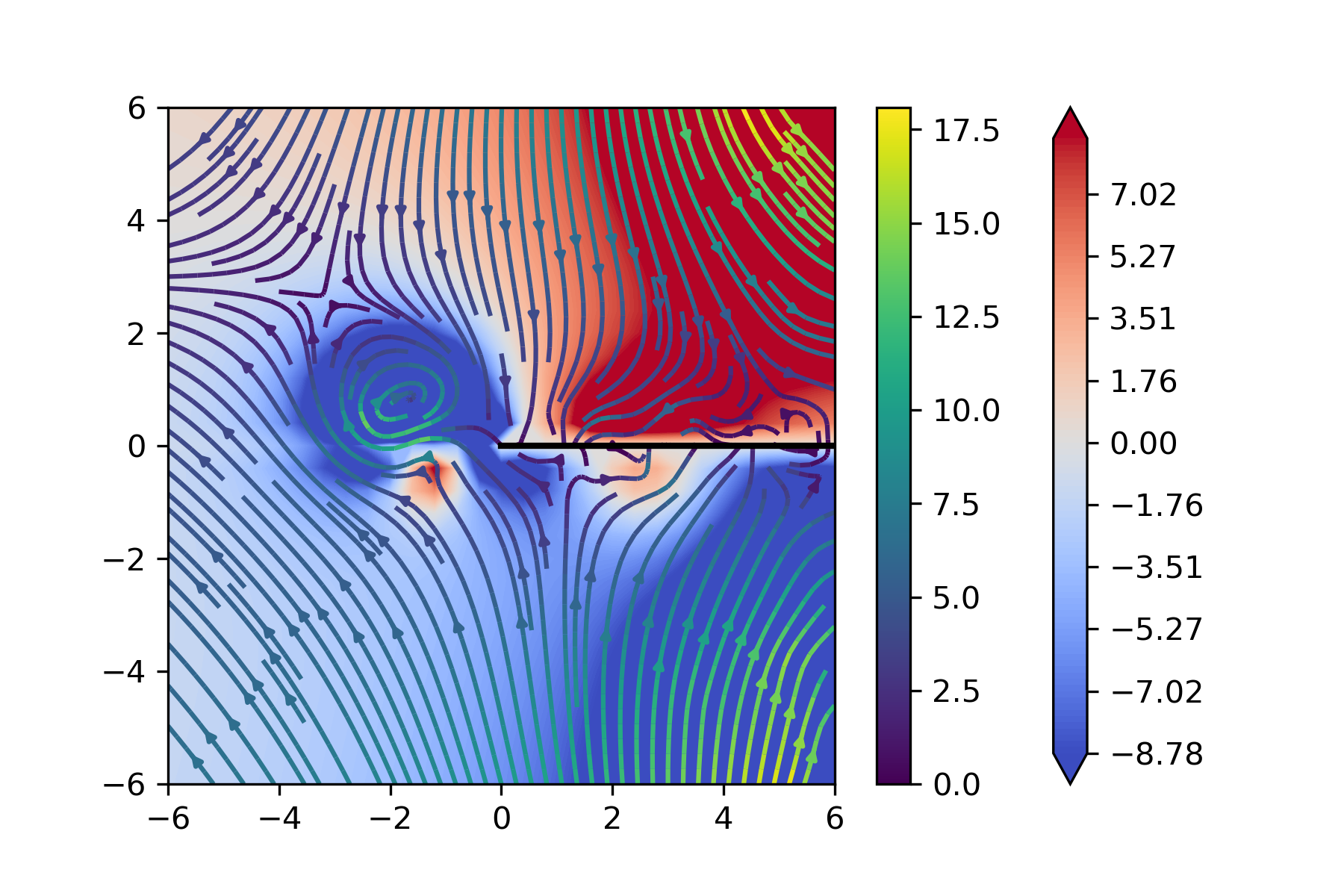}}}
    \caption{The outer layer flow at different times $t$.}
    \label{Exp1FigOFlow}
\end{figure}

\begin{figure}
    \centering
    \subfloat[\centering $t=0.01$]{{\includegraphics[width=.5\linewidth]{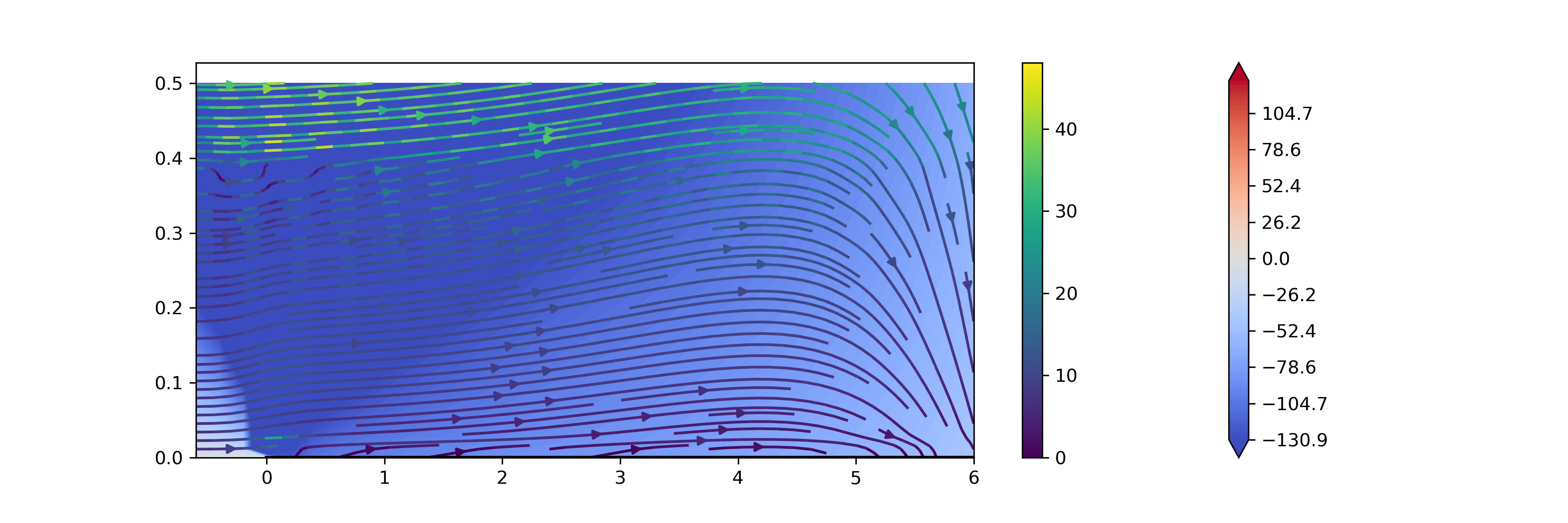} }}
    \subfloat[\centering $t=0.25$]{{\includegraphics[width=.5\linewidth]{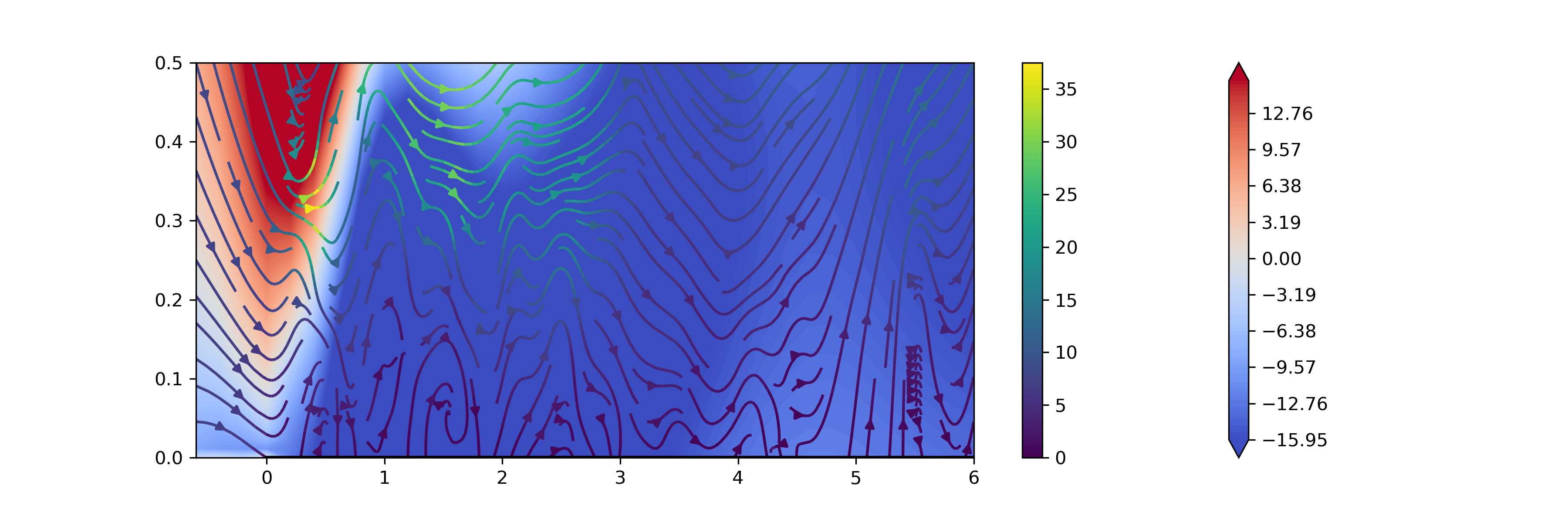}}}
    \qquad
    \subfloat[\centering $t=0.5$]{{\includegraphics[width=.5\linewidth]{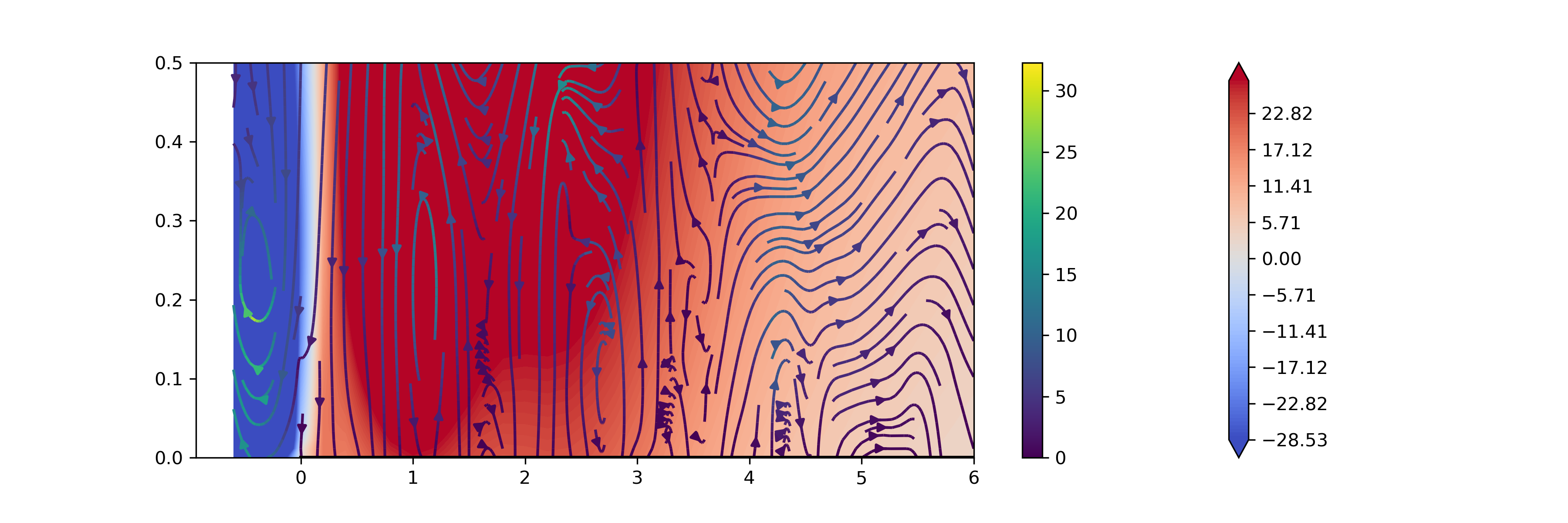} }}
    \subfloat[\centering $t=1.0$]{{\includegraphics[width=.5\linewidth]{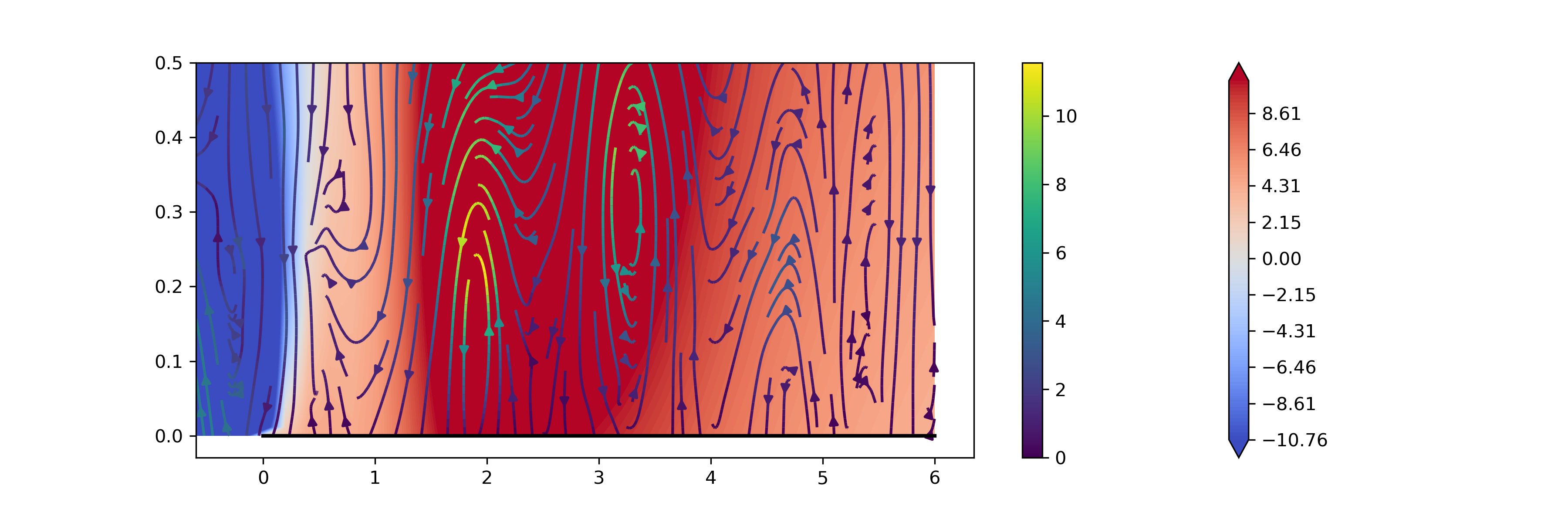}}}
    \caption{The upper boundary layer flow at different times $t$.}
    \label{Exp1FigUpperBFlow}
\end{figure}

\begin{figure}
    \centering
    \subfloat[\centering $t=0.01$]{{\includegraphics[width=.5\linewidth]{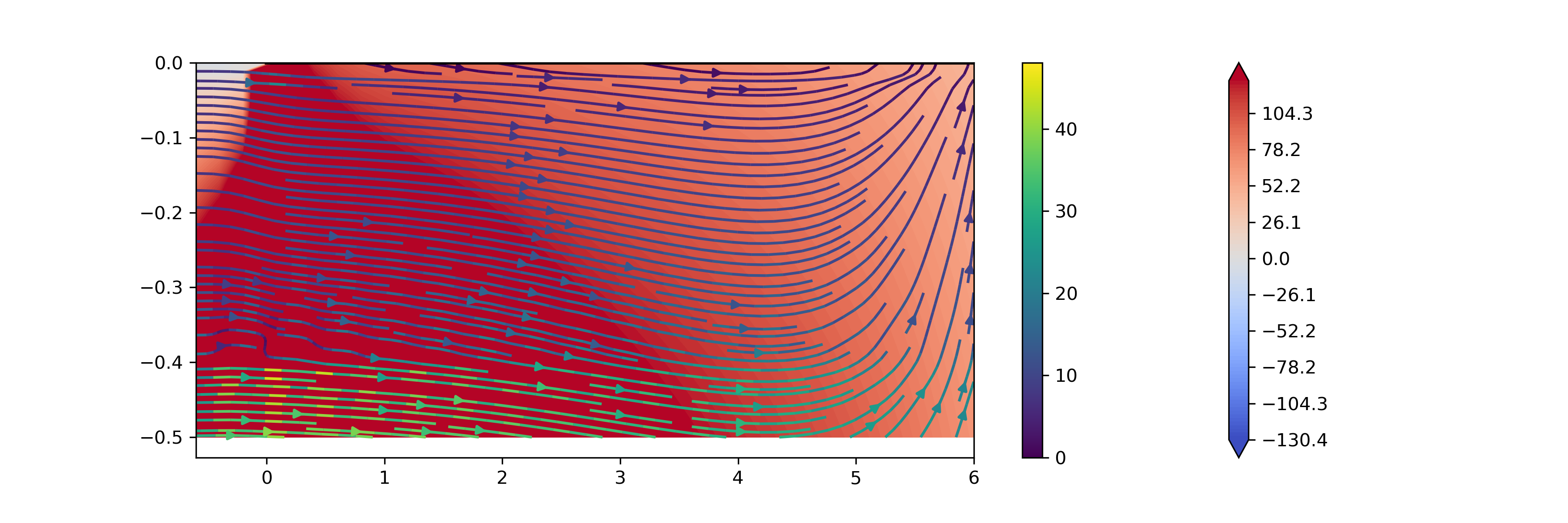} }}
    \subfloat[\centering $t=0.25$]{{\includegraphics[width=.5\linewidth]{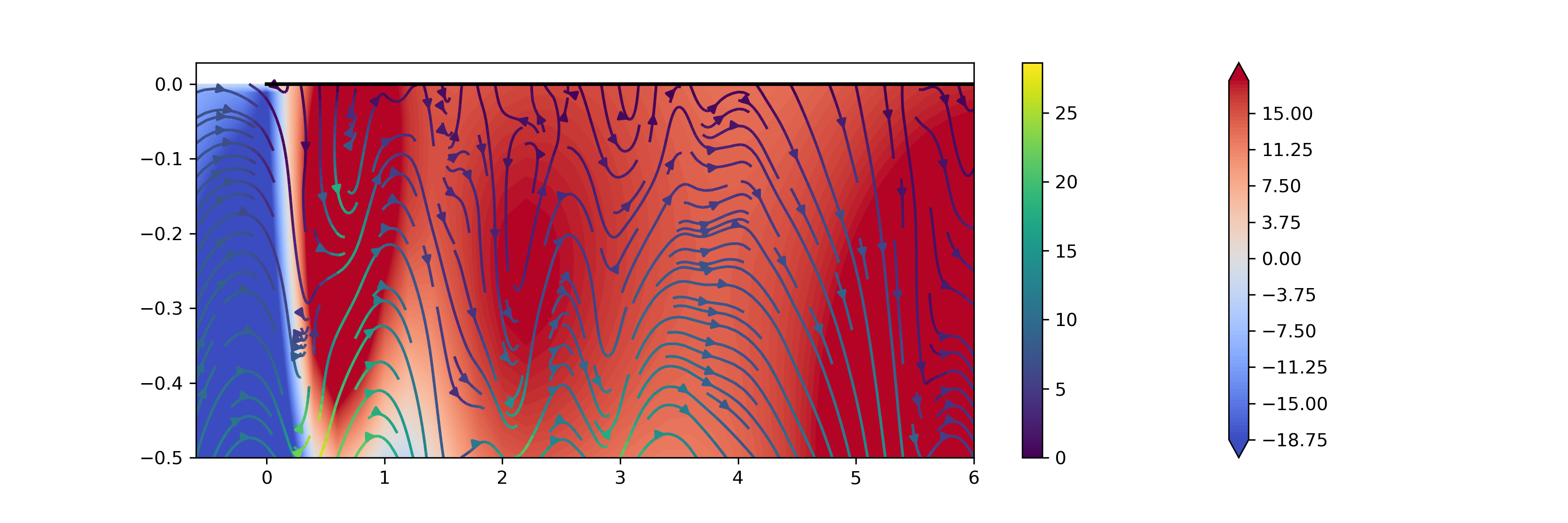}}}
    \qquad
    \subfloat[\centering $t=0.5$]{{\includegraphics[width=.5\linewidth]{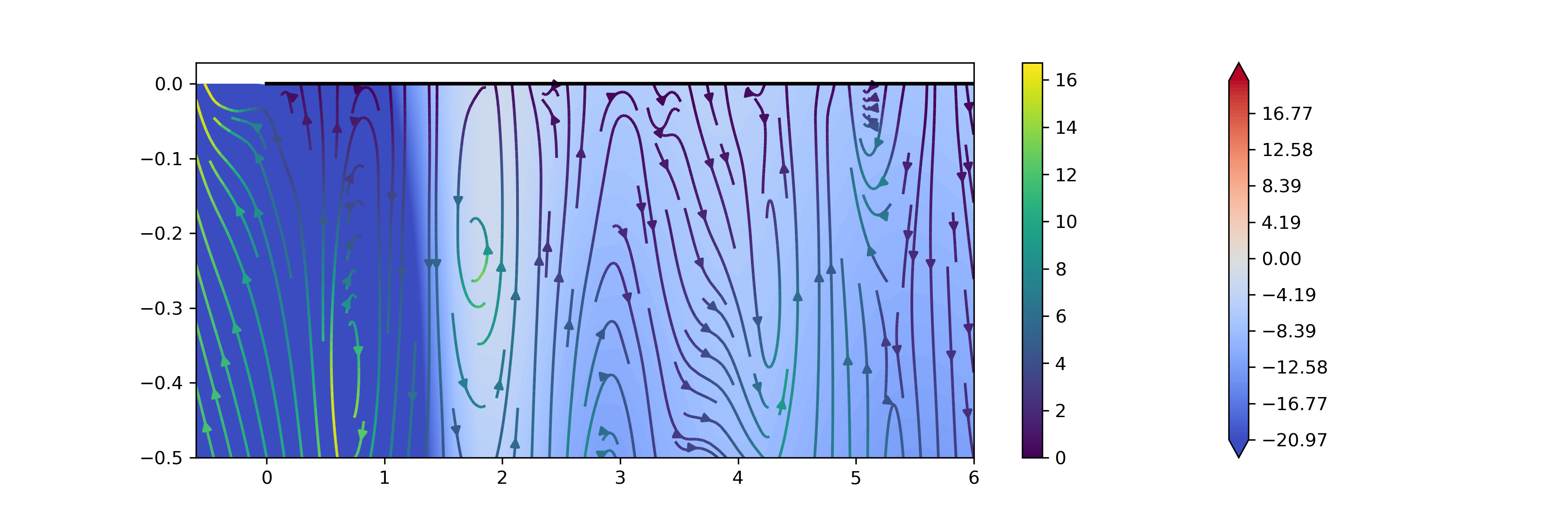} }}
    \subfloat[\centering $t=1.0$]{{\includegraphics[width=.5\linewidth]{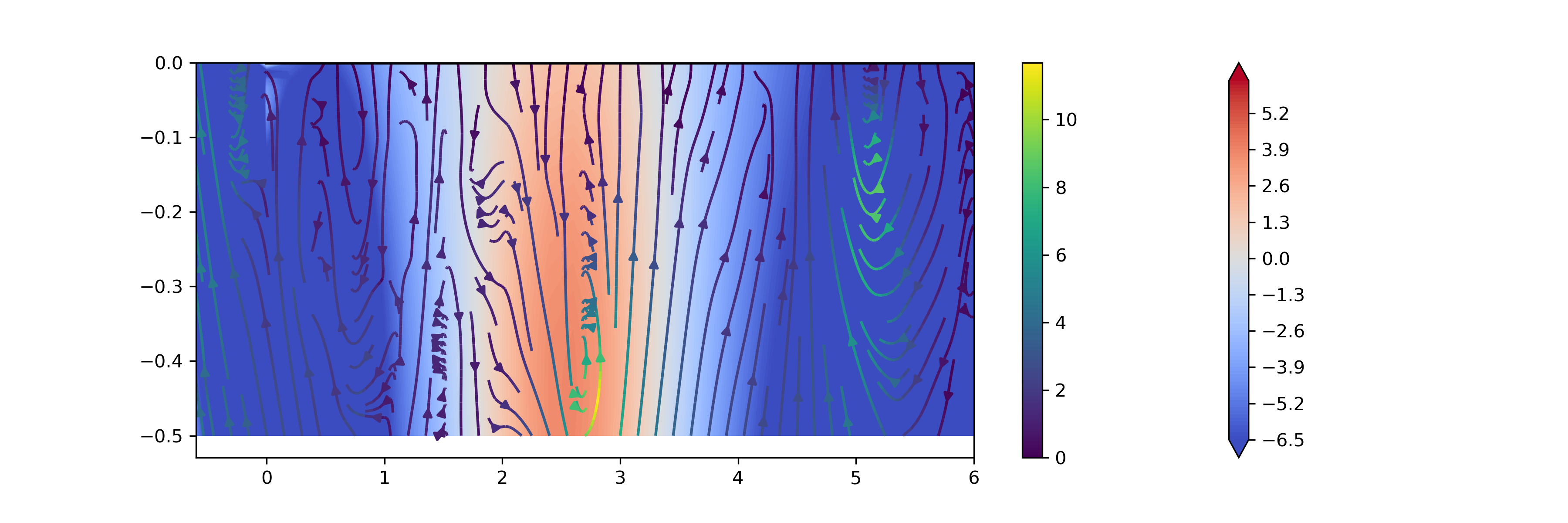}}}
    \caption{The lower boundary layer flow at different times $t$.}
    \label{Exp1FigLowerBFlow}
\end{figure}

\begin{figure}
    \centering
    \subfloat[\centering $t=0.01$]{{\includegraphics[width=.4\linewidth]{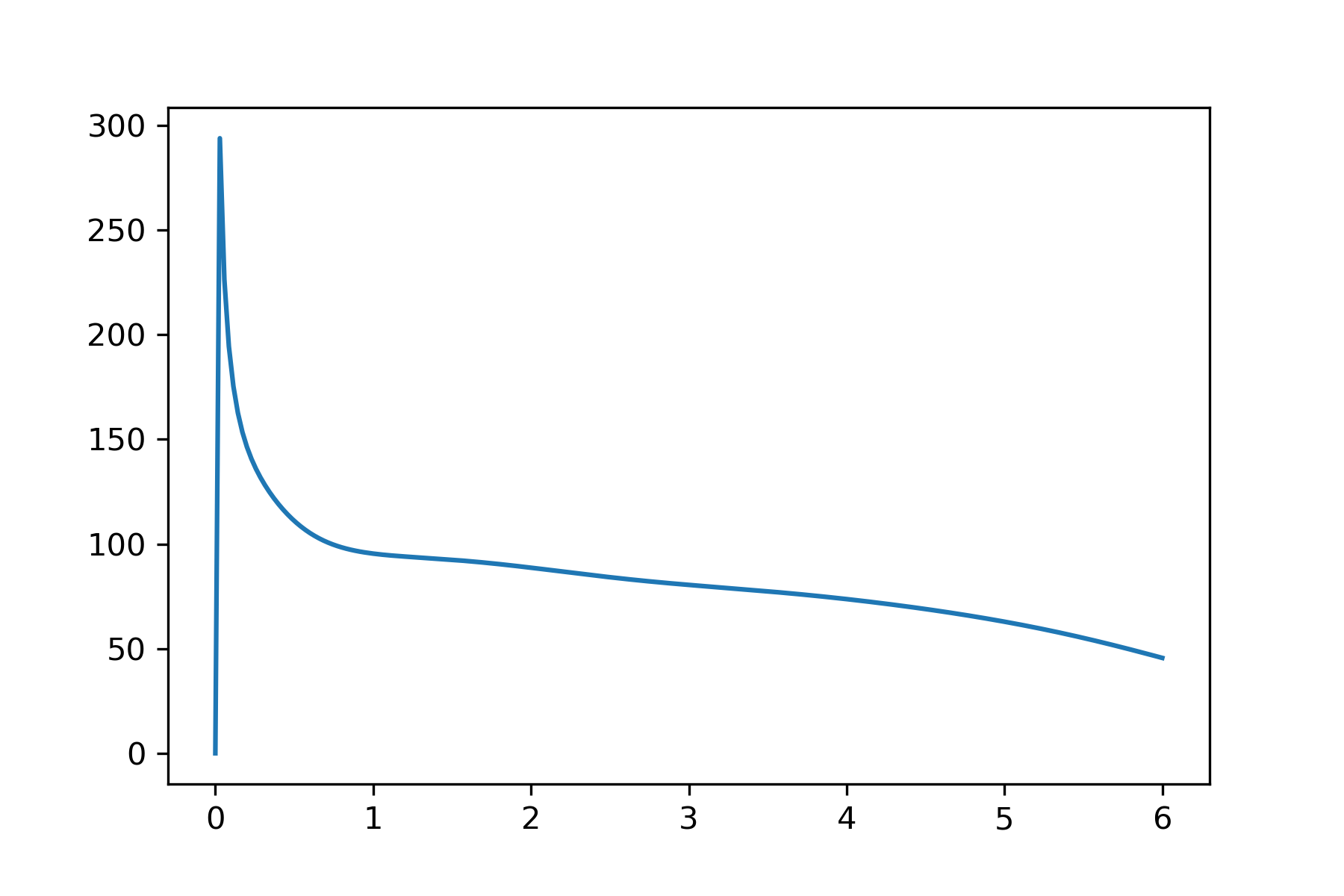} }}
    \subfloat[\centering $t=0.25$]{{\includegraphics[width=.4\linewidth]{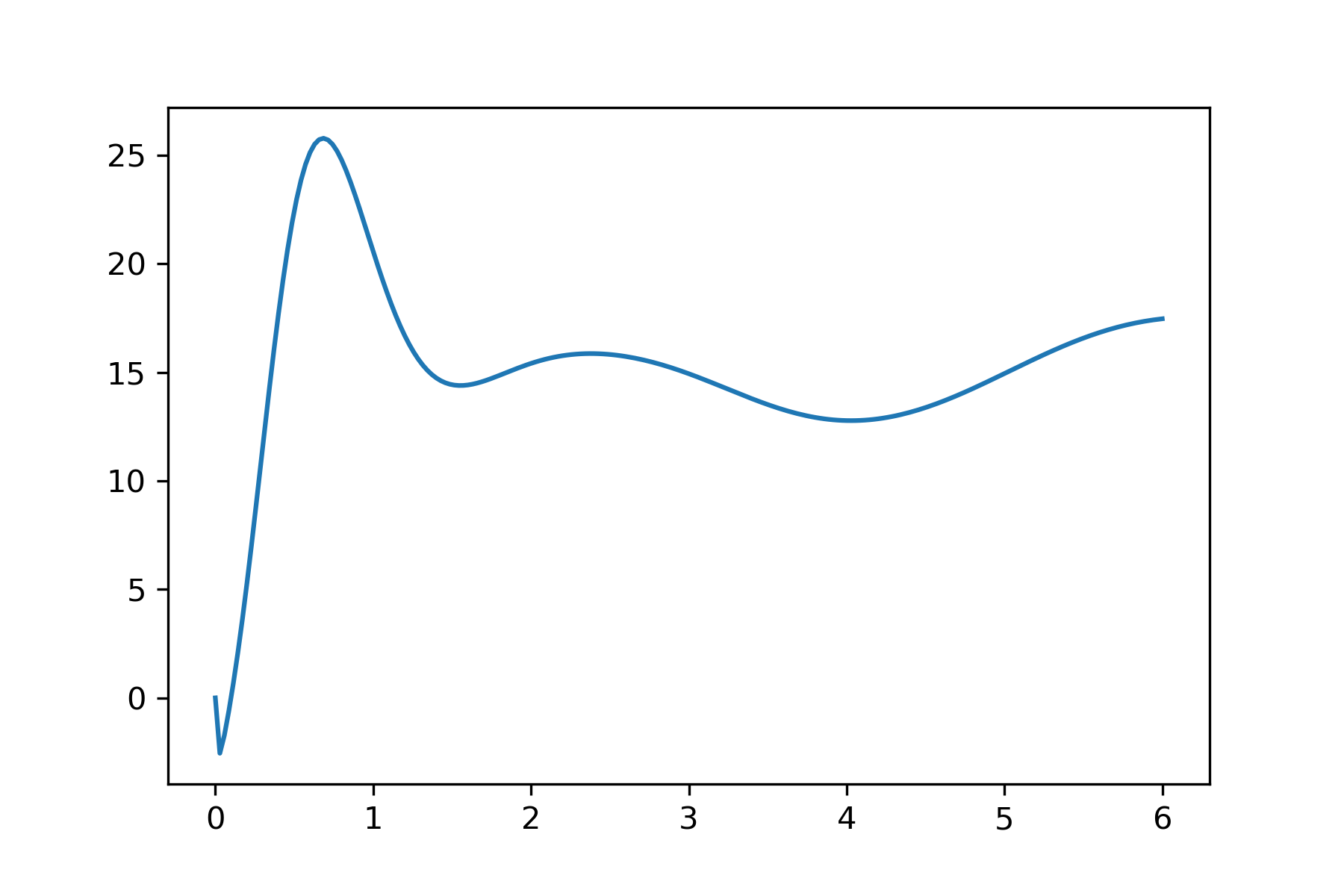}}}
    \qquad
    \subfloat[\centering $t=0.5$]{{\includegraphics[width=.4\linewidth]{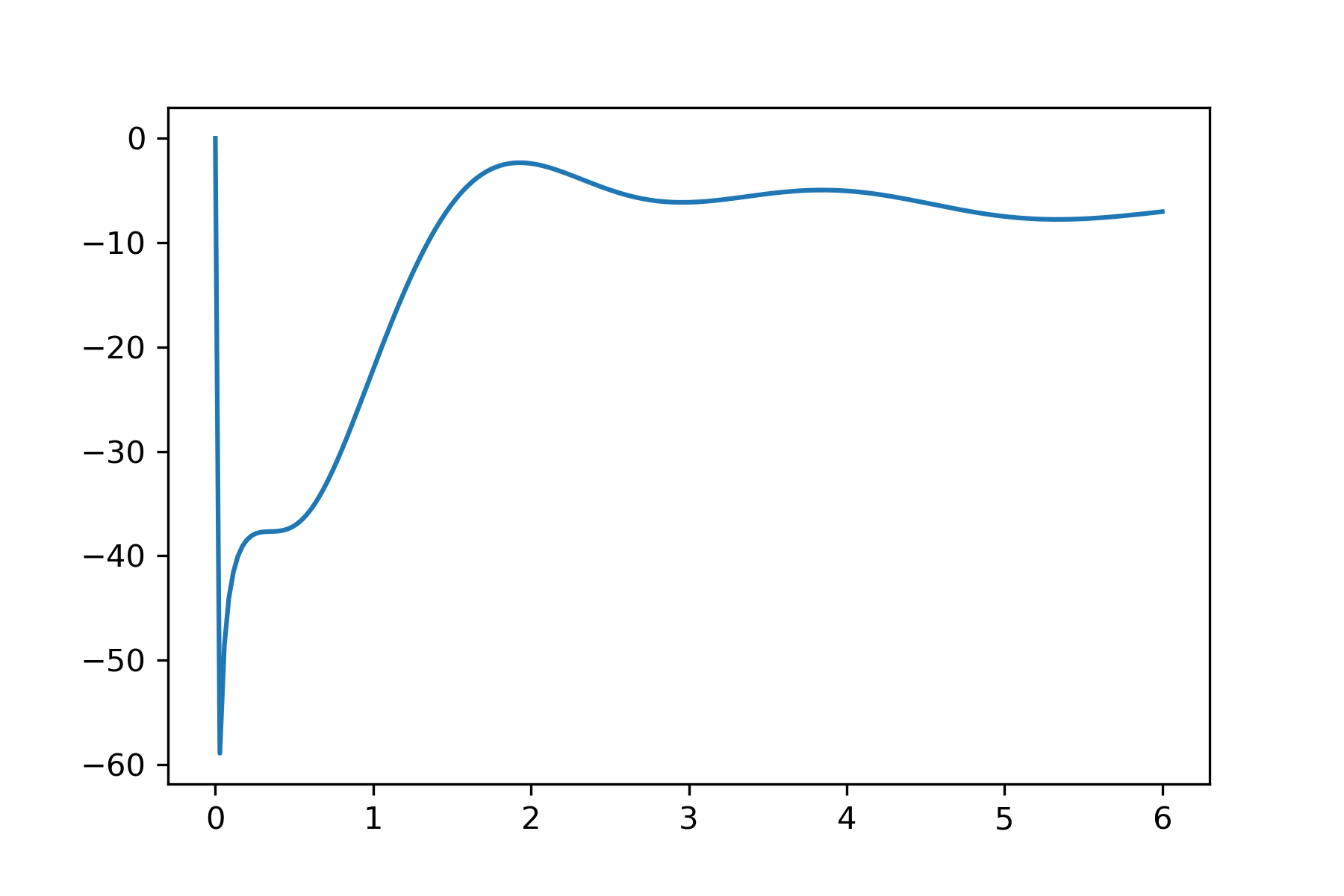} }}
    \subfloat[\centering $t=1.0$]{{\includegraphics[width=.4\linewidth]{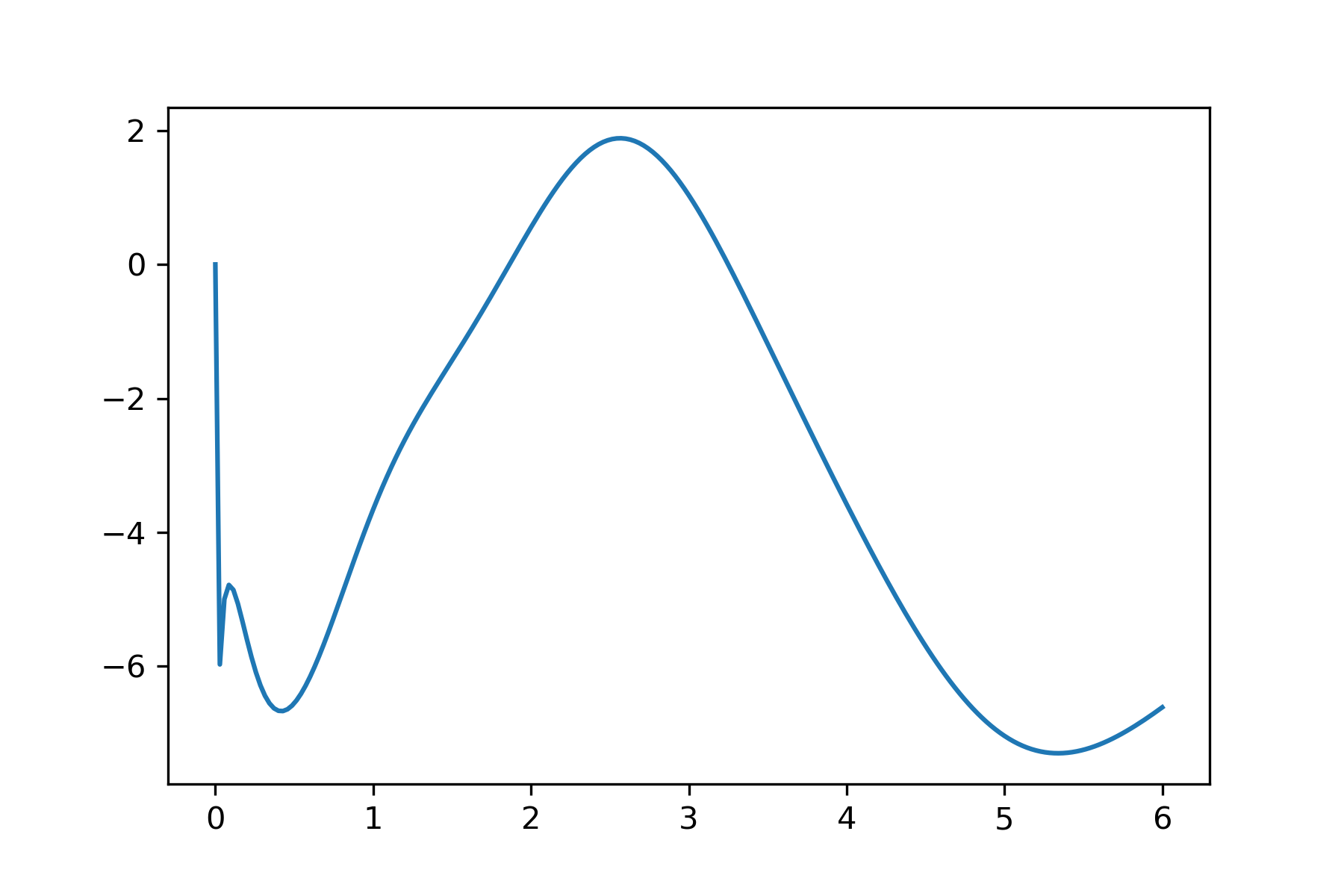}}}
    \caption{The stress applied to the lower boundary at different times $t$.}
    \label{Exp1FigLowerBStress}
\end{figure}

\begin{figure}
    \centering
    \subfloat[\centering $t=0.01$]{{\includegraphics[width=.4\linewidth]{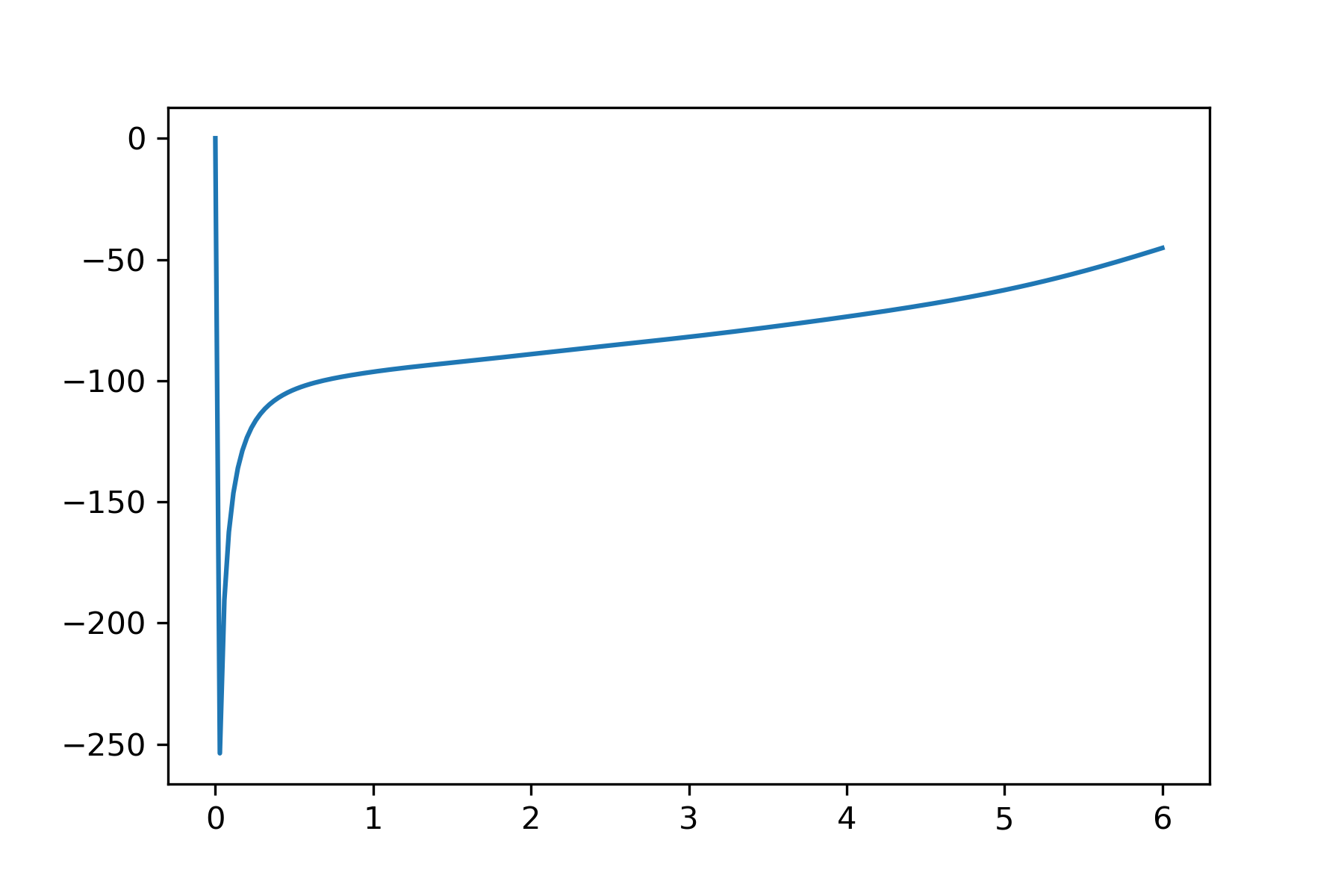} }}
    \subfloat[\centering $t=0.25$]{{\includegraphics[width=.4\linewidth]{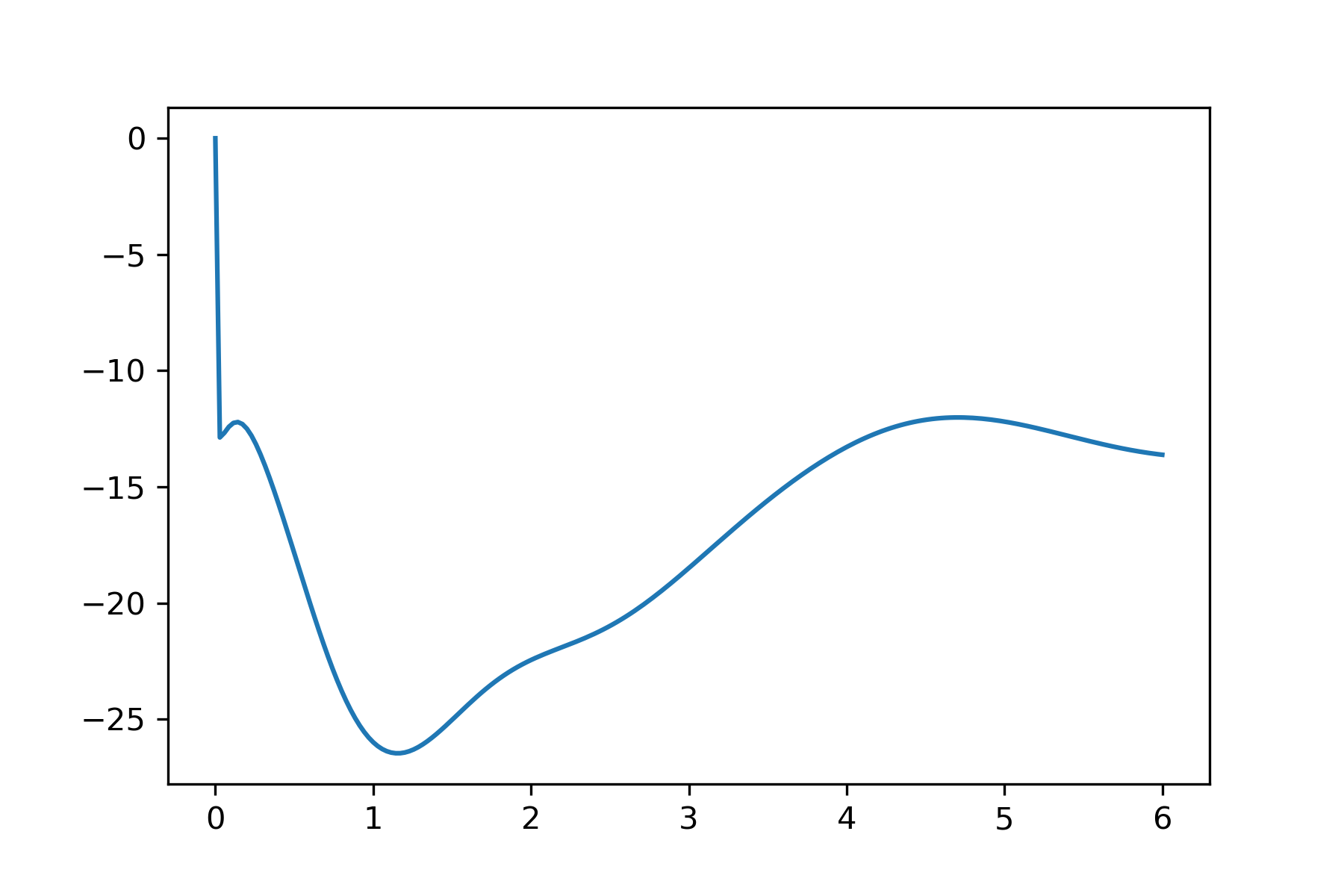}}}
    \qquad
    \subfloat[\centering $t=0.5$]{{\includegraphics[width=.4\linewidth]{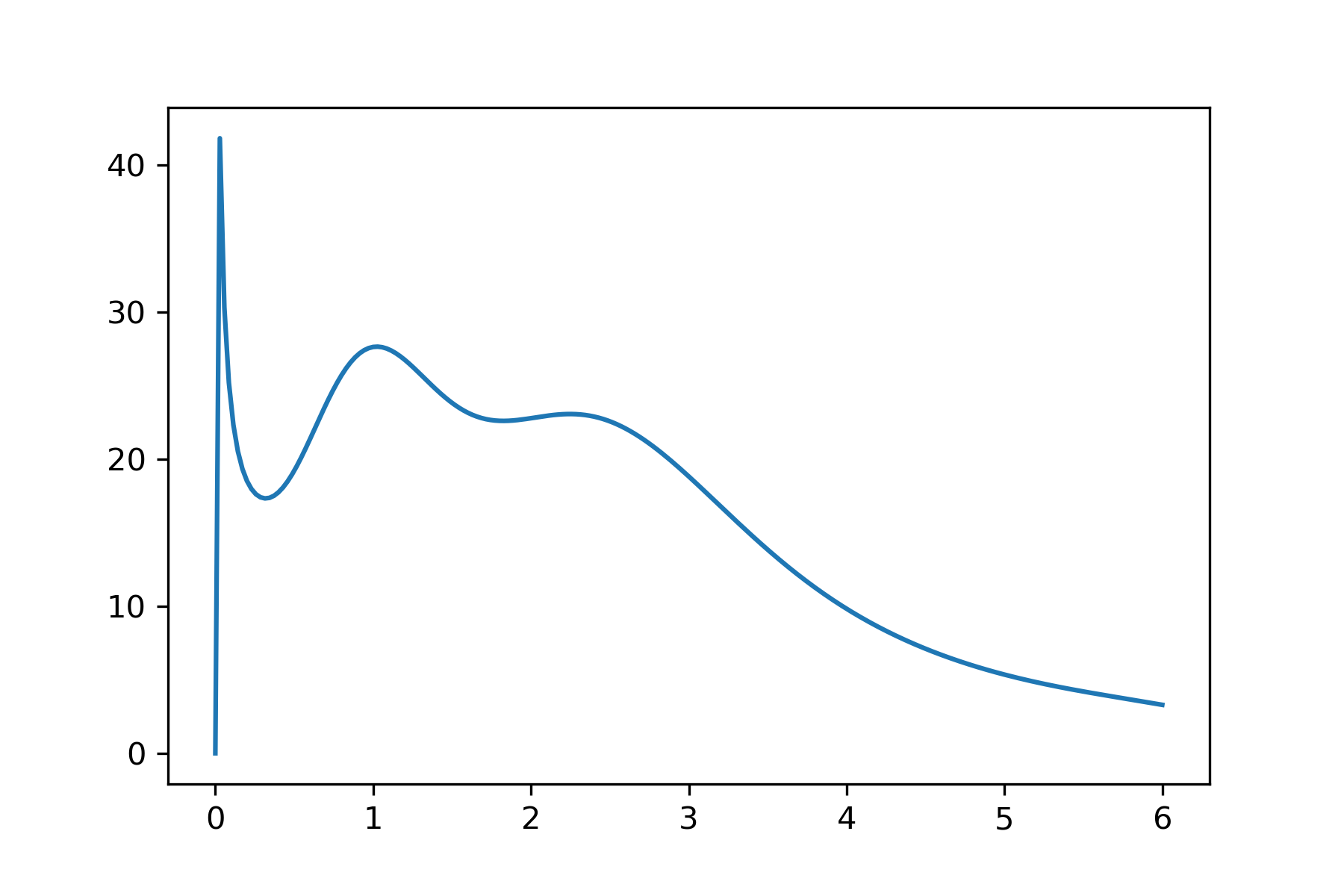} }}
    \subfloat[\centering $t=1.0$]{{\includegraphics[width=.4\linewidth]{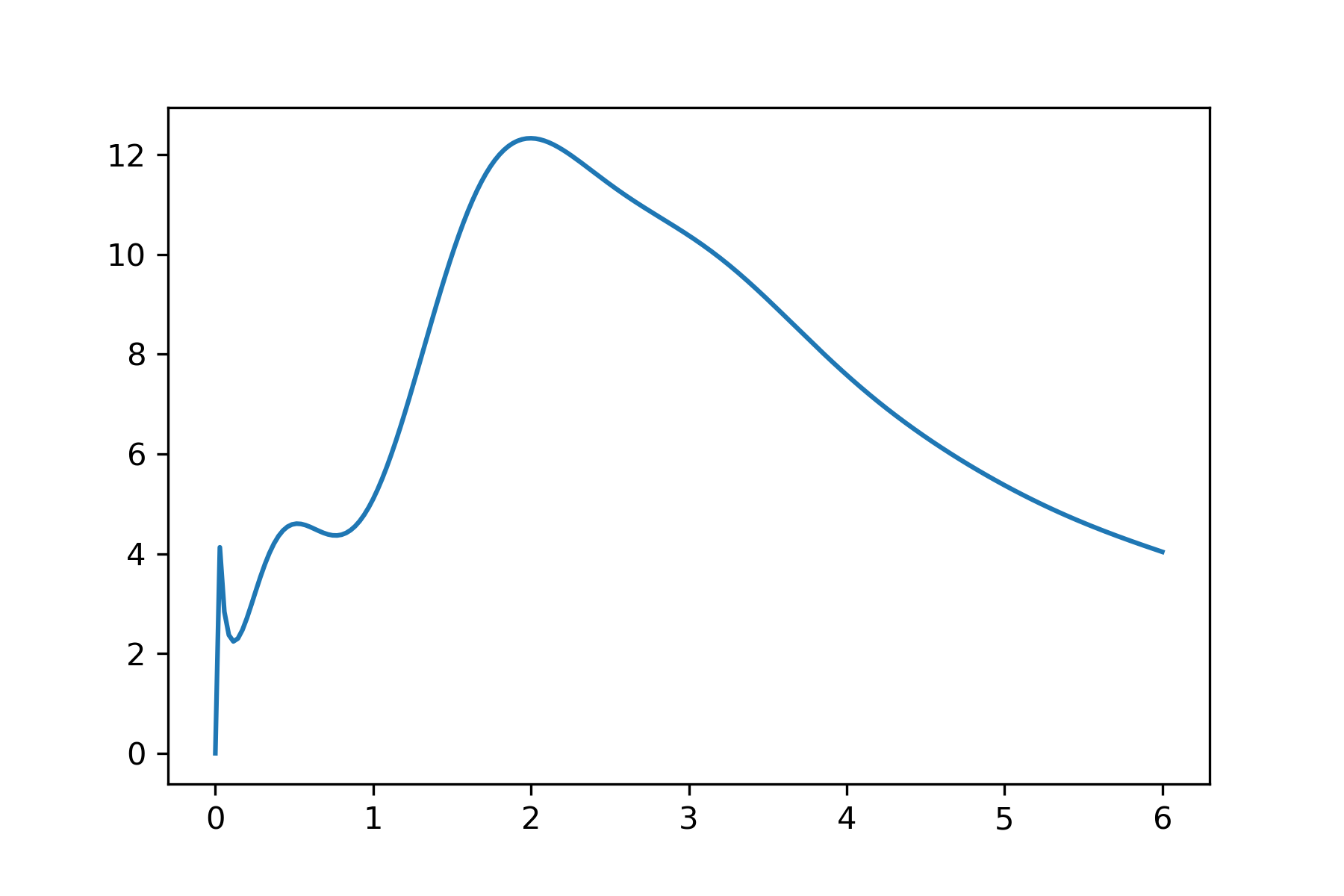}}}
    \caption{The stress applied to the upper boundary at different times $t$.}
    \label{Exp1FigUpperBStress}
\end{figure}

\begin{figure}
    \centering
    \subfloat[\centering $t=0.01$]{{\includegraphics[width=.5\linewidth]{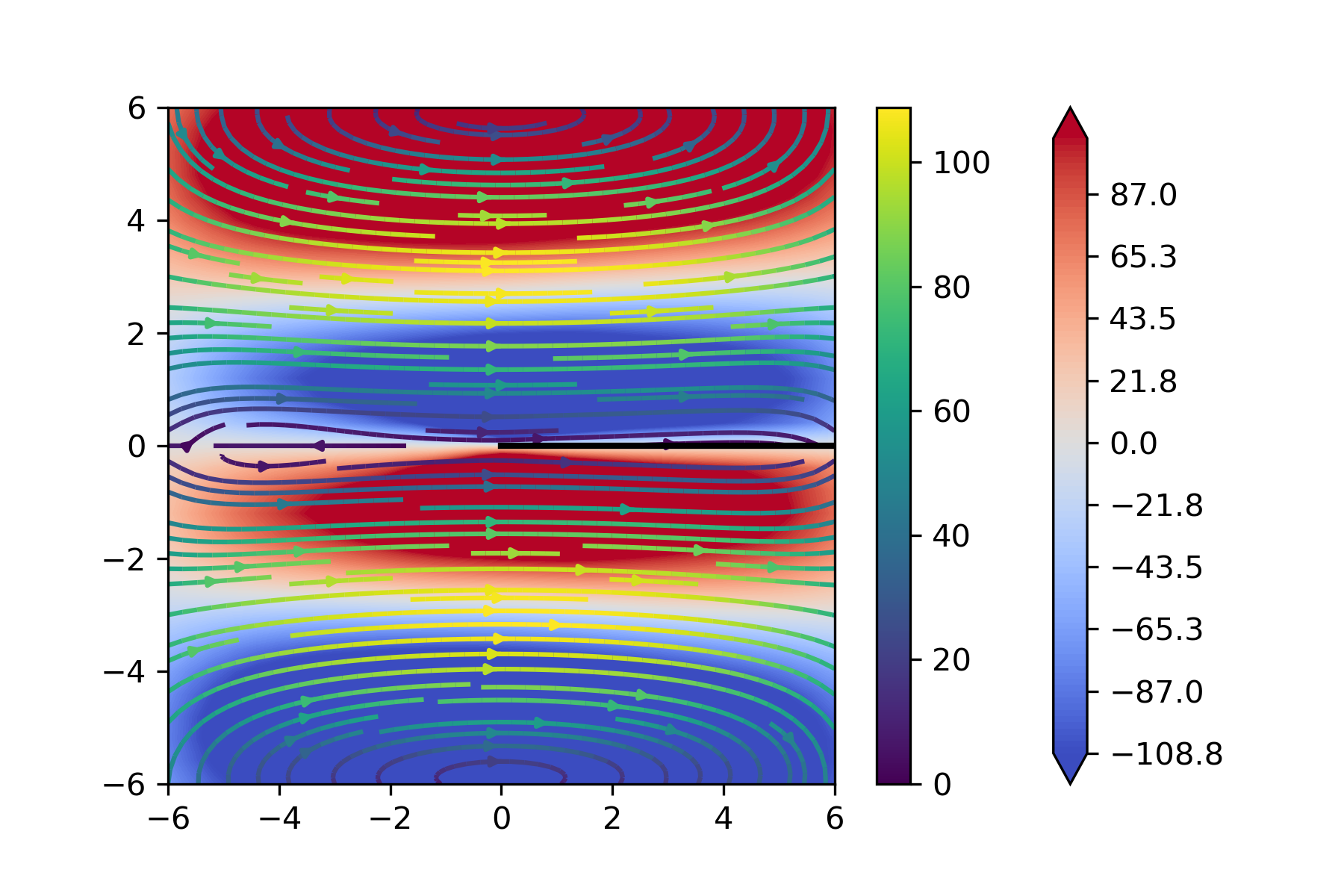} }}
    \subfloat[\centering $t=0.25$]{{\includegraphics[width=.5\linewidth]{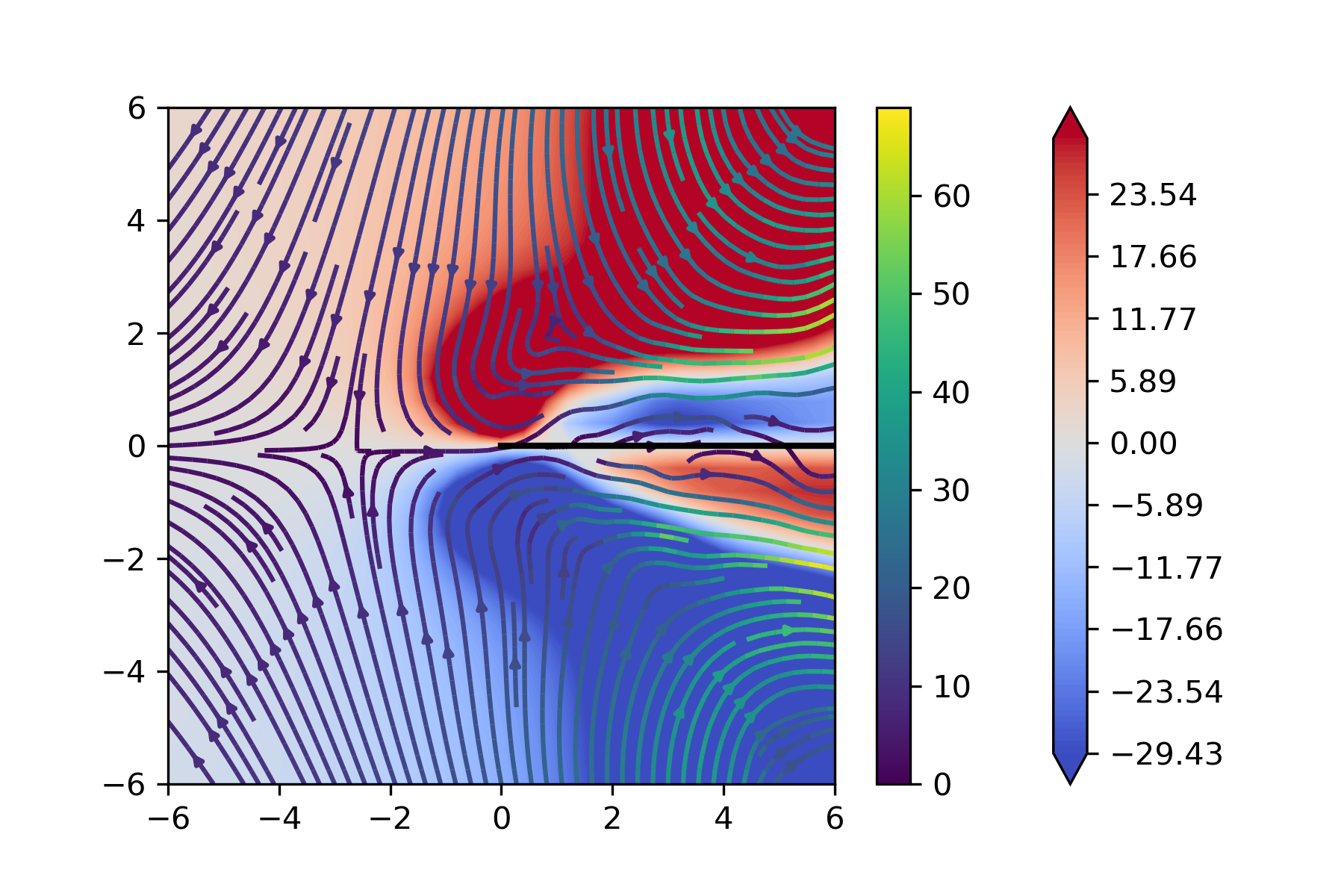}}}
    \qquad
    \subfloat[\centering $t=0.5$]{{\includegraphics[width=.5\linewidth]{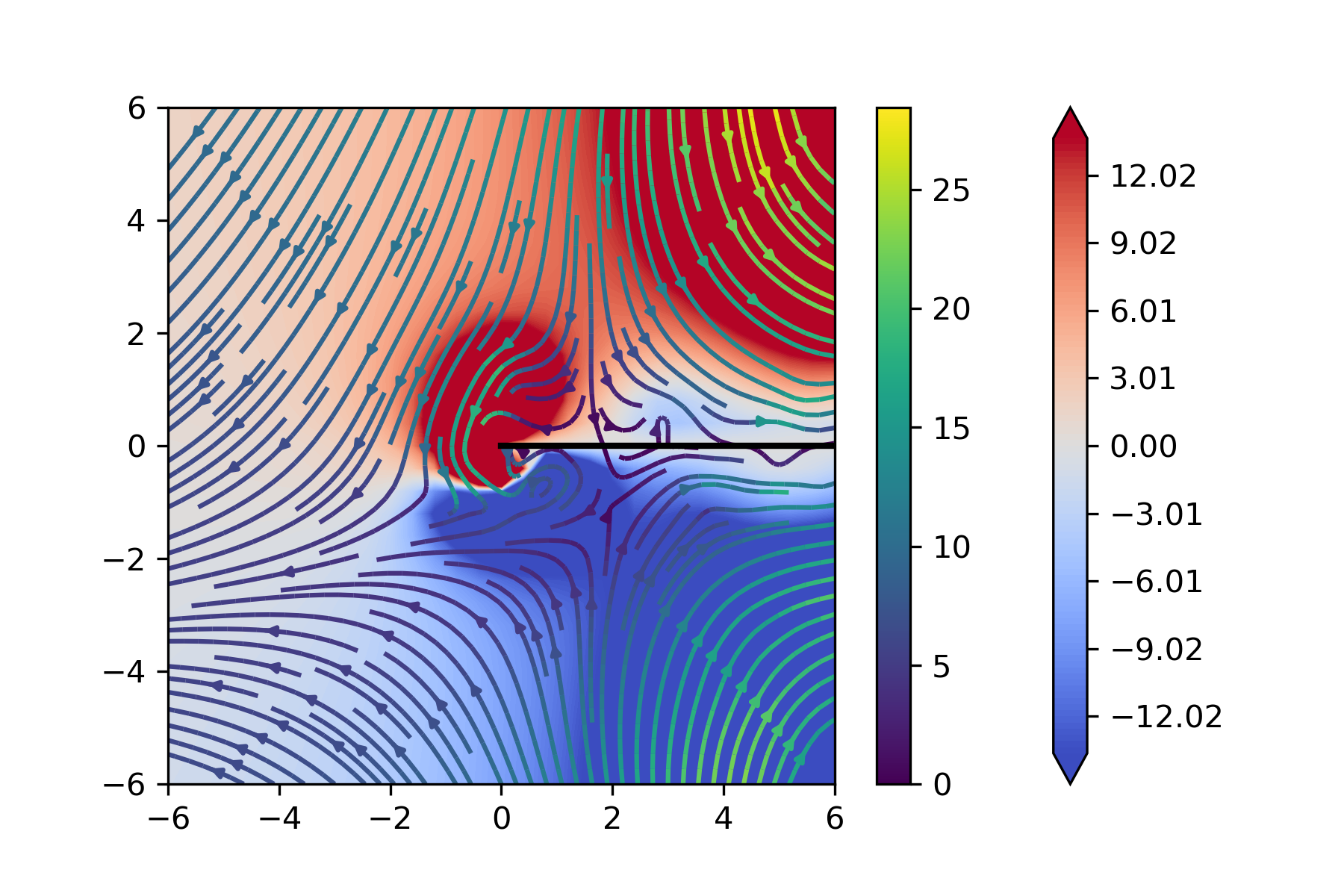} }}
    \subfloat[\centering $t=1.0$]{{\includegraphics[width=.5\linewidth]{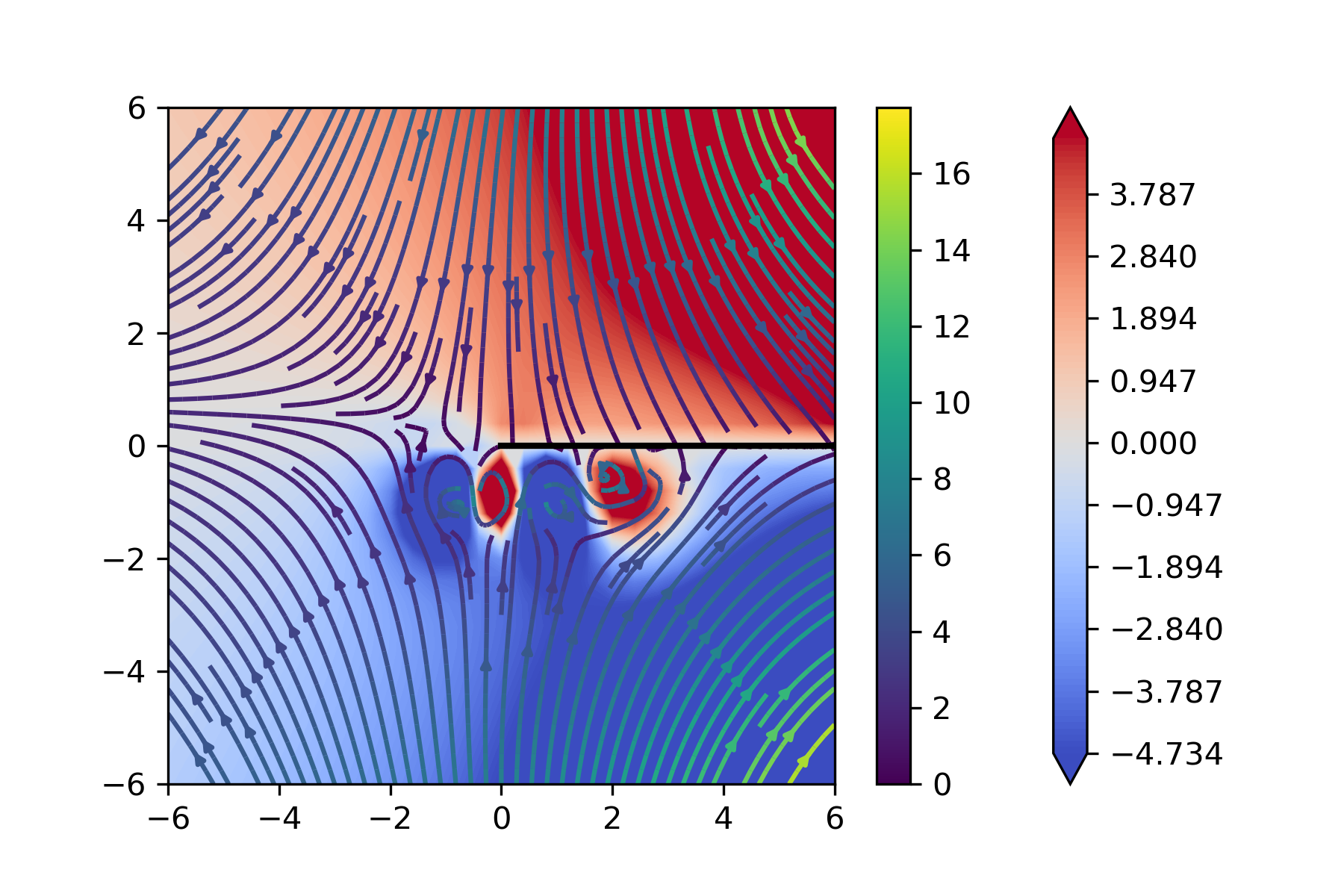}}}
    \caption{The outer layer flow at different times $t$.}
    \label{Exp2FigOFlow}
\end{figure}

\begin{figure}
    \centering
    \subfloat[\centering $t=0.01$]{{\includegraphics[width=.5\linewidth]{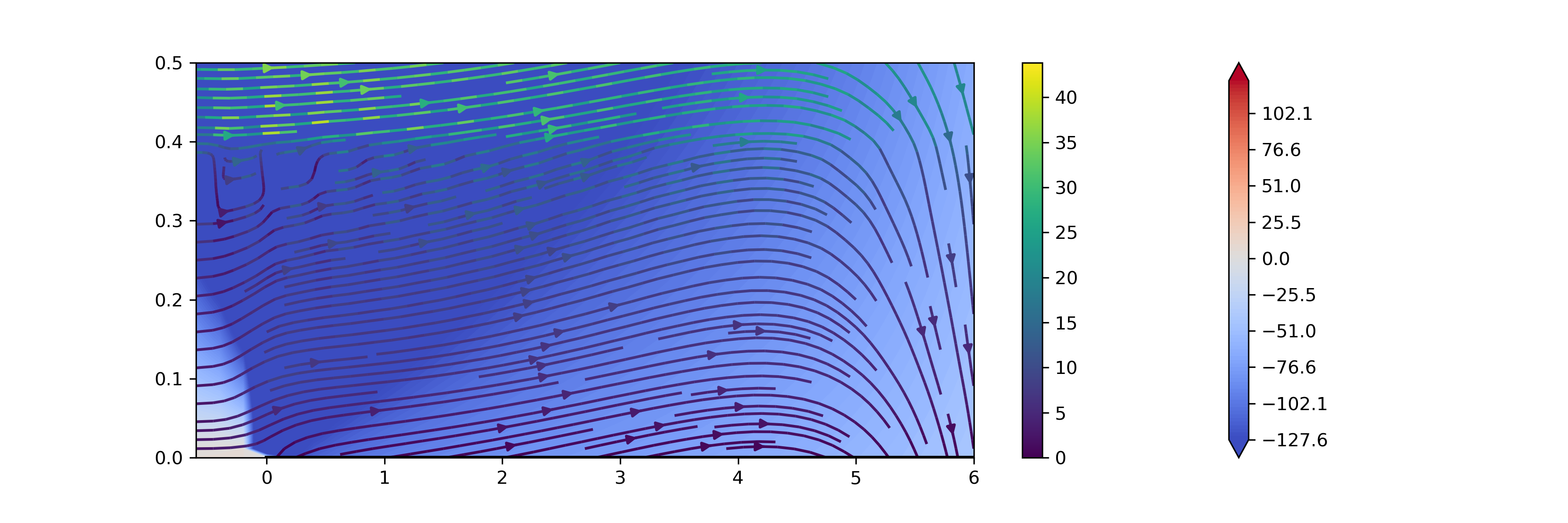} }}
    \subfloat[\centering $t=0.25$]{{\includegraphics[width=.5\linewidth]{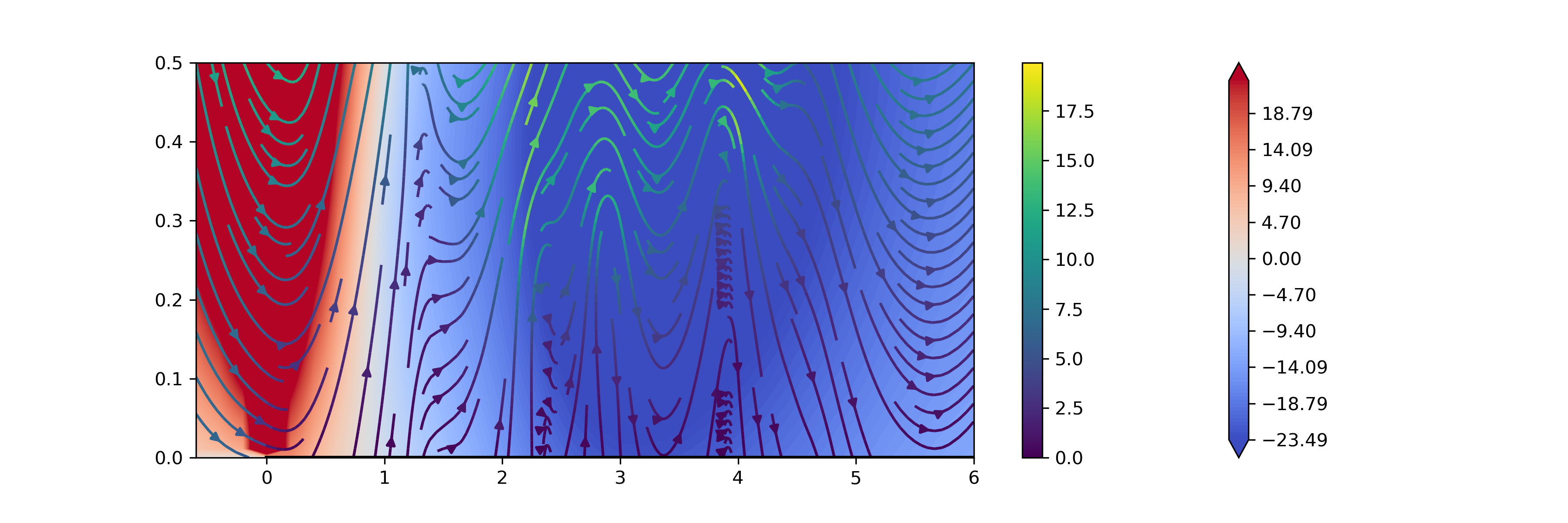}}}
    \qquad
    \subfloat[\centering $t=0.5$]{{\includegraphics[width=.5\linewidth]{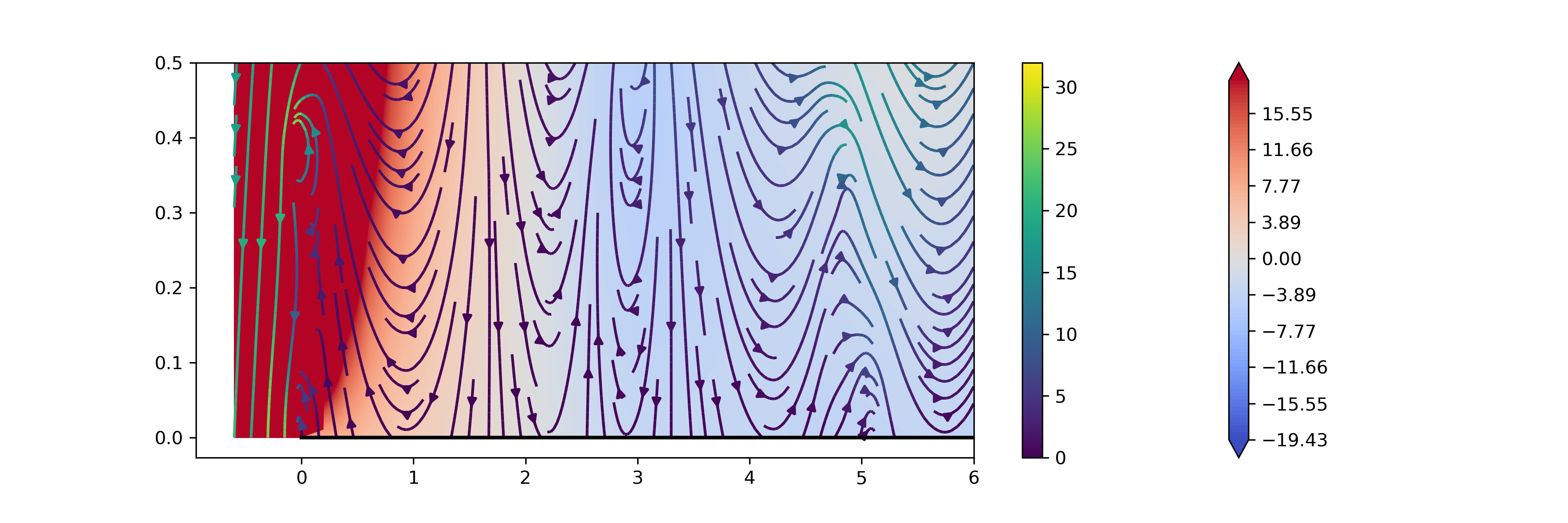} }}
    \subfloat[\centering $t=1.0$]{{\includegraphics[width=.5\linewidth]{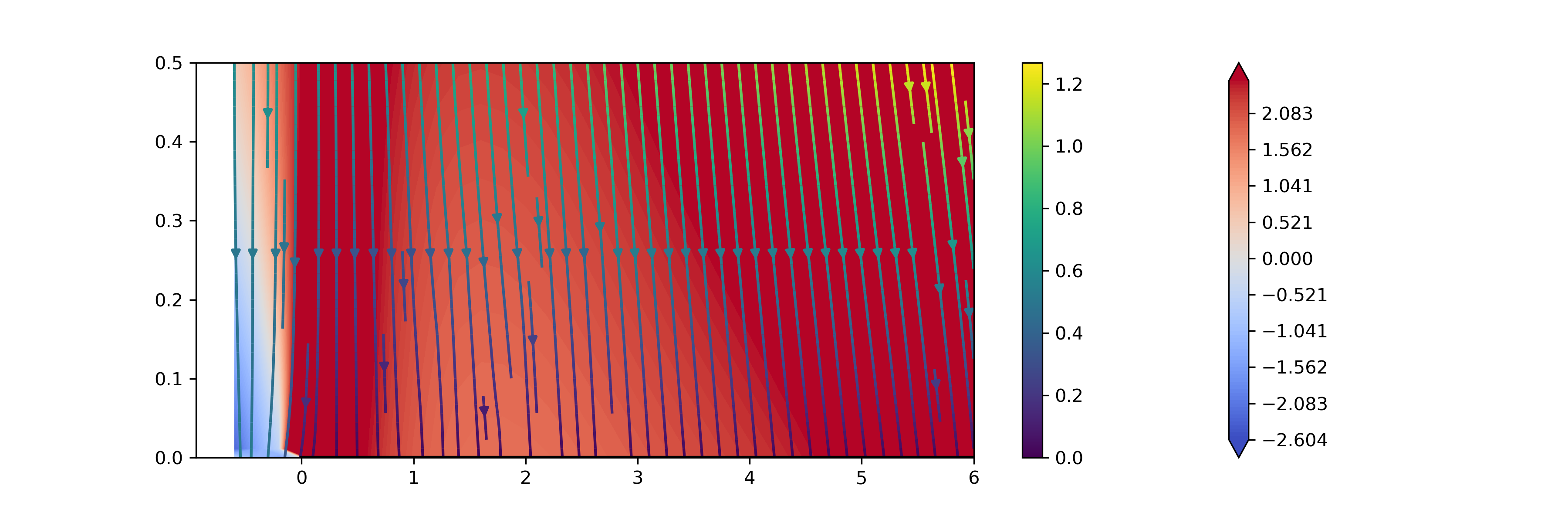}}}
    \caption{The upper boundary layer flow at different times $t$.}
    \label{Exp2FigUpperBFlow}
\end{figure}

\begin{figure}
    \centering
    \subfloat[\centering $t=0.01$]{{\includegraphics[width=.5\linewidth]{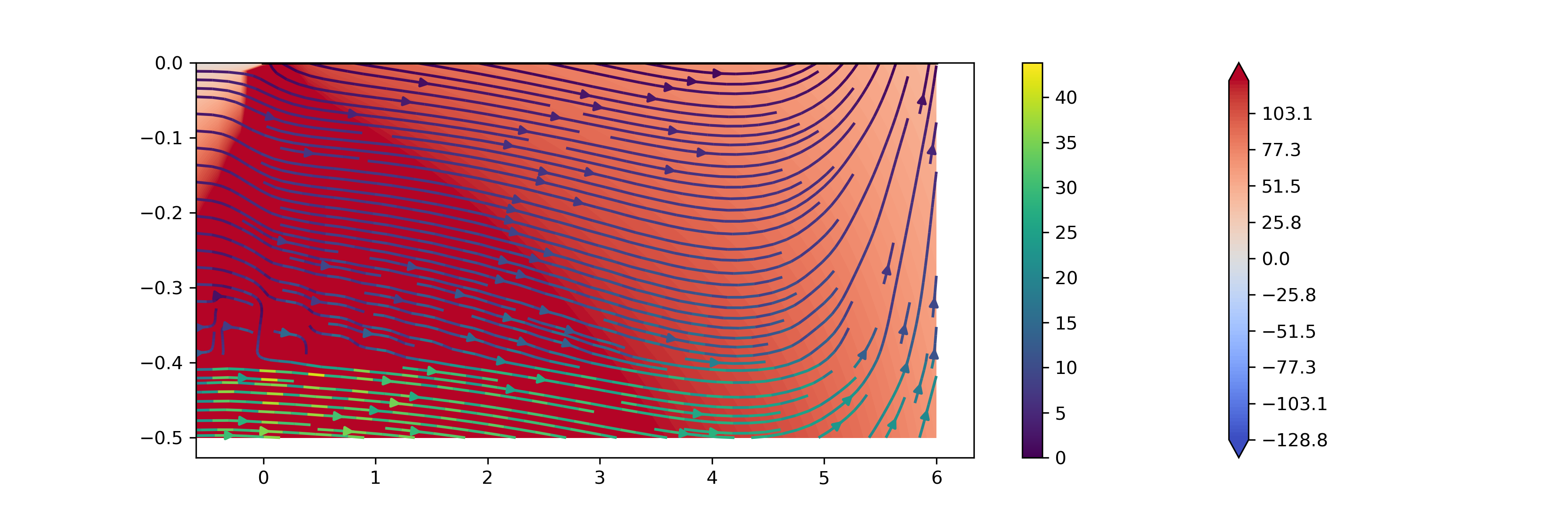} }}
    \subfloat[\centering $t=0.25$]{{\includegraphics[width=.5\linewidth]{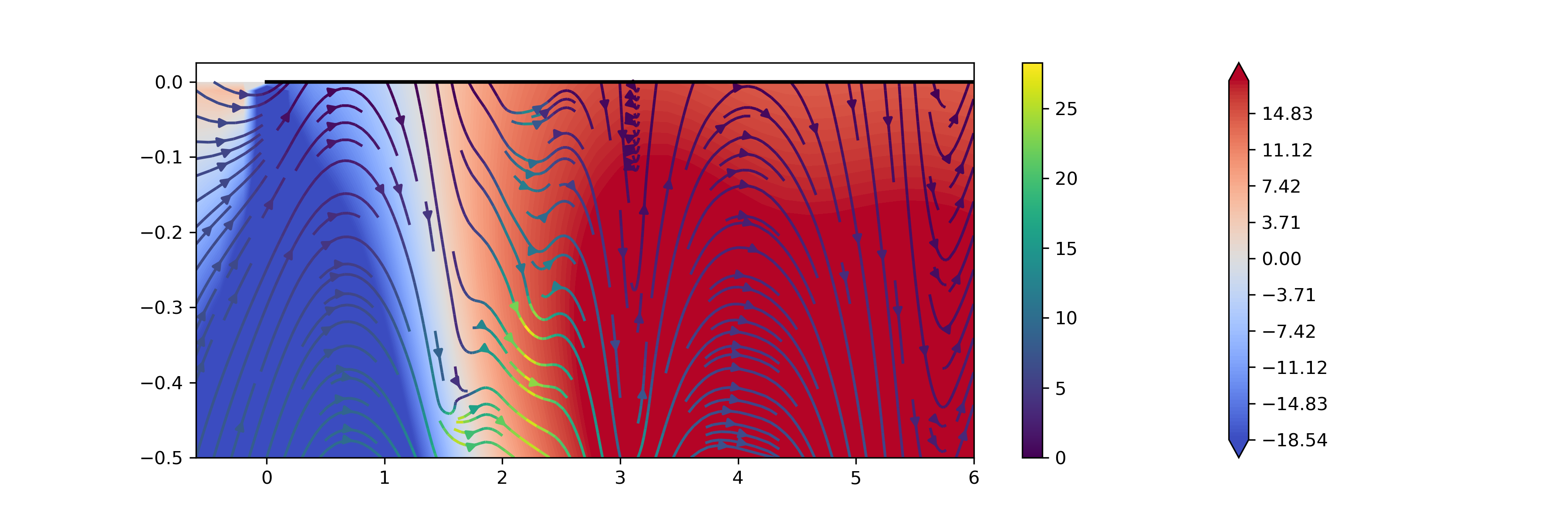}}}
    \qquad
    \subfloat[\centering $t=0.5$]{{\includegraphics[width=.5\linewidth]{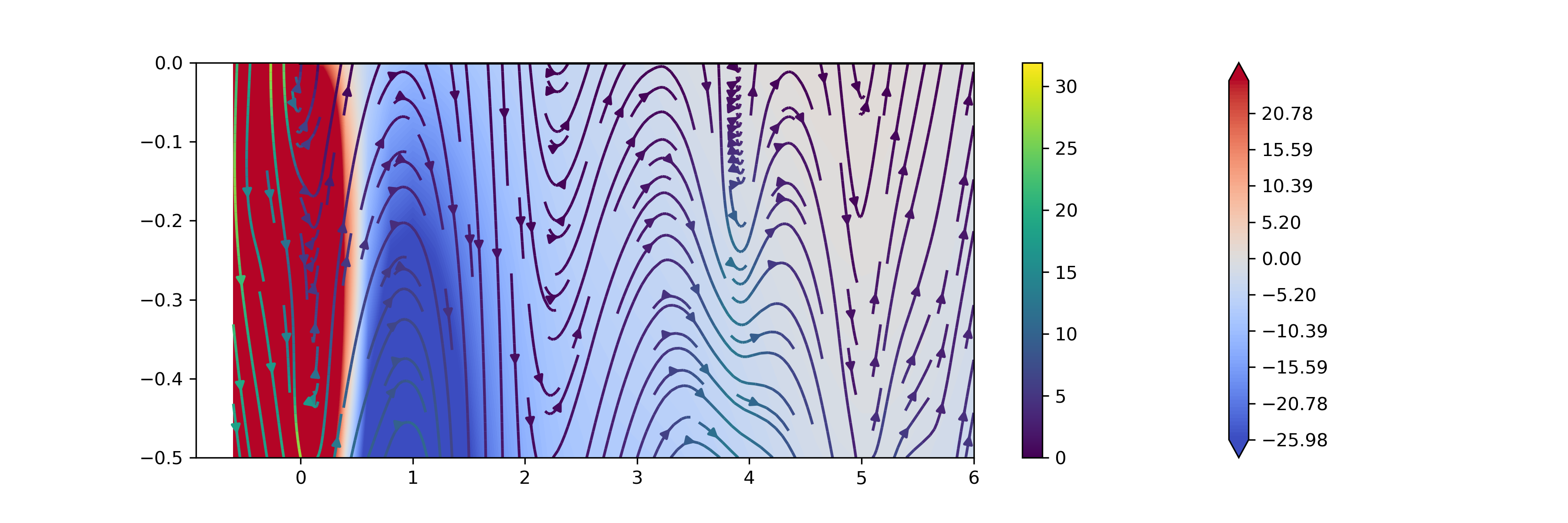} }}
    \subfloat[\centering $t=1.0$]{{\includegraphics[width=.5\linewidth]{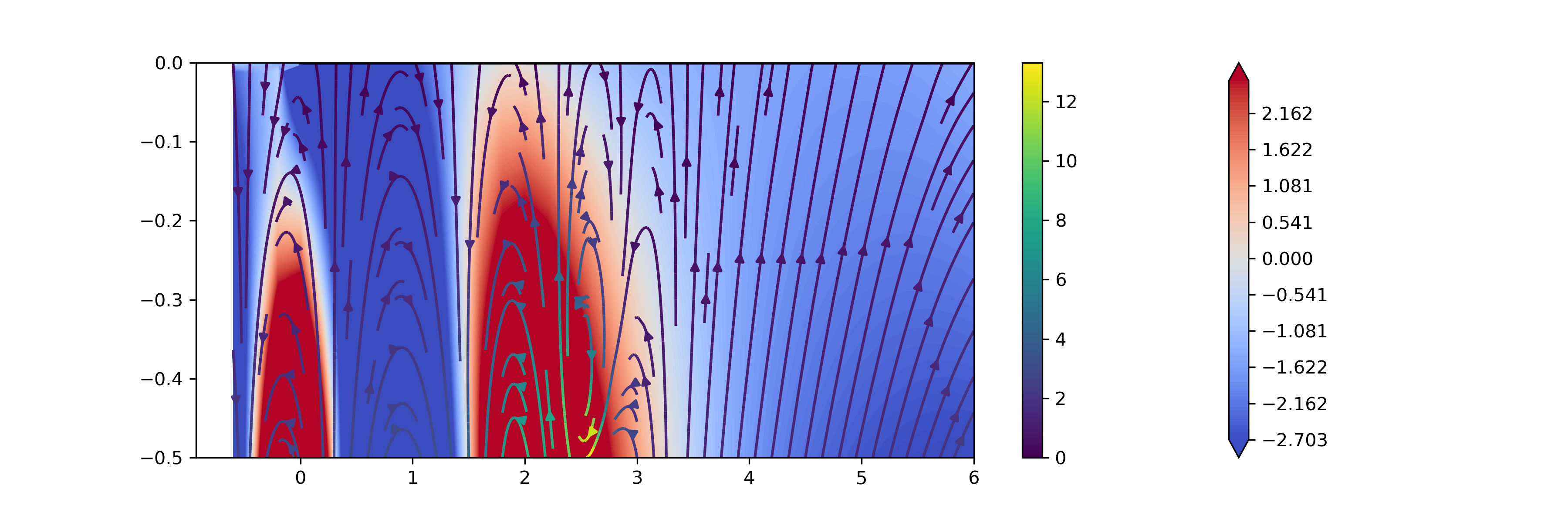}}}
    \caption{The lower boundary layer flow at different times $t$.}
    \label{Exp2FigLowerBFlow}
\end{figure}

\begin{figure}
    \centering
    \subfloat[\centering $t=0.01$]{{\includegraphics[width=.4\linewidth]{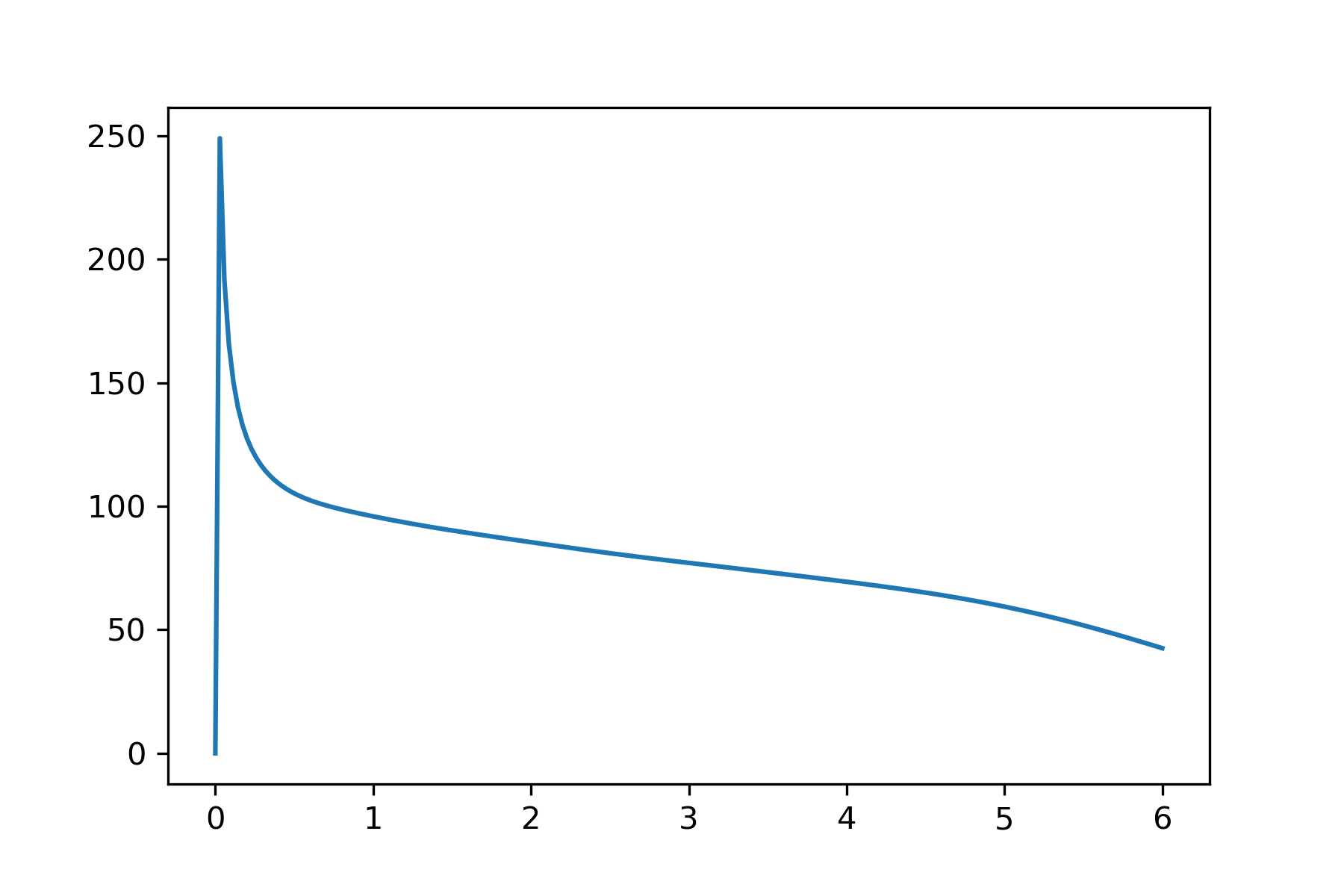} }}
    \subfloat[\centering $t=0.25$]{{\includegraphics[width=.4\linewidth]{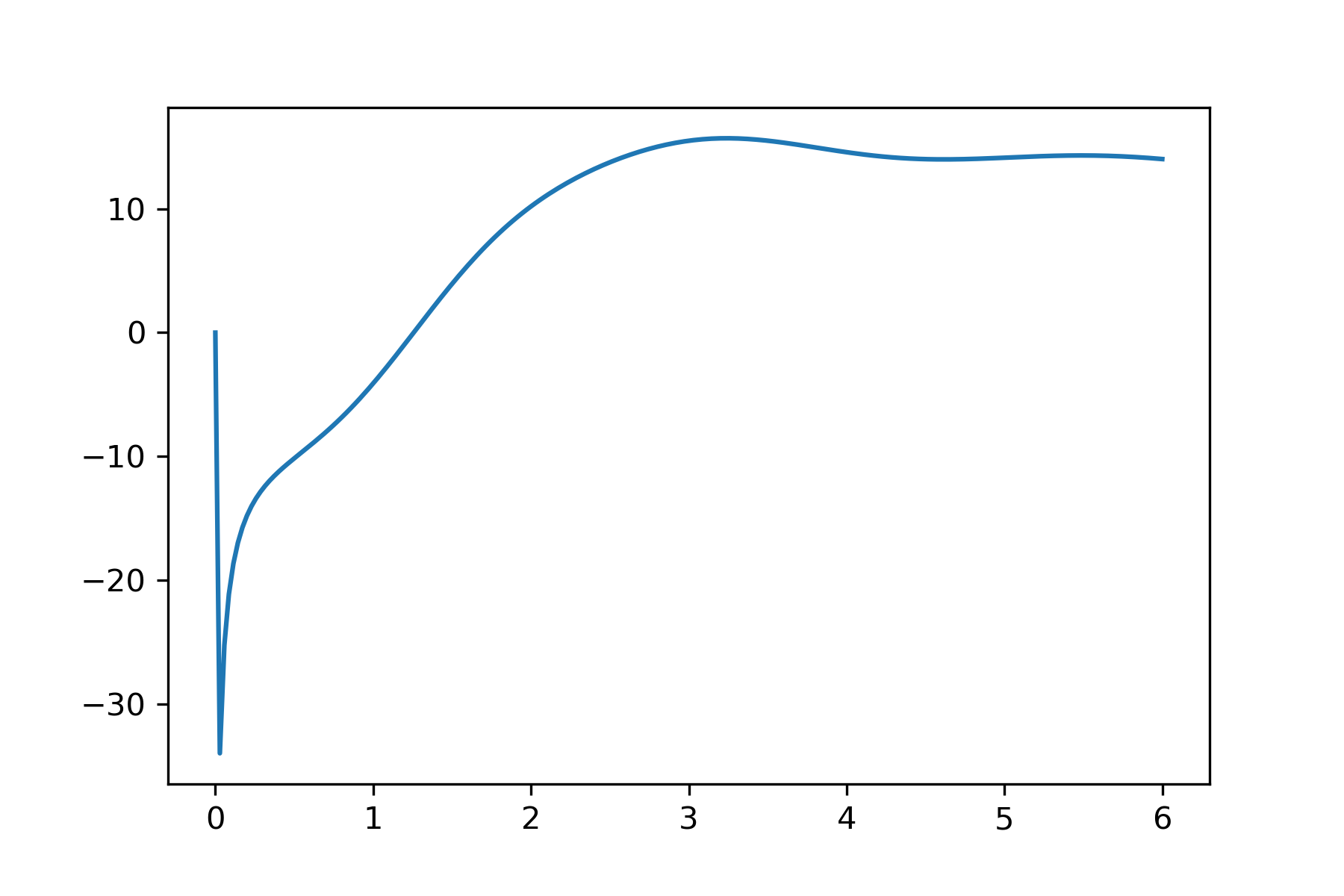}}}
    \qquad
    \subfloat[\centering $t=0.5$]{{\includegraphics[width=.4\linewidth]{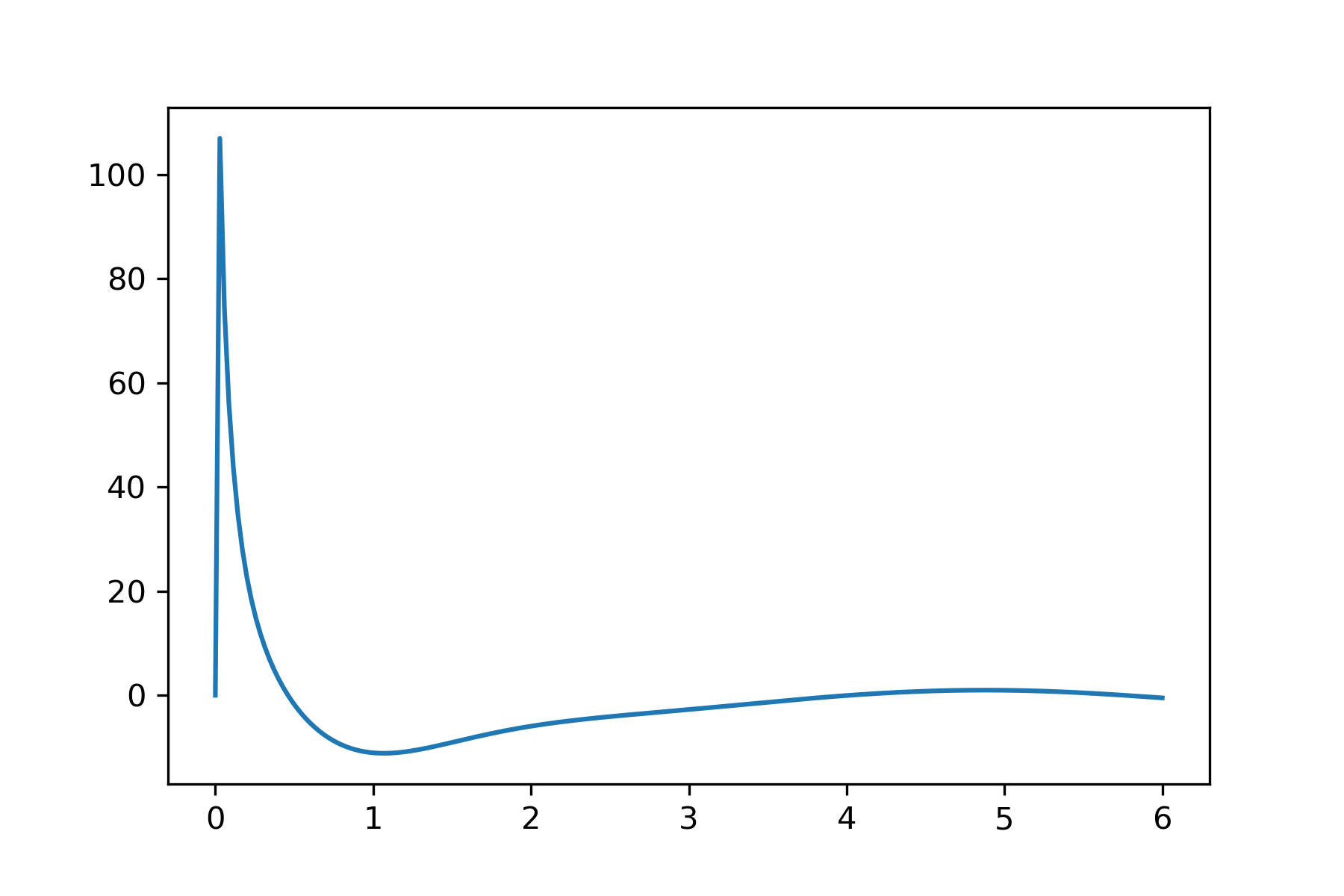} }}
    \subfloat[\centering $t=1.0$]{{\includegraphics[width=.4\linewidth]{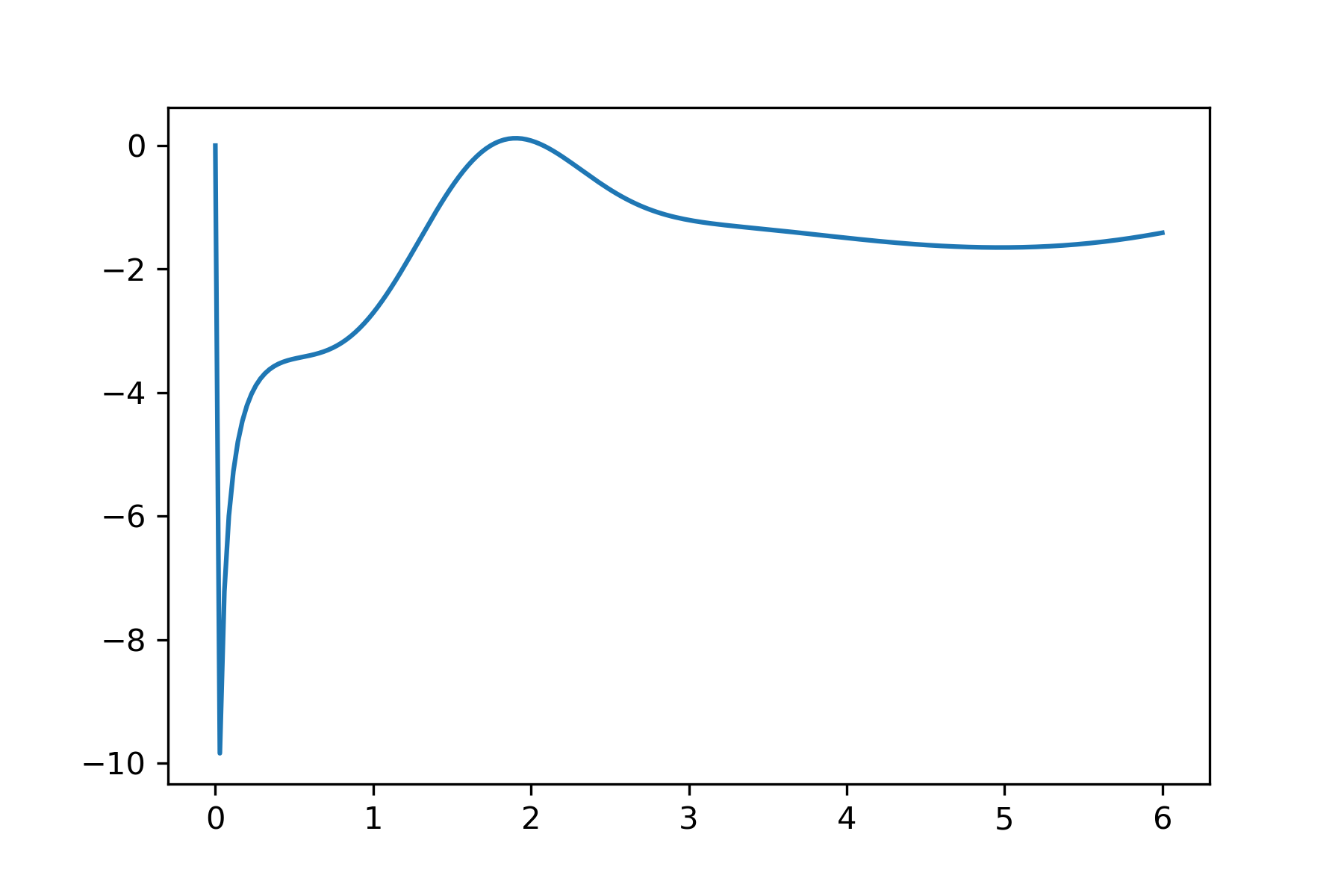}}}
    \caption{The stress applied to the lower boundary at different times $t$.}
    \label{Exp2FigLowerBStress}
\end{figure}

\begin{figure}
    \centering
    \subfloat[\centering $t=0.01$]{{\includegraphics[width=.4\linewidth]{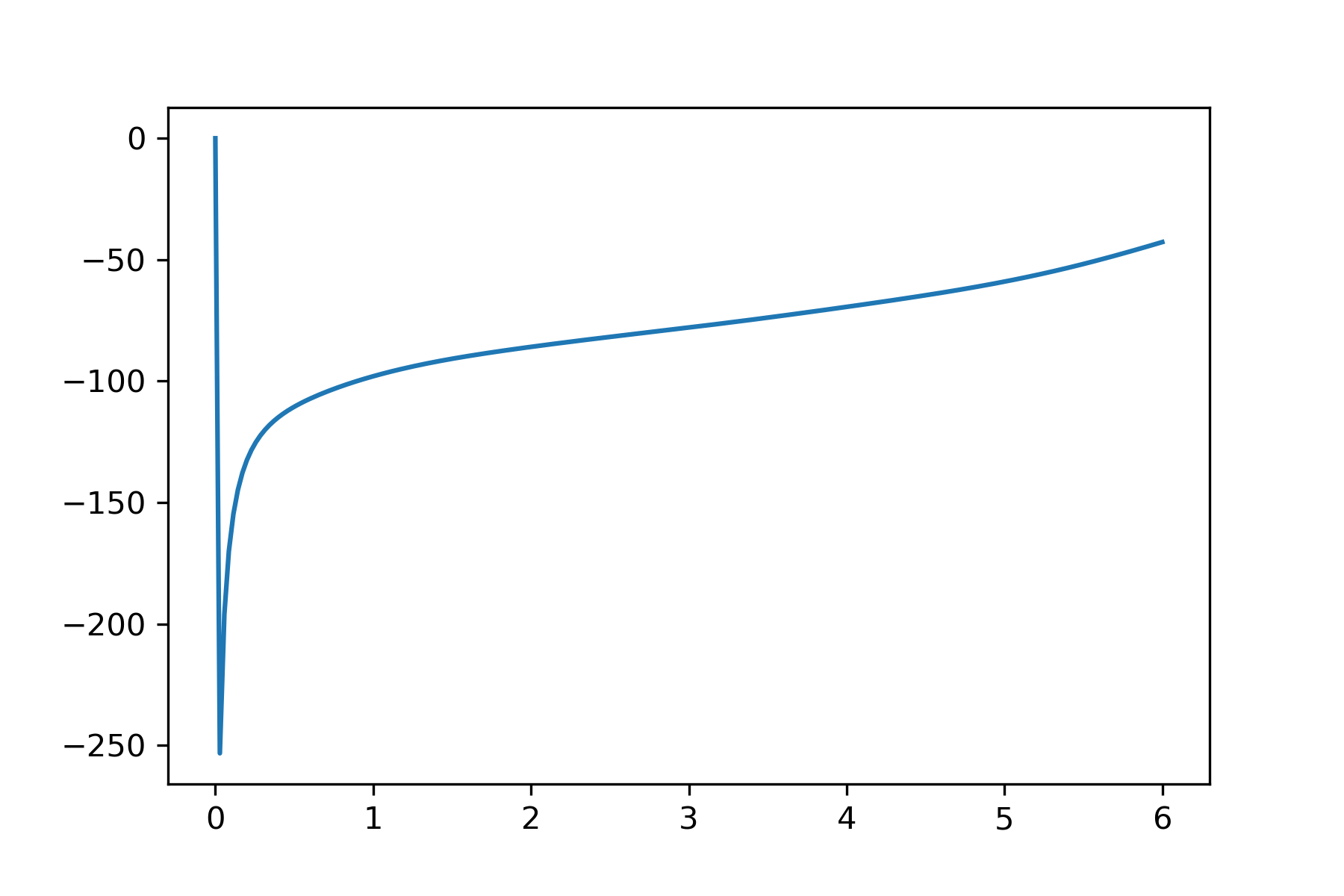} }}
    \subfloat[\centering $t=0.25$]{{\includegraphics[width=.4\linewidth]{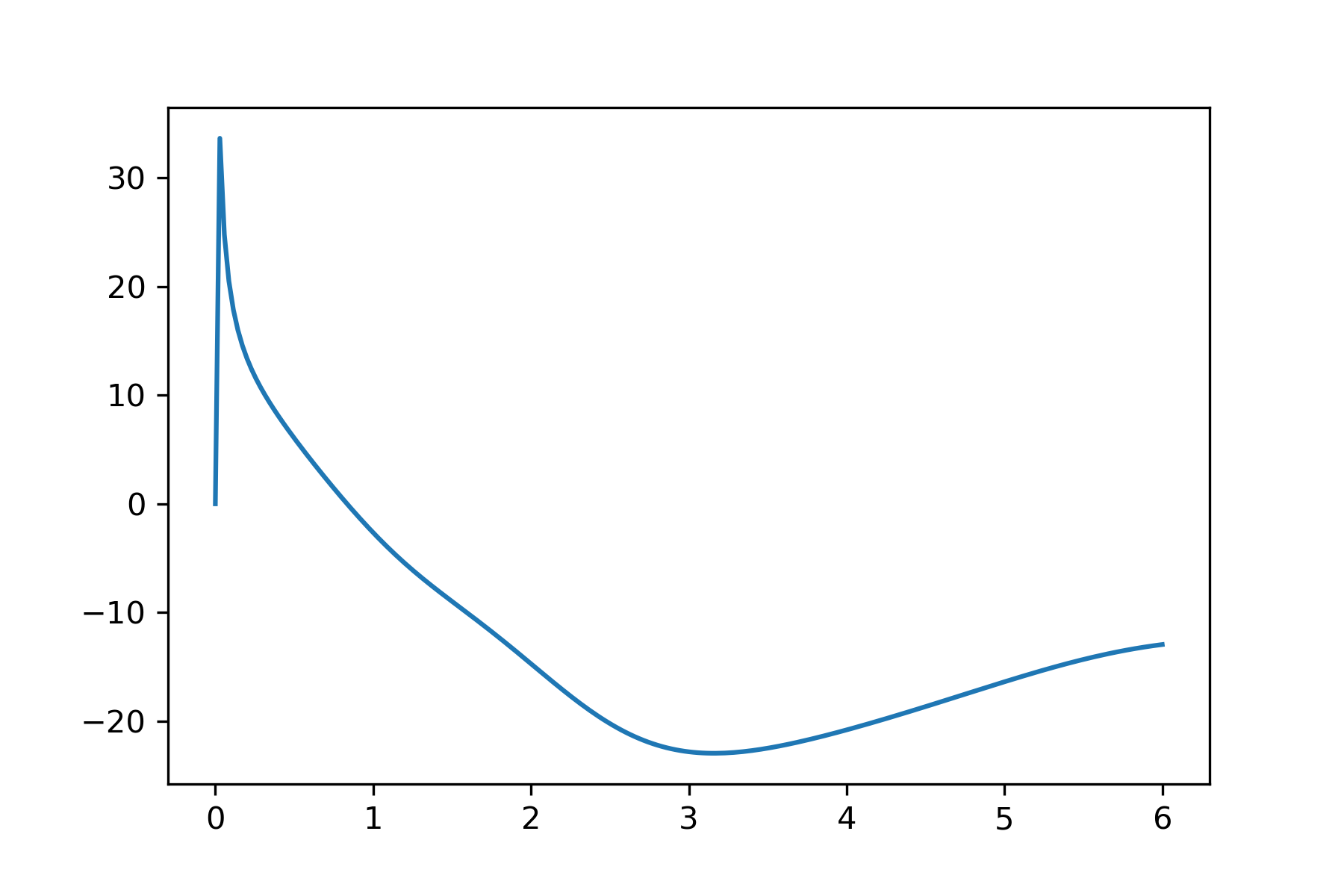}}}
    \qquad
    \subfloat[\centering $t=0.5$]{{\includegraphics[width=.4\linewidth]{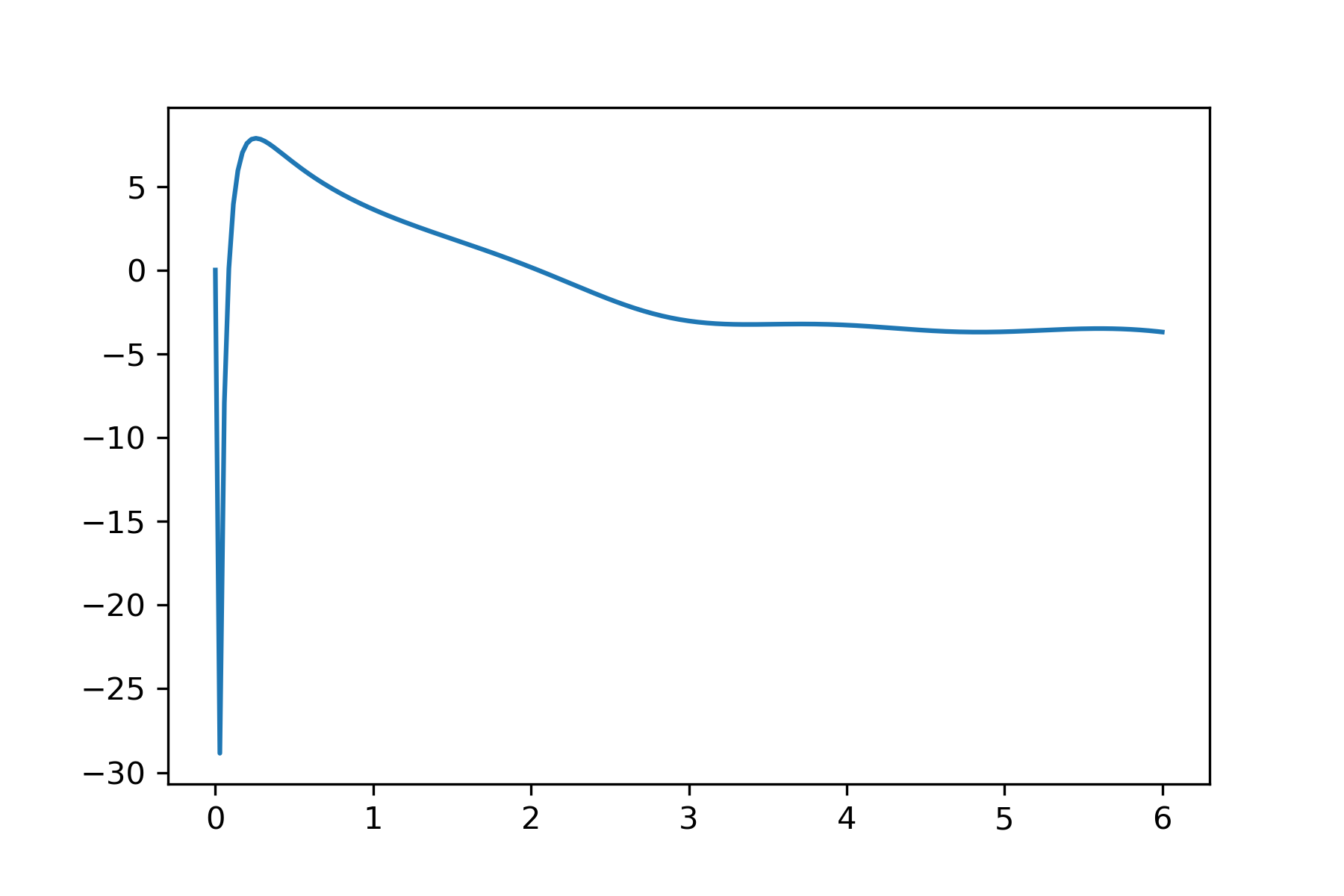} }}
    \subfloat[\centering $t=1.0$]{{\includegraphics[width=.4\linewidth]{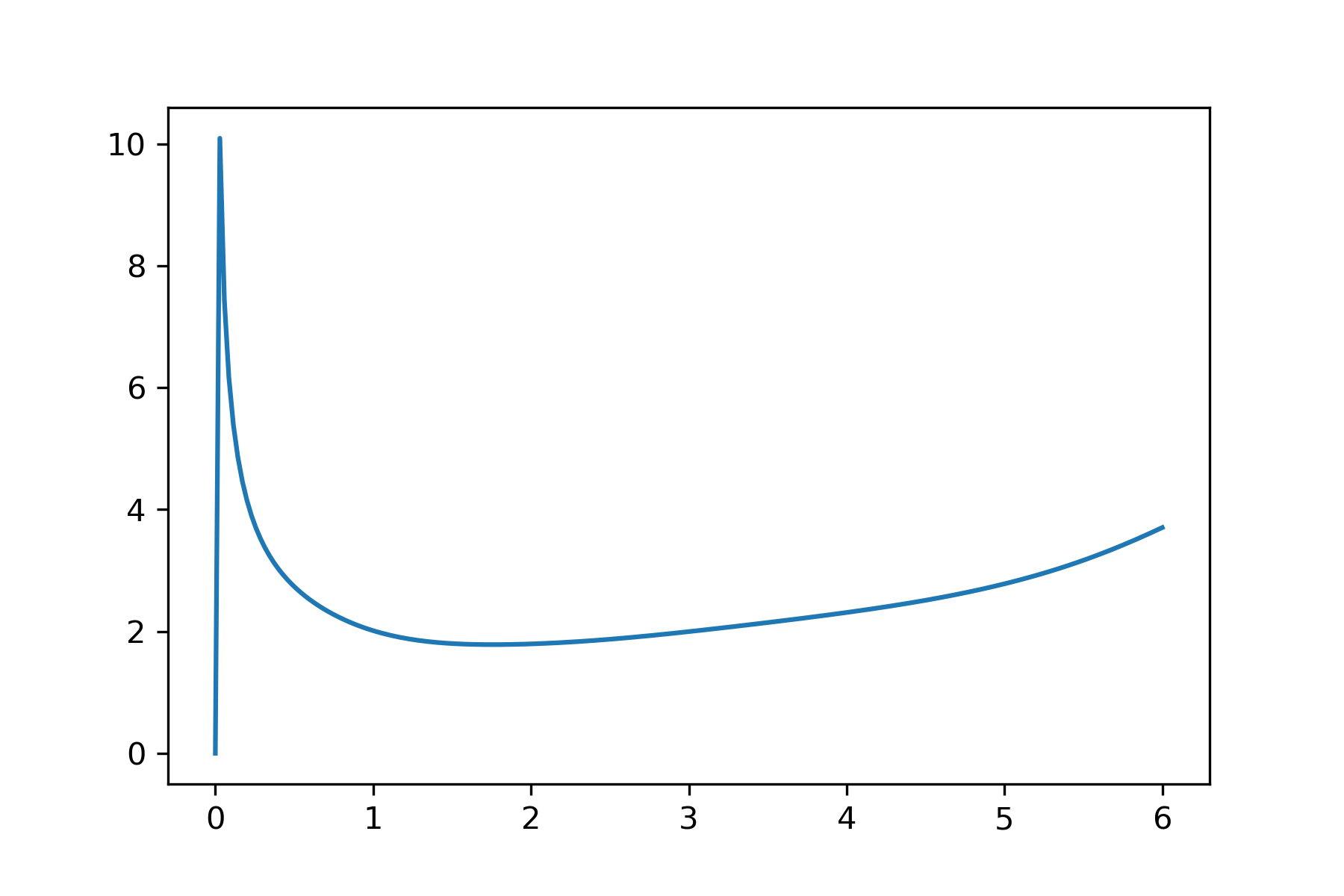}}}
    \caption{The stress applied to the upper boundary at different times $t$.}
    \label{Exp2FigUpperBStress}
\end{figure}

\begin{figure}
    \centering
    \subfloat[\centering $t=0.001$]{{\includegraphics[width=.5\linewidth]{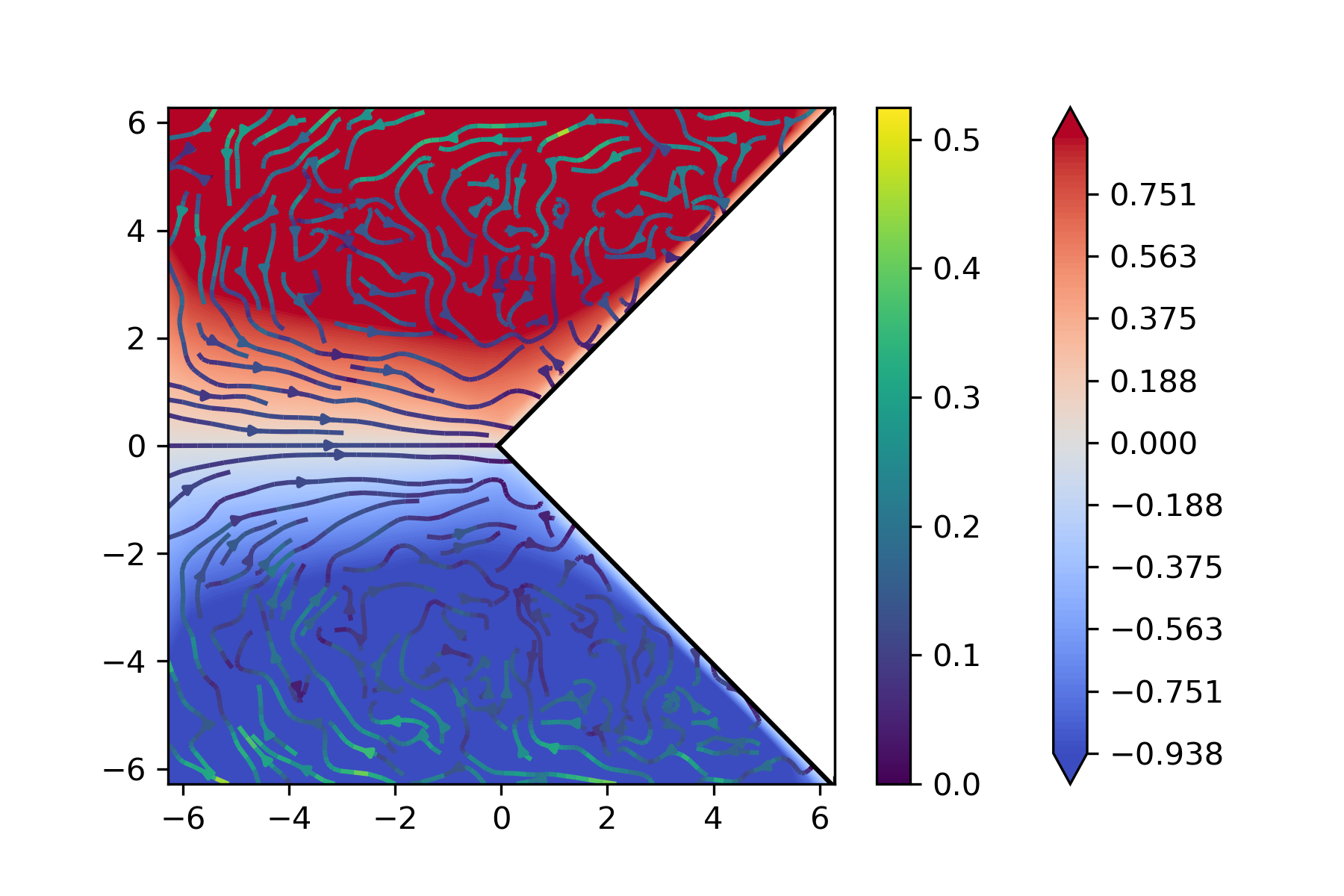} }}
    \subfloat[\centering $t=0.05$]{{\includegraphics[width=.5\linewidth]{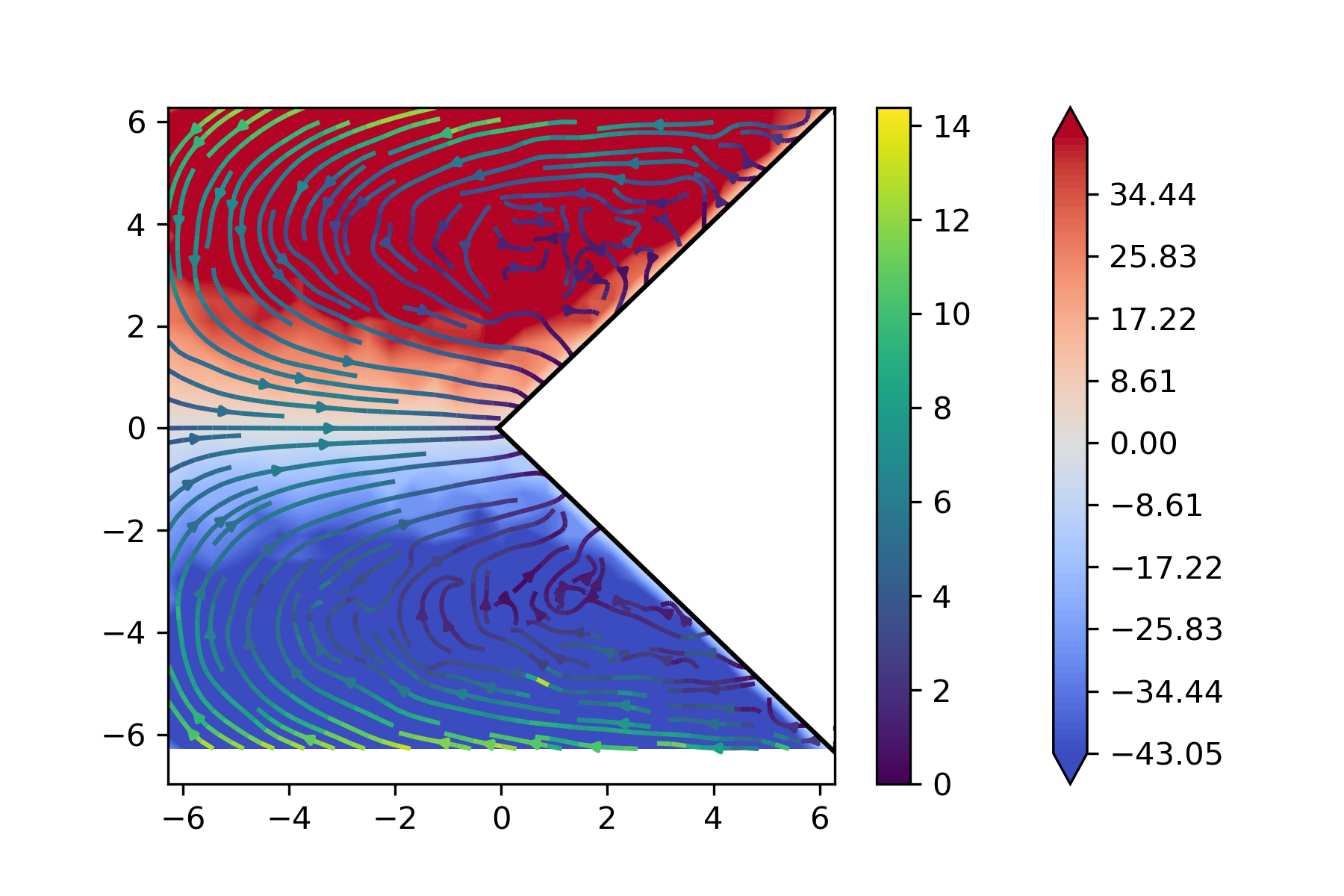}}}
    \qquad
    \subfloat[\centering $t=0.1$]{{\includegraphics[width=.5\linewidth]{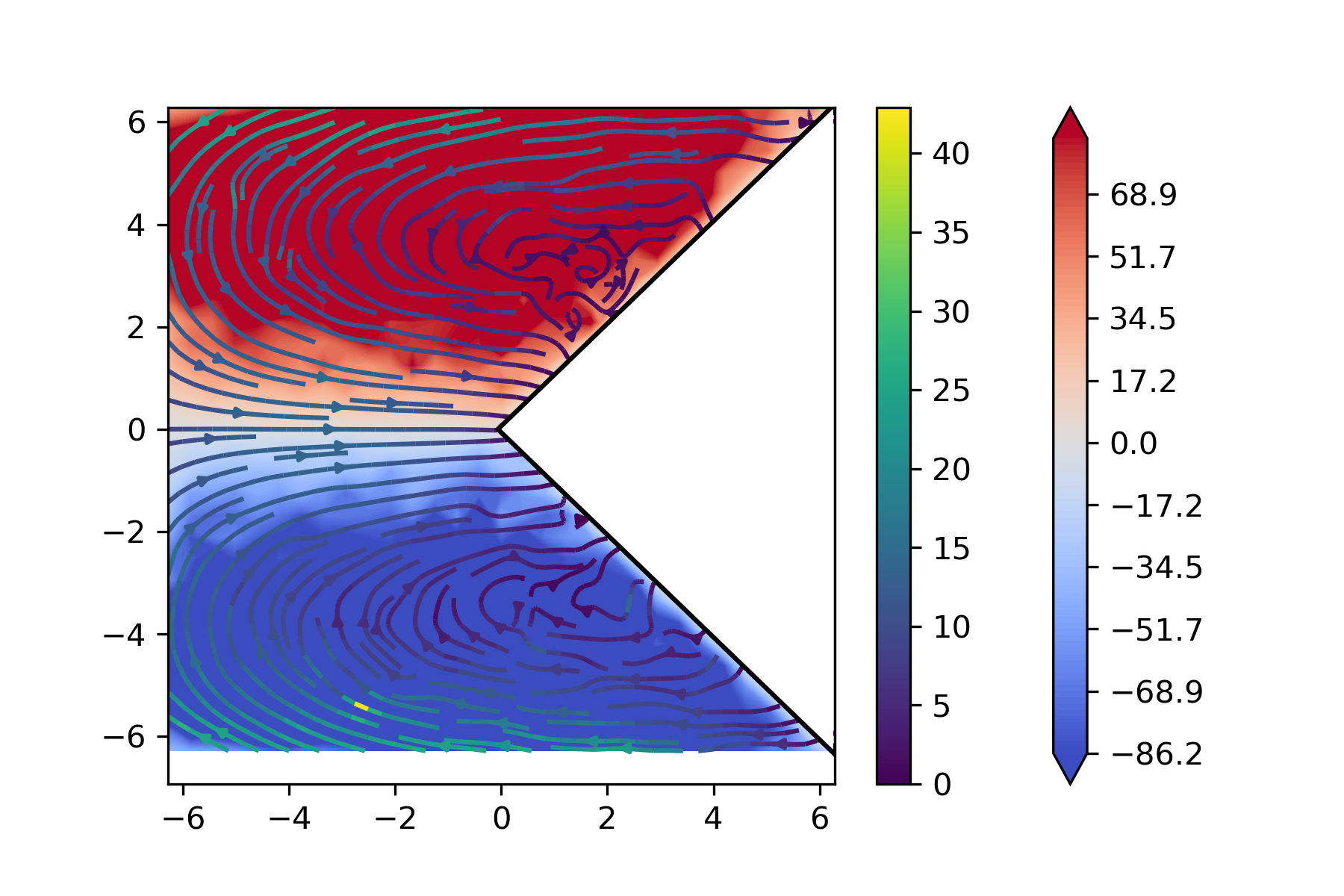} }}
    \subfloat[\centering $t=0.2$]{{\includegraphics[width=.5\linewidth]{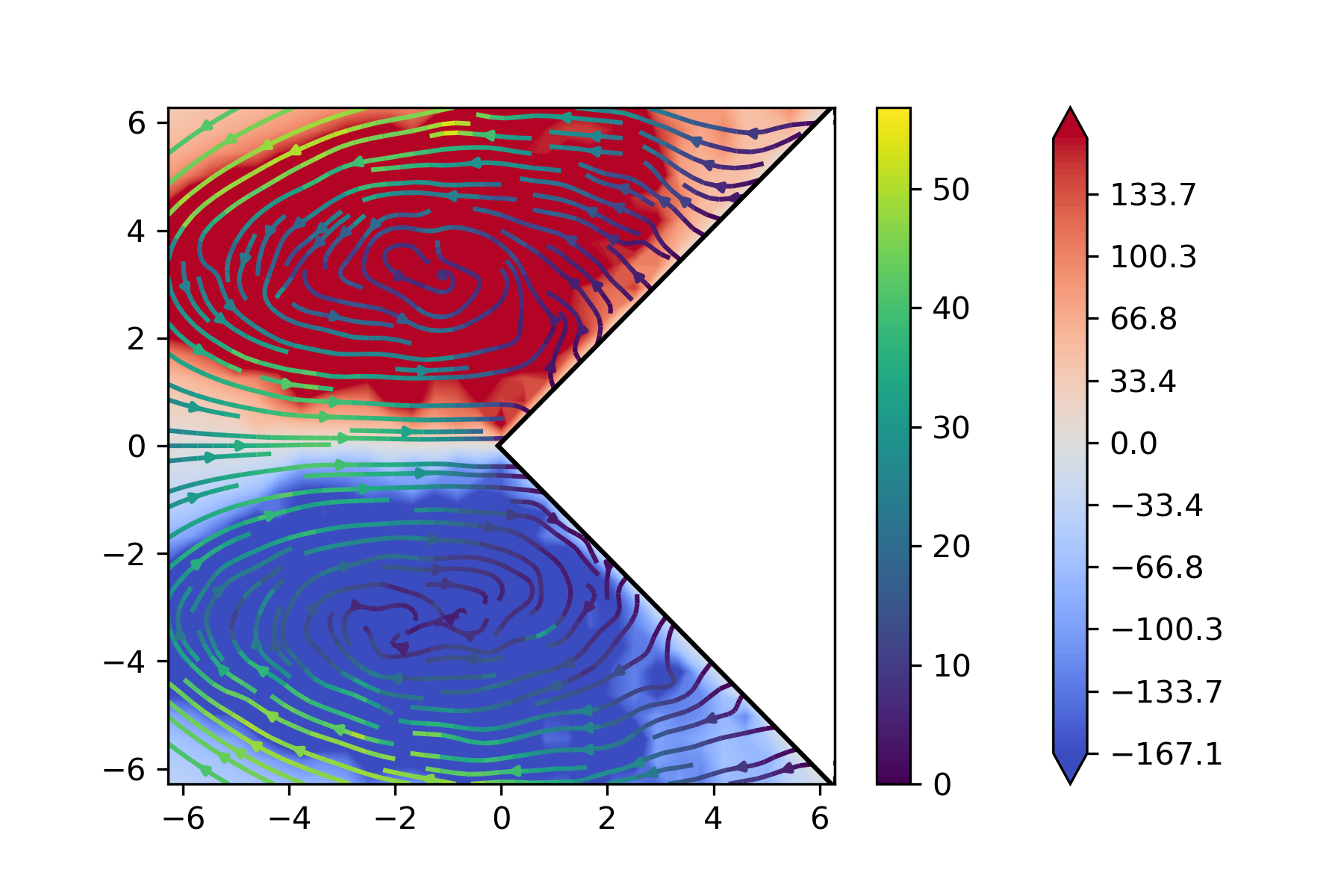}}}
    \caption{The outer layer flow at different times $t$.}
    \label{Exp3FigOFlow}
\end{figure}

\begin{figure}
    \centering
    \subfloat[\centering $t=0.001$]{{\includegraphics[width=.5\linewidth]{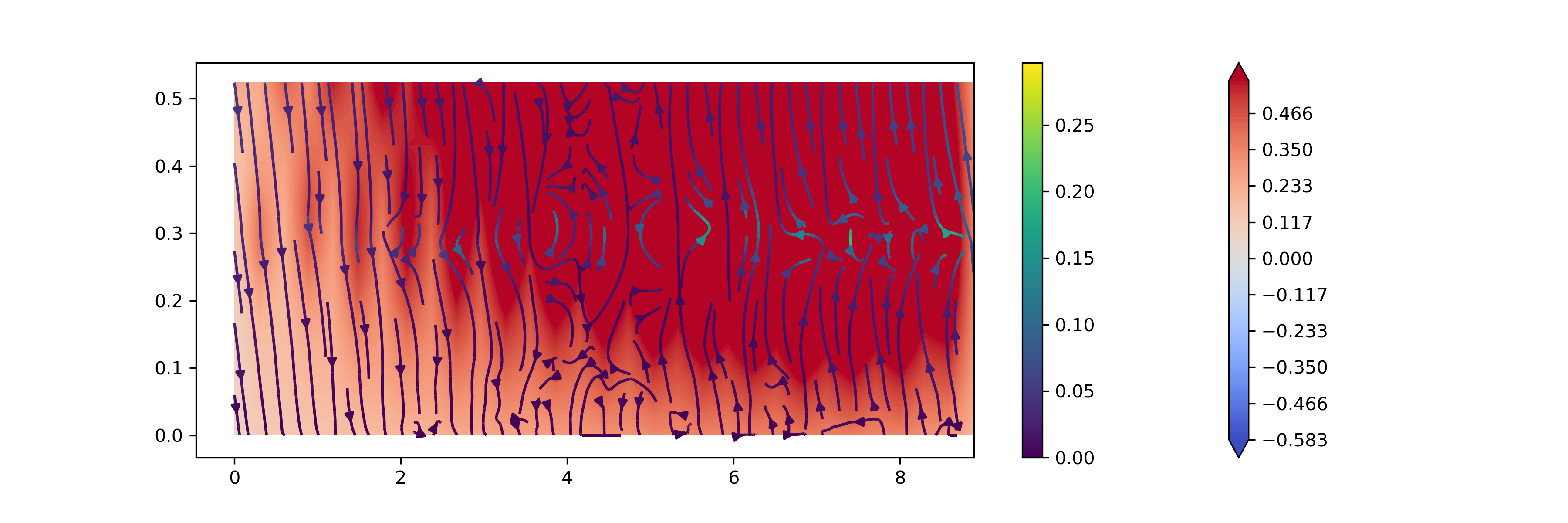} }}
    \subfloat[\centering $t=0.05$]{{\includegraphics[width=.5\linewidth]{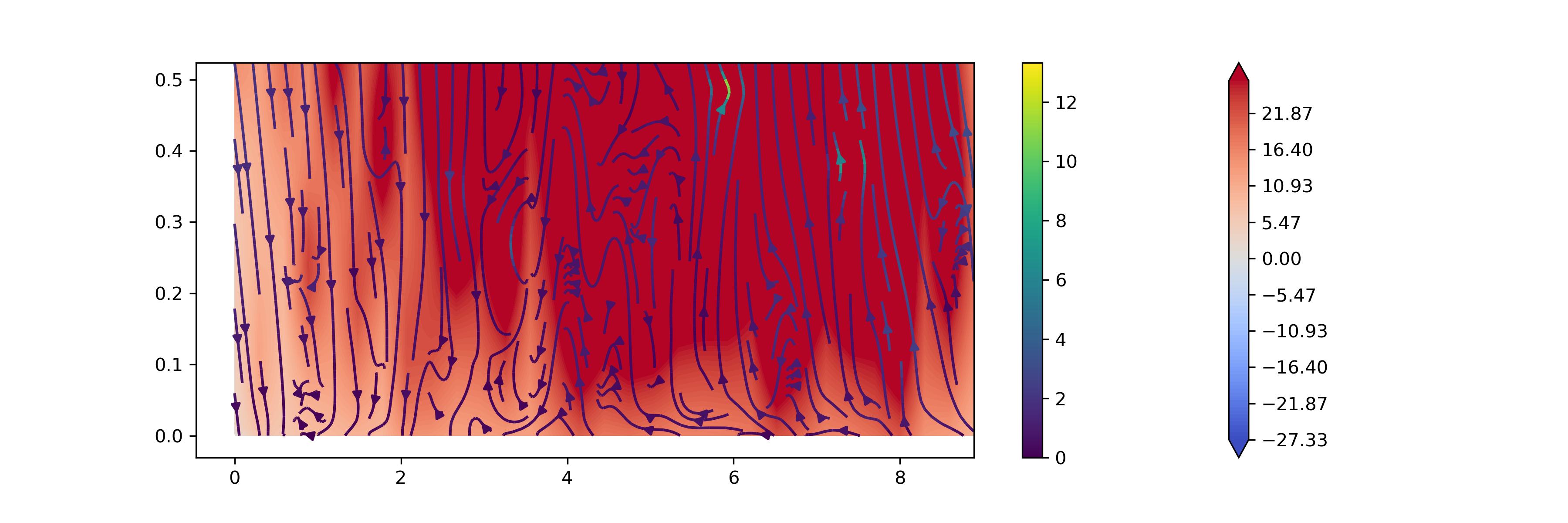}}}
    \qquad
    \subfloat[\centering $t=0.1$]{{\includegraphics[width=.5\linewidth]{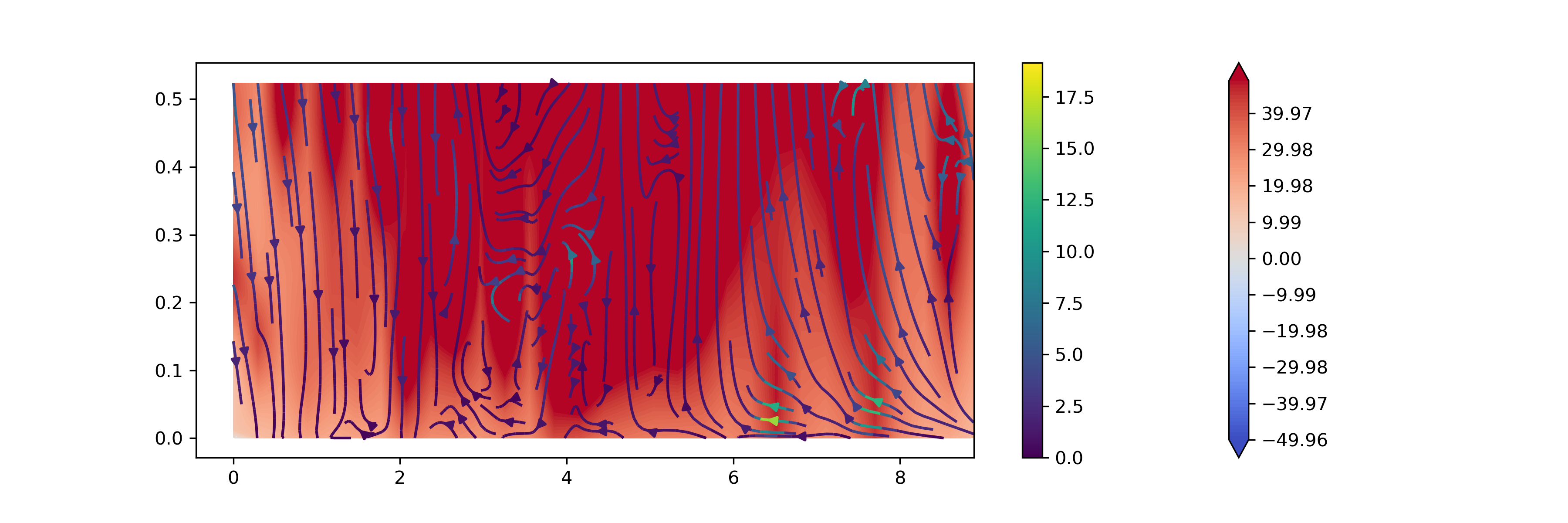} }}
    \subfloat[\centering $t=0.2$]{{\includegraphics[width=.5\linewidth]{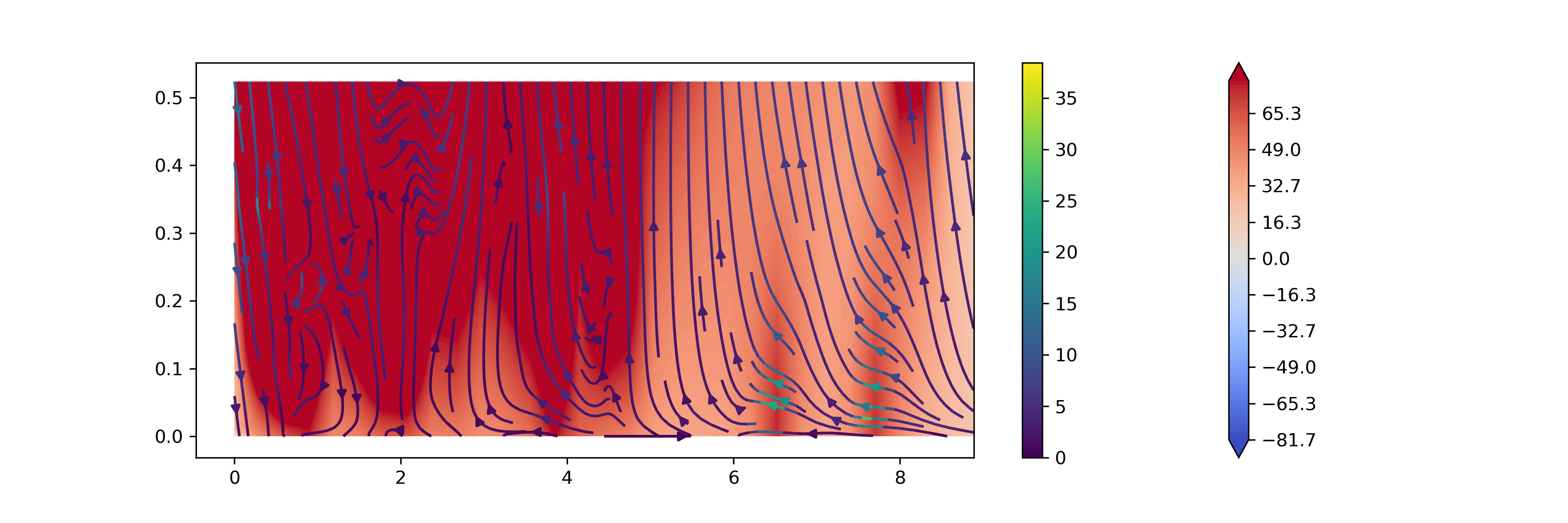}}}
    \caption{The upper boundary layer flow at different times $t$.}
    \label{Exp3FigUpperBFlow}
\end{figure}

\begin{figure}
    \centering
    \subfloat[\centering $t=0.001$]{{\includegraphics[width=.5\linewidth]{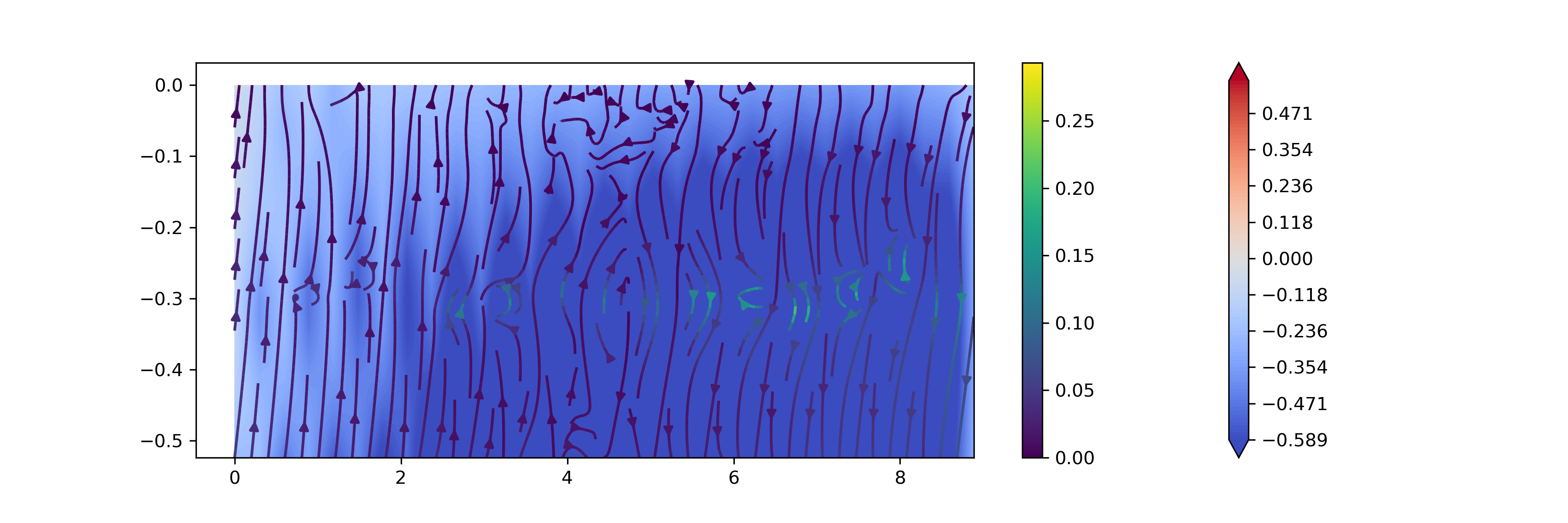} }}
    \subfloat[\centering $t=0.05$]{{\includegraphics[width=.5\linewidth]{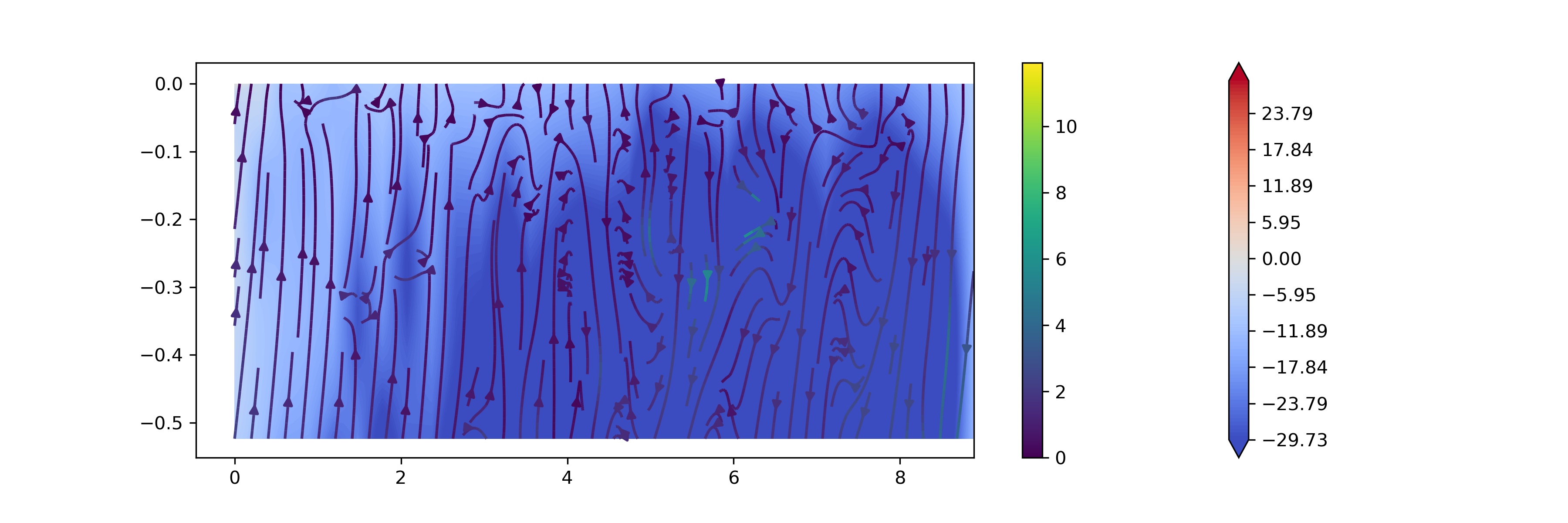}}}
    \qquad
    \subfloat[\centering $t=0.1$]{{\includegraphics[width=.5\linewidth]{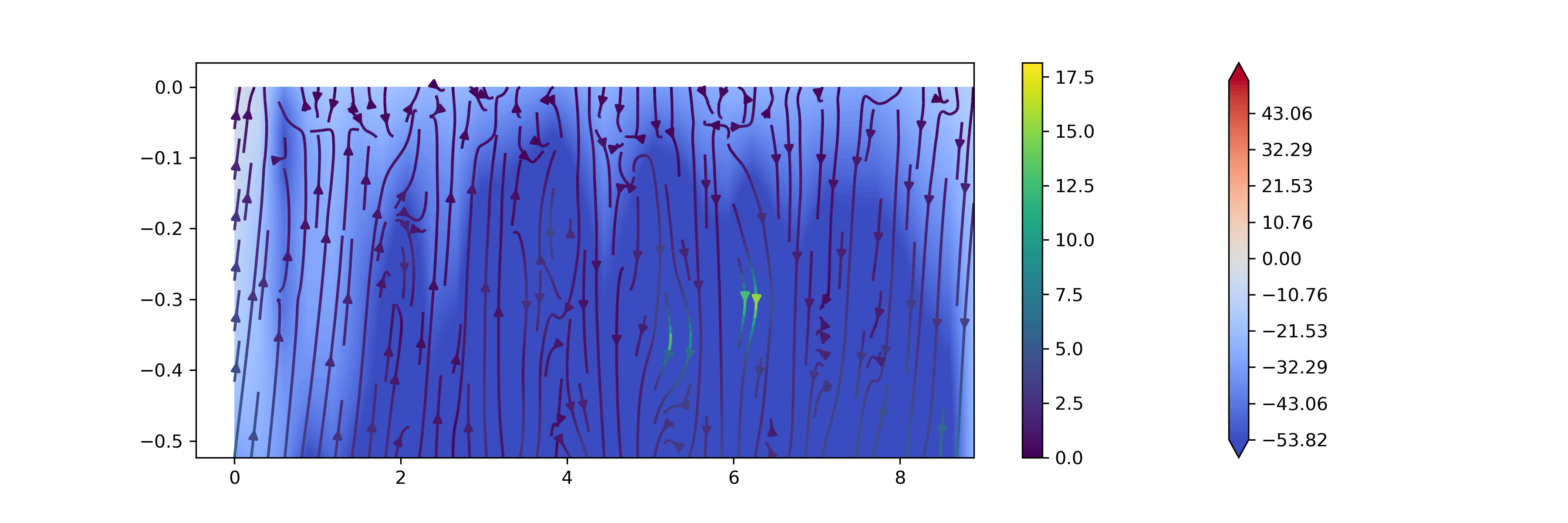} }}
    \subfloat[\centering $t=0.2$]{{\includegraphics[width=.5\linewidth]{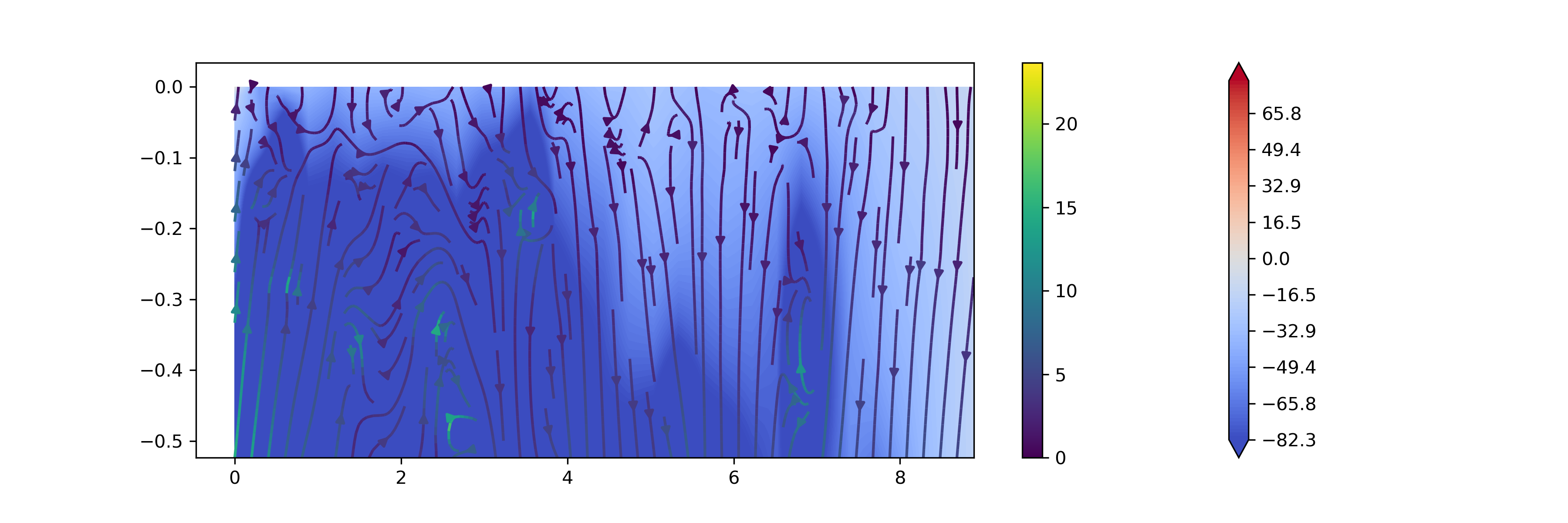}}}
    \caption{The lower boundary layer flow at different times $t$.}
    \label{Exp3FigLowerBFlow}
\end{figure}

\begin{figure}
    \centering
    \subfloat[\centering $t=0.001$]{{\includegraphics[width=.4\linewidth]{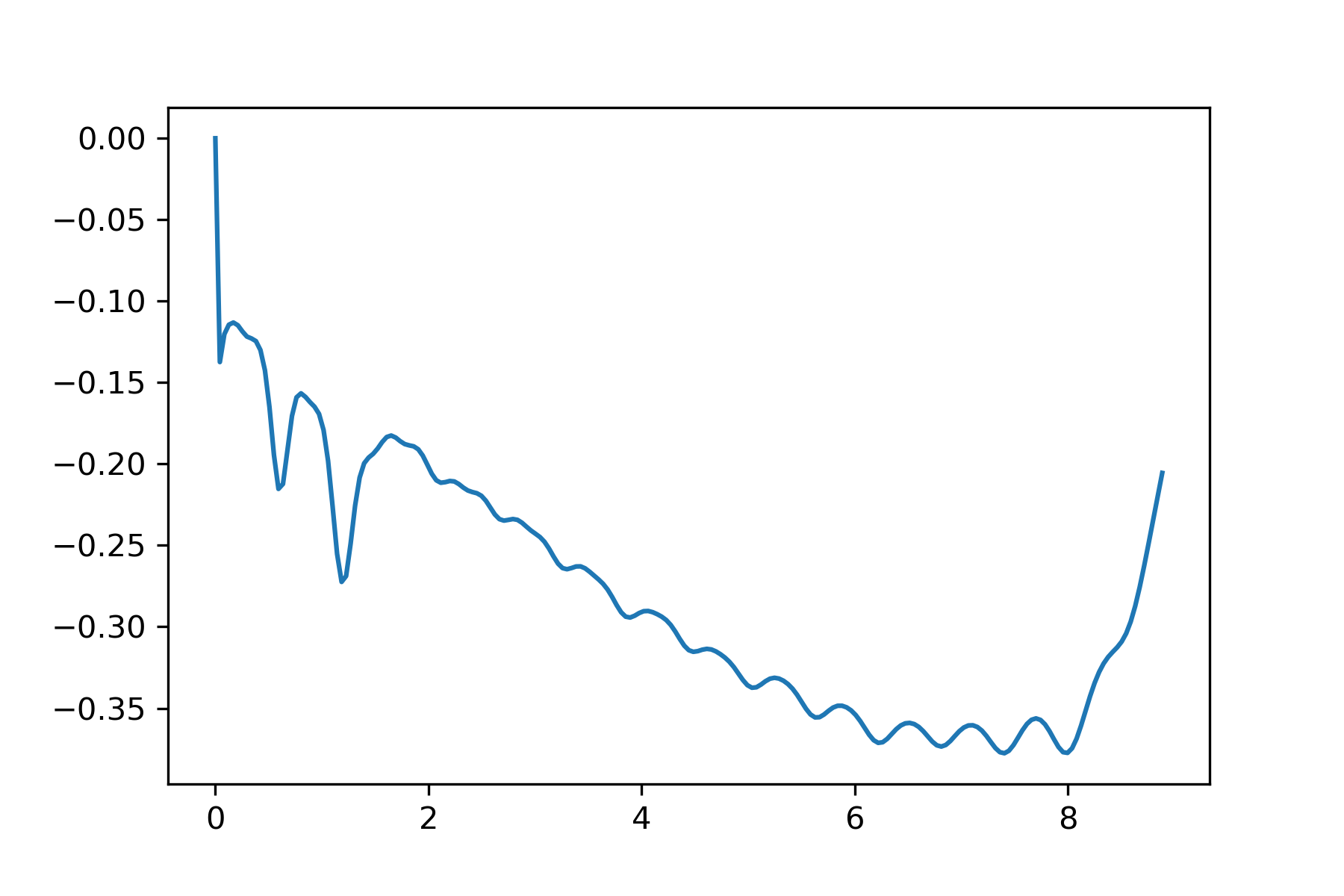} }}
    \subfloat[\centering $t=0.05$]{{\includegraphics[width=.4\linewidth]{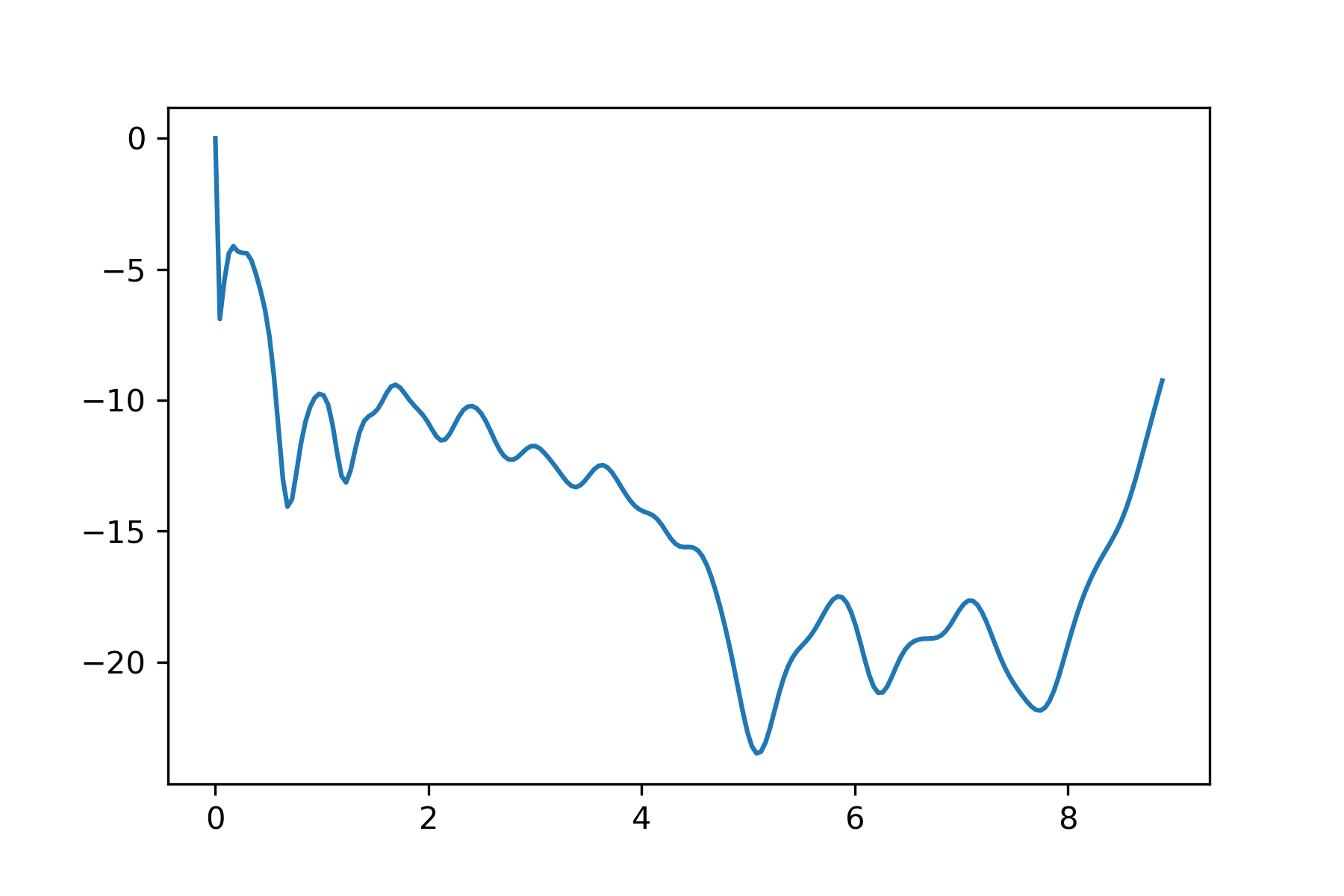}}}
    \qquad
    \subfloat[\centering $t=0.1$]{{\includegraphics[width=.4\linewidth]{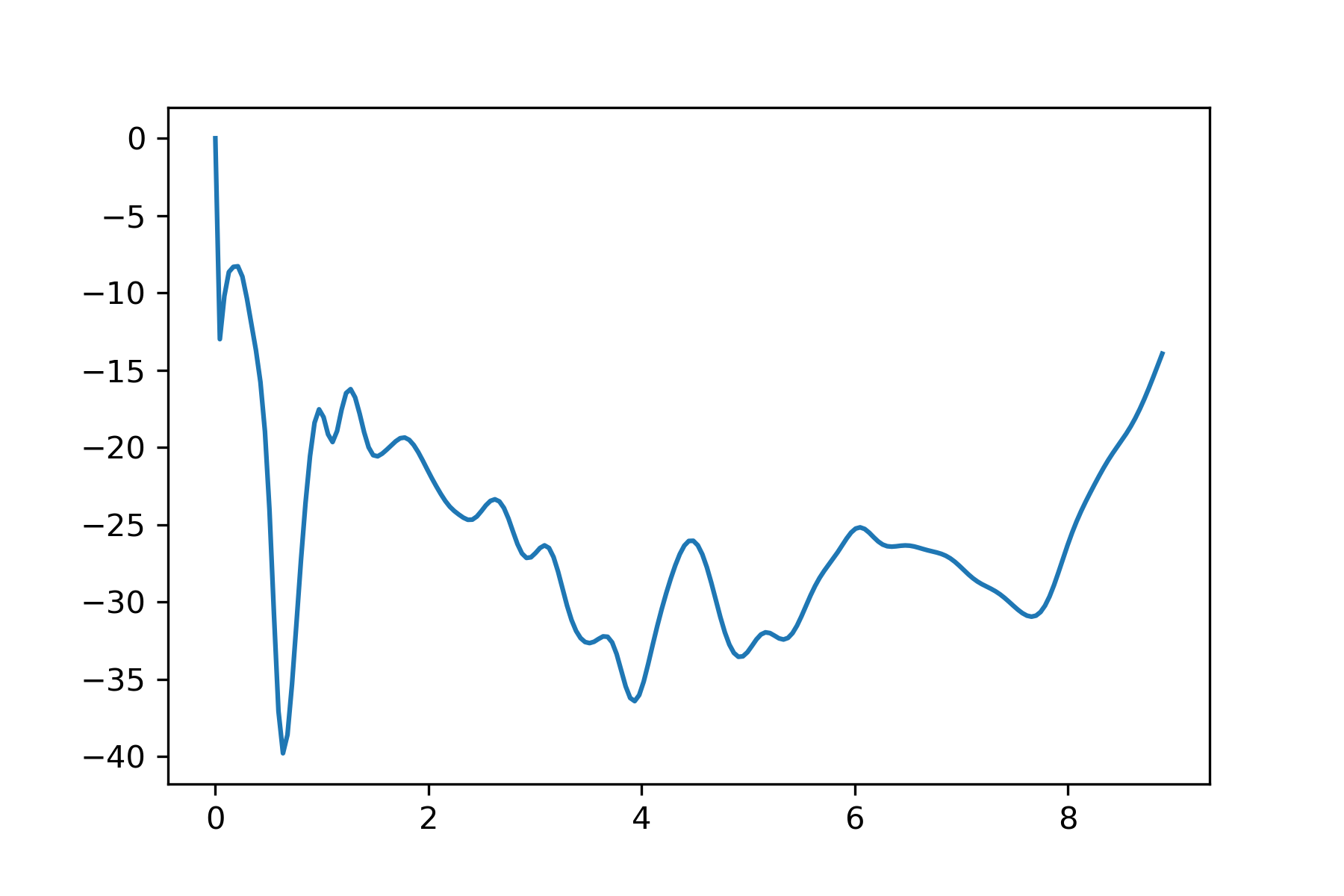} }}
    \subfloat[\centering $t=0.2$]{{\includegraphics[width=.4\linewidth]{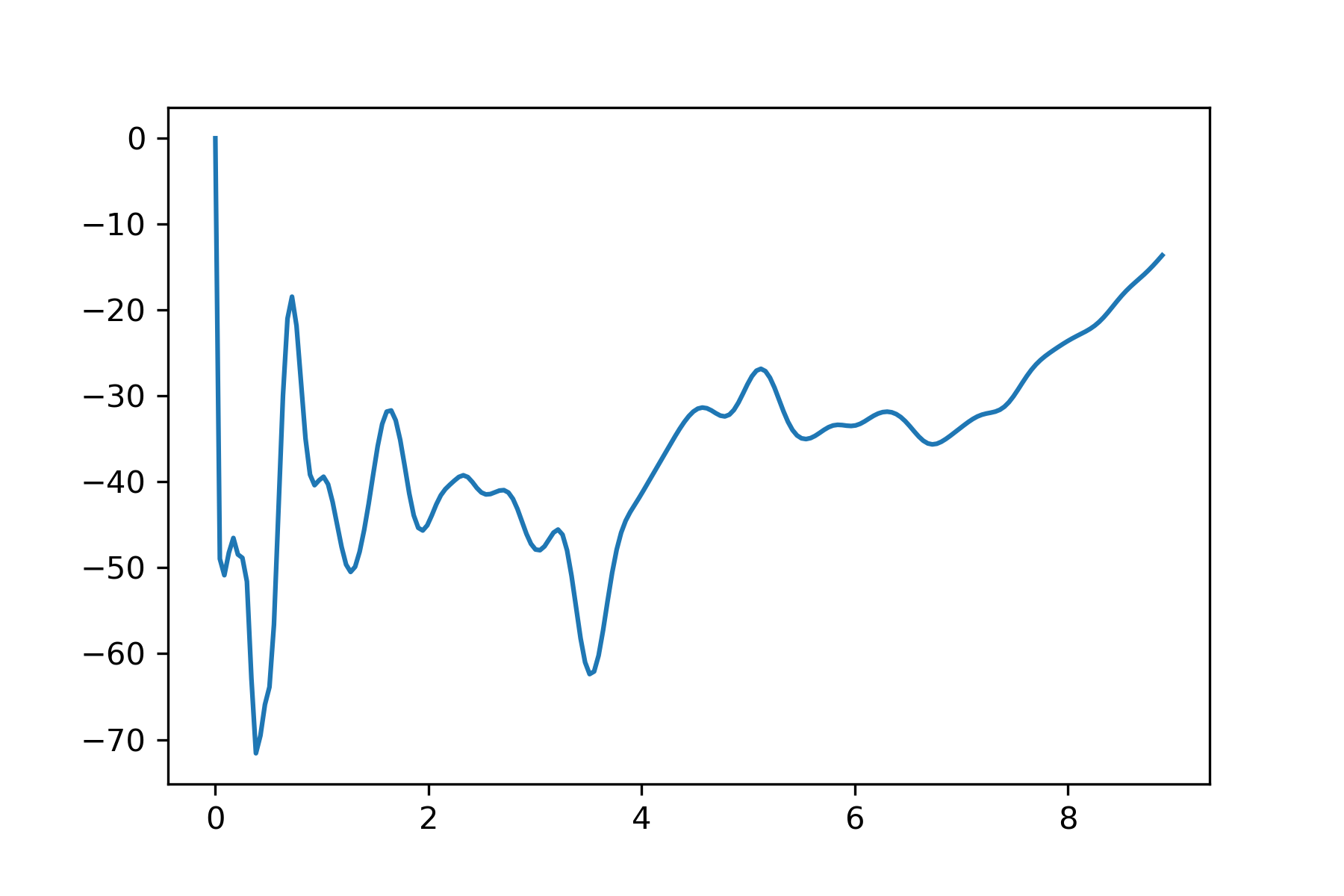}}}
    \caption{The stress applied to the lower boundary at different times $t$.}
    \label{Exp3FigLowerBStress}
\end{figure}

\begin{figure}
    \centering
    \subfloat[\centering $t=0.001$]{{\includegraphics[width=.4\linewidth]{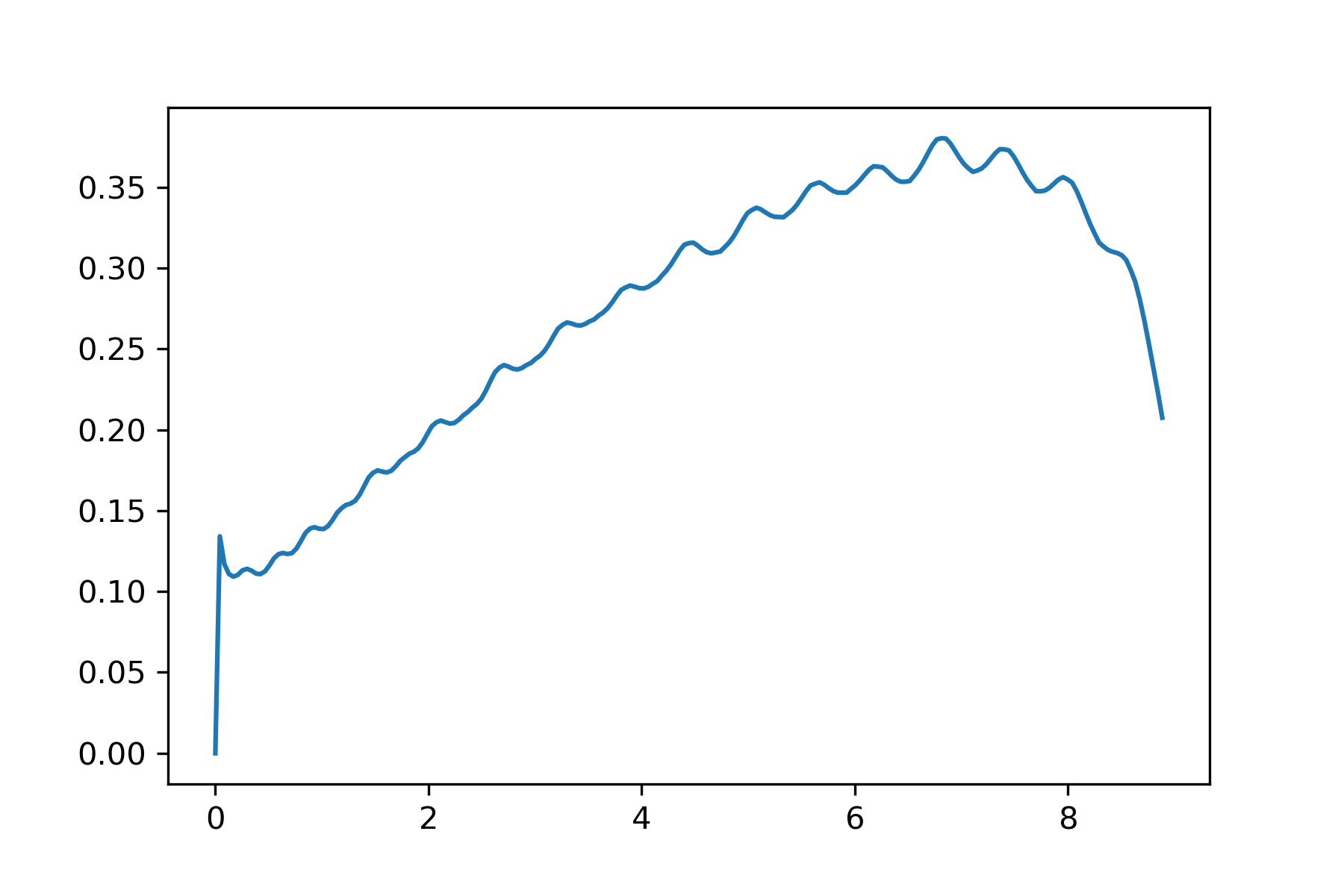} }}
    \subfloat[\centering $t=0.05$]{{\includegraphics[width=.4\linewidth]{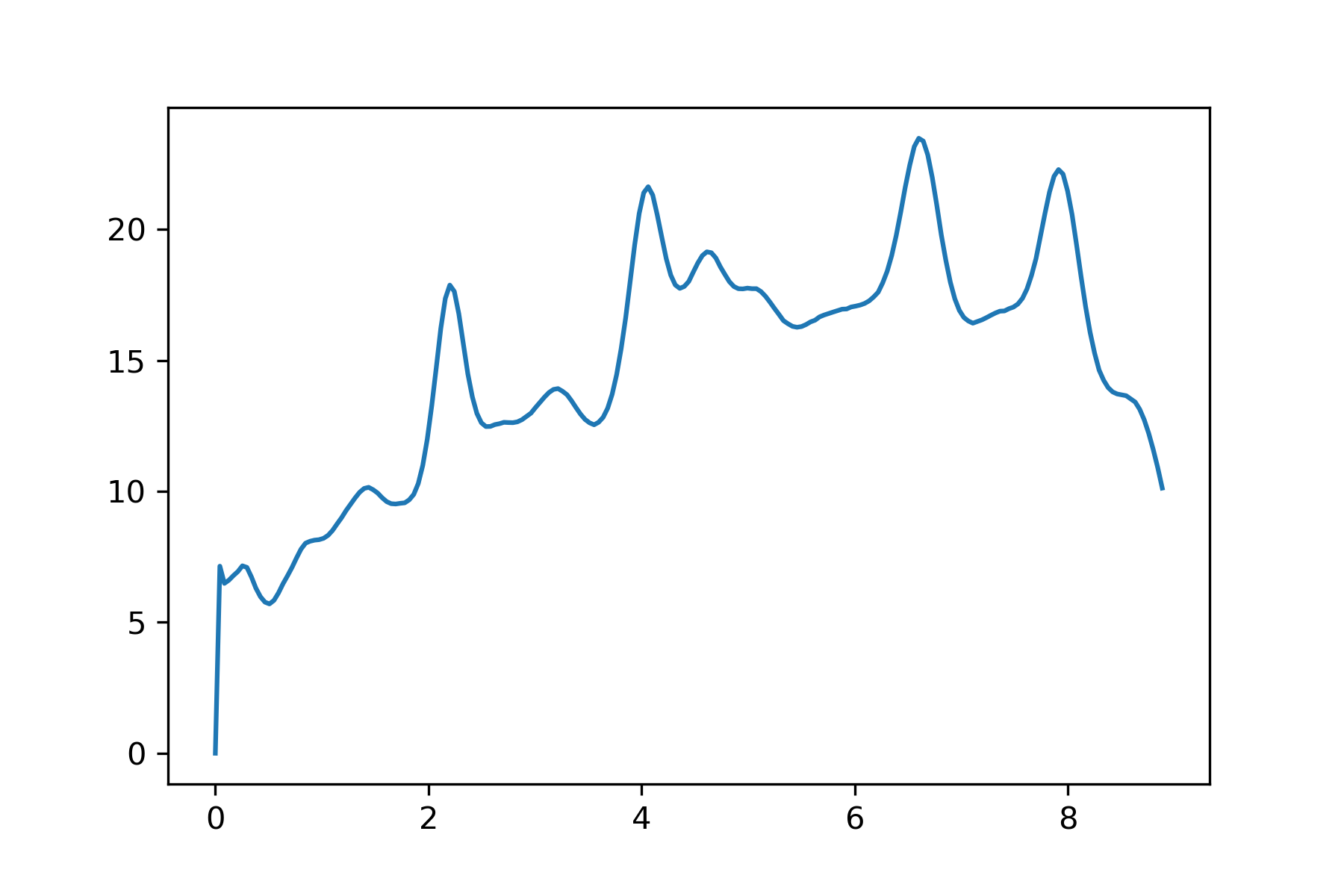}}}
    \qquad
    \subfloat[\centering $t=0.1$]{{\includegraphics[width=.4\linewidth]{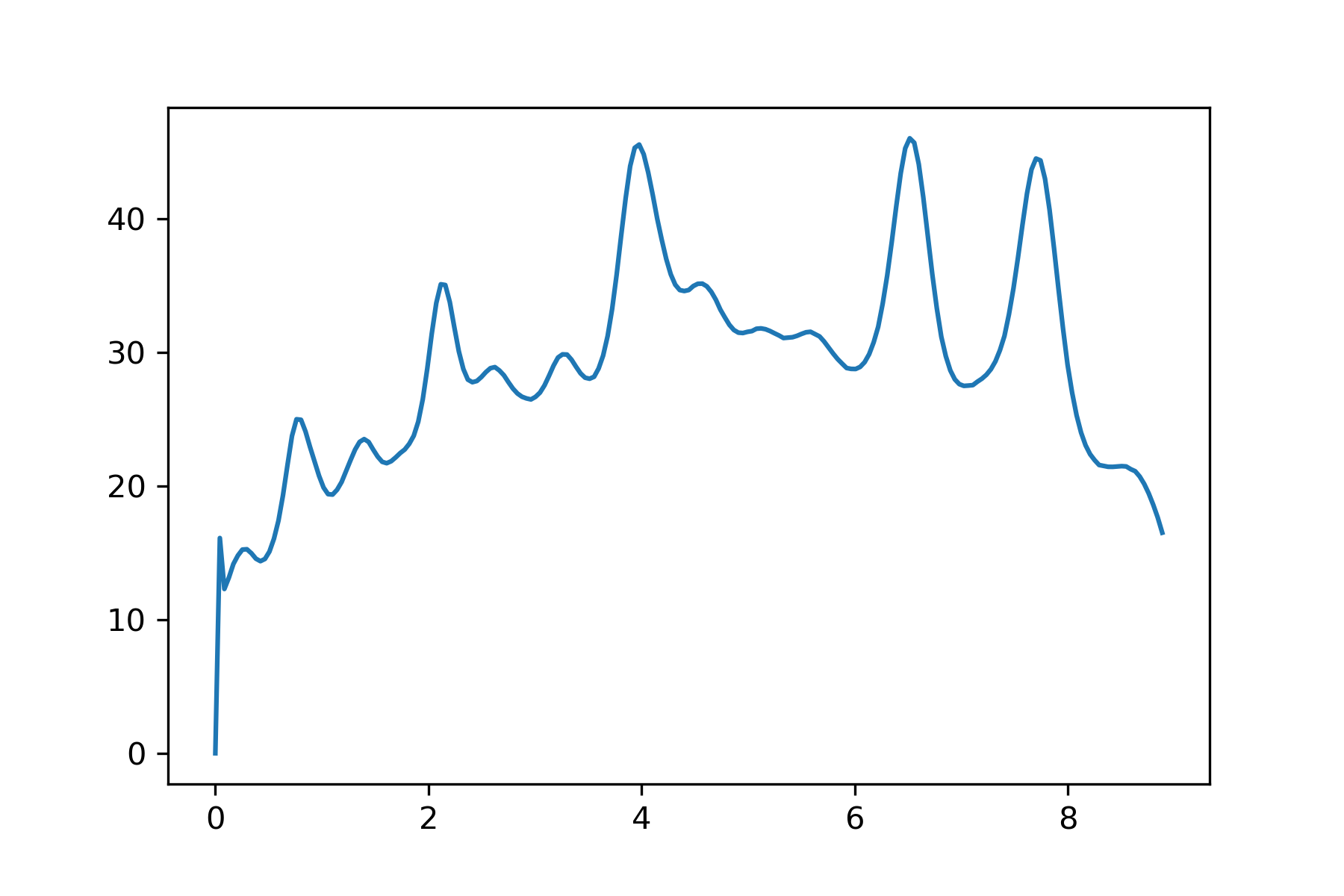} }}
    \subfloat[\centering $t=0.2$]{{\includegraphics[width=.4\linewidth]{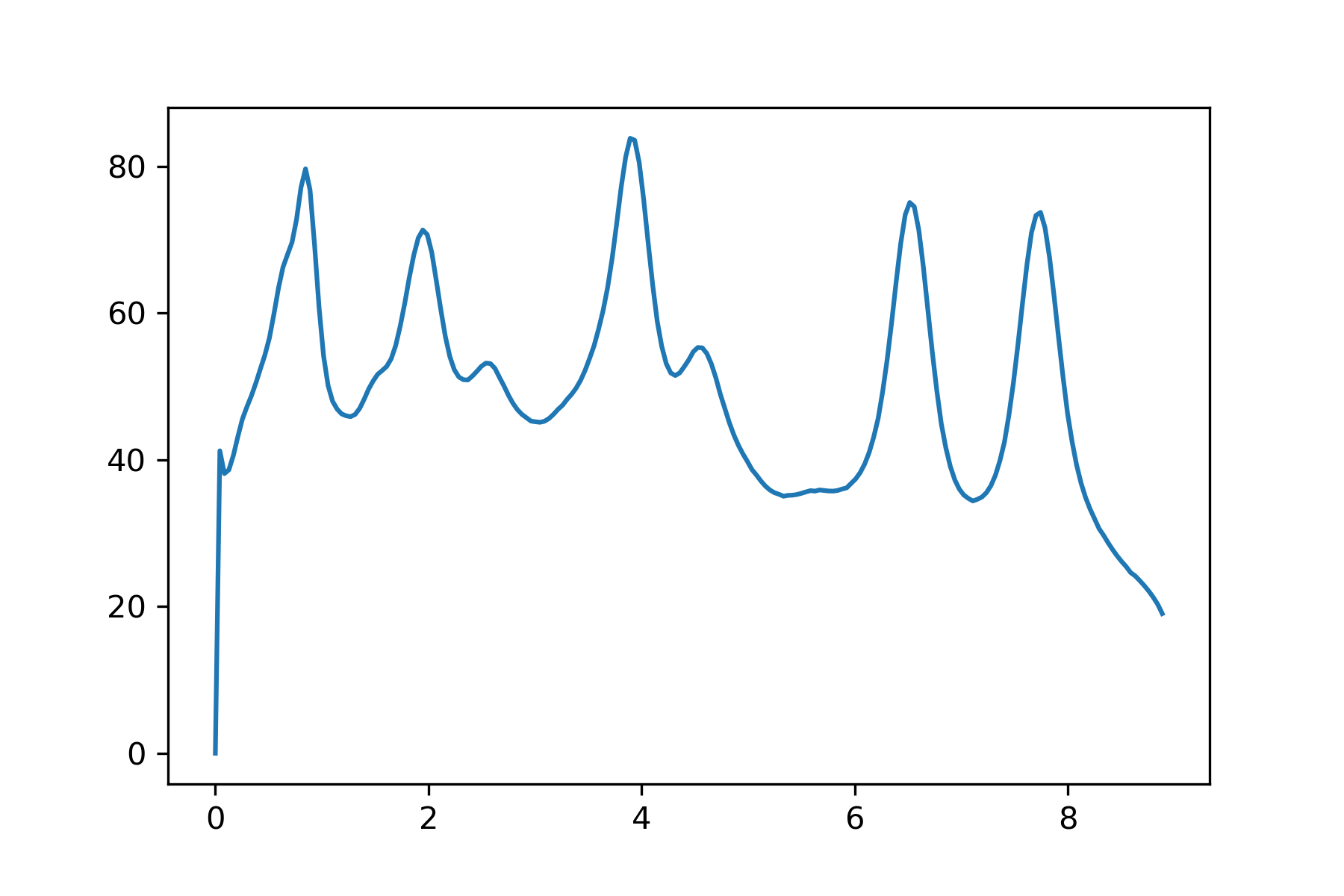}}}
    \caption{The stress applied to the upper boundary at different times $t$.}
    \label{Exp3FigUpperBStress}
\end{figure}


\newpage

\begin{thebibliography}{99}
\bibitem{AndersonGreengard1985} Anderson, C. and Greengard, C. 1985.
On vortex methods\emph{. SIAM J. Numer. Anal}. $\mathbf{22}$ (3),
413-440.

\bibitem{BalakumarAdrian2007} Balakumar, B.J. and Adrian R.J. 2007
Large- and very-large-scale motions in channel and boundary-layer
flows. \emph{Phil. Trans.} $\mathbf{365}$, 665--681.

\bibitem{ChauhanPhilipetl2014} Chauhan, K. Philip, J., De Silva,
C.C.M., Hutchins, N. and Marusic, I. 2014 The turbulent/non-turbulent
interface and entrainment in a boundary layer. \emph{J. Fluid Mech}.
742, 119--151. 

\bibitem{Chorin 1973} Chorin, A. J. 1973 Numerical study of slightly
viscous flow. \emph{J. Fluid Mech.} $\mathbf{57}$, 785-796.

\bibitem{Constantin2001a} Constantin, P. 2001 An Eulerian-Lagrangian approach for incompressible fluids: local theory.
\emph{J. Amer. Math. Soc.} $\mathbf{14}$ no. 2, 263-278 (electronic).

\bibitem{Constantin2001b} Constantin, P. 2001  An Eulerian-Lagrangian approach to the Navier-Stokes equations. \emph{Comm. Math. Phys.} $\mathbf{216}$, no. 3, 663-686.


\bibitem{ConstantinIyer2011} Constantin, P.  and Iyer, G. 2011 A stochastic-Lagrangian approach to the Navier-Stokes equations in domains with boundary, 
\emph{Ann. Appl. Probab}. 21, 1466-1492 (2011).


\bibitem{Drivas2017a} Drivas, T.D. and Eyink, G.L. 2017 A Lagrangian fluctuation-dissipation relation for scalar turbulence. Part I. Flows with no boundary walls.
\emph{Journal of Fluid Mechanics}, Volume $\mathbf{829}$, 25 October 2017 , pp. 153 - 189
DOI: https://doi.org/10.1017/jfm.2017.567

\bibitem{Drivas2017b} Drivas, T. D. and Eyink, G.L. 2017 A Lagrangian fluctuation-dissipation relation for scalar turbulence. Part II. Wall-bounded flows.
\emph{Journal of Fluid Mechanics}. 829, 236-279 (2017).

\bibitem{EyinkGuptaZaki2020a} Eyink, G., Gupta, A., and Zaki, T. 2020 Stochastic Lagrangian dynamics of vorticity. Part 1. General theory for viscous, incompressible fluids. 
\emph{Journal of Fluid Mechanics}, 901, A2. doi:10.1017/jfm.2020.491

\bibitem{EyinkGuptaZaki2020b} Eyink, G., Gupta, A., and Zaki, T. 2020 Stochastic Lagrangian dynamics of vorticity. Part 2. Application to near-wall channel-flow turbulence. 
\emph{Journal of Fluid Mechanics}, 901, A3. doi:10.1017/jfm.2020.492

\bibitem{Falkovich2001} Falkovich, G., Gaw{\k{e}}dzki, K. and Vergassola, M. 2001 Particles and fields in fluid turbulence, 
\emph{Rev. Mod. Phys.} $\mathbf{73}$ 913-975.

\bibitem{Feynman1948} Feynman, R. P. 1948 Space-time approach to non-relativistic quantum mechanics. \emph{Rev. Mod. Phys.} Vol. $\mathbf{20}$, No. 2, 367-387.

\bibitem{DawsonMcKeon2019} Dawson, S.T.M. and McKeon, B.J. 2019 On
the shape of resolvent modes in wall-bounded turbulence. \emph{J.
Fluid Mech}. $\mathbf{877}$, 682--716.

\bibitem{Deardorff1970} Deardorff, J. W. 1970 A numerical study of three-dimensional turbulent channel flow at large Reynolds numbers. \emph{J. Fluid Mech.}, $\mathbf{41}$:453-80

\bibitem{DeGraaffEaton2000} DeGraaff, D. B. \& Eaton, J. K. 2000
Reynolds-number scaling of the flat-plate turbulent boundary layer.
\emph{J. Fluid Mech}. $\mathbf{422}$, 319--346.

\bibitem{ErmJoubert1991} Erm, L. P. \& Joubert, P. N. 1991 Low Reynolds
number turbulent boundary layers. \emph{J. Fluid Mech}. $\mathbf{230}$,
1--44.

\bibitem{ErmSmitsJoubert1985} Erm, L. P., Smits, A. J. \& Joubert,
P. N. 1985 Low Reynolds number turbulent boundary layers on a smooth
flat surface in a zero pressure gradient. In \emph{Proceedings of
fifth Symposium on Turbulent Shear Flows}, Ithaca, NY. 

\bibitem{Fletcher1991} Fletcher, C. A. J. 1991 \emph{Computational techniques for fluid dynamics}, 
Vol. I and II, second edition. Springer-Verlag. 

\bibitem{Friedman1964} Friedman, A. 1964 \emph{Partial differential
equations of parabolic type}. Prentice-Hall, Inc.

\bibitem{Goodman1987} Goodman, J. 1987 Convergence of the random
vortex method. \emph{Comm. Pure Appl. Math.} $\mathbf{40}$(2), 189-220.

\bibitem{Heisel2018} Heisel, M., Dasari, T., Liu, Y., Hong, J., Coletti,
F. \& Guala, M. 2018 The spatial structure of the logarithmic region
in very-high-Reynolds-number rough wall turbulent boundary layers.
\emph{J. Fluid Mech}. $\boldsymbol{857}$, 704--747. 

\bibitem{HeadBandyopadhyay1981} Head, M. R. \& Bandyopadhyay, P.
1981 New aspects of turbulent boundary layer structure. \emph{J. Fluid
Mech}. $\mathbf{107}$, 297--338.

\bibitem{HonkanAndreopoulos1997} Honkan, A. and Andreopoulos, Y.
1997 Vorticity, strain-rate and dissipation characteristics in the
near-wall region of turbulent boundary layers. \emph{J. Fluid Mech}.
$\mathbf{350}$, 29--96.

\bibitem{Kac1949} Kac, M. 1949 On Distributions of Certain Wiener Functionals. 
\emph{Transactions of the American Mathematical Society}, Jan., 1949, Vol. $\mathbf{65}$, No. 1 (Jan., 1949), pp. 1-13

\bibitem{Keller1978} Keller, H. B. 1978 Numerical methods in boundary-layer
theory. \emph{Ann. Rev. Fluid Mech}. $\mathbf{10}$, 417-33.

\bibitem{LesieurLDS} Lesieur, M., M{\'e}tais, O. and Comte, P. 2005 
\emph{Large-Eddy simulations of turbulence}. Cambridge University Press.

\bibitem{LQX2023} Li, J., Qian, Z., Xu, M. 2023 Twin Brownian particle method for the study of Oberbeck-Boussinesq fluid flows. https://doi.org/10.48550/arXiv.2303.17260

\bibitem{Lilly1967} Lilly, D. K. 1967 The representation of small-scale
turbulence in numerical simulation experiments. In H. H. Goldstine
(Ed.), \emph{Proc. IBM Scientific Computing Symp. on Environmental
Sciences}, pp. 195-210. Yorktown Heights, NY: IBM. 

\bibitem{Long1988} Long, D. G. 1988 Convergence of the random vortex
method in two dimensions. \emph{J. of Amer. Math. Soc}. $\mathbf{1}$(4
), 779-804.

\bibitem{Majda and Bertozzi 2002} Majda, A. J. and Bertozzi A. L.
2002 \emph{Vorticity and incompressible flow}. Cambridge University
Press.

\bibitem{MoinMashesh1998} Moin, P. and Mahesh, K. 1998 Direct numerical
simulation: a tool in turbulence research. \emph{Annu. Rev. Fluid
Mech.} $\mathbf{30}$, 539--78

\bibitem{OrszagPatterson1972} Orszag S.A. and Patterson G.S. 1972
Numerical simulation of three-dimensional homogeneous isotropic turbulence.
\emph{Phys. Rev. Lett}. $\mathbf{28}$, 76--79.

\bibitem{Pope2000} Pope, S. B. 2000 \emph{Turbulent flows}. Cambridge
University Press. 

\bibitem{Prandtl1904} Prandtl, L. 1904 \"Uber Fl\"ussigkeitsbewegung
bei sehr kleiner Reibung. \emph{Proc. Third Intern. Math. Congress},
Heidelberg, 848-491. 

\bibitem{Qian-Stochastic2022} Qian, Z. 2022 Stochastic formulation of incompressible fluid flows in wall bounded regions. https://doi.org/10.48550/arXiv.2206.05198

\bibitem{QSZ3D} Qian, Z.,  S\"uli, E. and Zhang, Y. 2022 Random vortex dynamics via functional stochastic differential equations. \emph{Proc. R. Soc. A}
478: 20220030. https://doi.org/10.1098/rspa.2022.0030

\bibitem{QQZW2022} Qian, Z., Qiu, Y., Zhao, L. and Wu, J. 2022 Monte-Carlo simulations for wall-bounded fluid flows via random vortex method. https://doi.org/10.48550/arXiv.2208.13233

\bibitem{RaiMoin1993} Rai, M. M. \& Moin, P. 1993 Direct numerical
simulation of transition and turbulence in a spatially evolving boundary
layer.\emph{J. Comput. Phys}. $\mathbf{109}$, 169--192.

\bibitem{Schlichting9th-2017} Schlichting, H. and Gersten, K. 2017
\emph{Boundary-Layer Theory} (Ninth Edition). Springer.

\bibitem{Spalart1988} Spalart, P. R. 1988 Direct simulation of a
turbulent boundary layer up to $R_{\theta}=1410$. \emph{J. Fluid
Mech}. $\mathbf{187}$, 61--98.

\bibitem{SpalartWatmuff1993} Spalart, P. R. and Watmuff, J. H. 1993
Experimental and numerical study of a turbulent boundary layer with
pressure gradients. \emph{J. Fluid Mech}. $\mathbf{249}$, 337--371.

\bibitem{Taylor1921} Taylor, G. I. 1921 Diffusion by continuous movements.
\emph{Proc. Lond. Math. Soc}. $\mathbf{20}$, 196.

\bibitem{WeinanLiu1996} Weinan, E. and Liu, J.-G. 1996 Vorticity boundary
condition and related issues for finite difference schemes. \emph{J.
of Comp. Phys}. $\mathbf{124}$, 368--382.

\bibitem{Wesseling2001}  Wesseling, P. 2001 \emph{Principles of computational
fluid dynamics.} Springer-Verlag Berlin Heidelberg.

\bibitem{WuJacobsHuntDurbin1999} Wu, X., Jacobs, R., Hunt, J. C. R.
and Durbin, P. A. 1999 Simulation of boundary layer transition induced
by periodically passing wakes. \emph{J. Fluid Mech}. $\mathbf{398}$,
109--153.

\bibitem{WuMoin2008} Wu, X. and Moin, P. 2008 A direct numerical simulation
study on the mean velocity characteristics in turbulent pipe flow.
\emph{J. Fluid Mech}. $\mathbf{608}$, 81--112.

\bibitem{WuMoin2009} Wu, X. and Moin, P. 2009 Direct numerical simulation
of turbulence in a nominally zero-pressure-gradient flat-plate boundary
layer. \emph{J. Fluid Mech}. $\mathbf{630}$, pp. 5--41.

\bibitem{WuMoinHickey2014} Wu, X., Moin, P. and Hickey, J. P. 2014
Boundary layer bypass transition. \emph{Physics of Fluids}, $\mathbf{26}$,
091104.
\end{thebibliography}
\end{document}